\RequirePackage{fix-cm}
\documentclass[smallextended]{svjour3}       
\smartqed  
\usepackage{graphicx,subfigure}
\usepackage{amsfonts,amsmath,mathtools,bbm,cool}
\usepackage{multirow}
\usepackage{xcolor}

\newcommand{\be}{\begin{equation}}
\newcommand{\ee}{\end{equation}}

\newcommand{\LL}{\Lambda}

\renewcommand{\aa}{\mathbf{a}}
\newcommand{\as}{{a}}
\newcommand{\z}{\mathbf{z}}
\newcommand{\fI}{i}
\newcommand{\fId}{i_I}
\newcommand{\fR}{r}
\newcommand{\fS}{s}
\newcommand{\fE}{e}
\newcommand{\fP}{i_P}
\newcommand{\fA}{i_A}
\newcommand{\h}{h}
\newcommand{\pI}{p_i}
\newcommand{\pR}{p_r}
\newcommand{\pS}{p_s}

\begin{document}

\title{Control with uncertain data of socially structured compartmental epidemic models
}

\titlerunning{Control with uncertain data of epidemic models}        

\author{Giacomo Albi         \and
        Lorenzo Pareschi \and
        Mattia Zanella
}


\institute{G. Albi \at
              Department of Computer Science \\
              University of Verona, Italy \\
              \email{giacomo.albi@univr.it}           
           \and
           L. Pareschi \at
              Department of Mathematics and Computer Science \\
              University of Ferrara, Italy \\
               \email{ lorenzo.pareschi@unife.it}
               \and
               M. Zanella \at
               Department of Mathematics "F. Casorati" \\
               University of Pavia, Italy \\
               \email{mattia.zanella@unipv.it}         
}

\date{Received: date / Accepted: date}

\maketitle

\begin{abstract}
The adoption of containment measures to reduce the amplitude of the epidemic peak is a key aspect in tackling the rapid spread of an epidemic. Classical compartmental models must be modified and studied to correctly describe the effects of forced external actions to reduce the impact of the disease. {The importance of social structure, such as the age dependence that proved essential in the recent COVID-19 pandemic, must be considered, and in addition, the available data are often incomplete and heterogeneous, so a high degree of uncertainty must be incorporated into the model from the beginning. In this work we address these aspects, through an  optimal control formulation of a socially structured epidemic model in presence of uncertain data.} After the introduction of the optimal control problem, we formulate an instantaneous approximation of the control that allows us to derive new feedback controlled compartmental models capable of describing the epidemic peak reduction. 
The need for long-term interventions shows that alternative actions based on the social structure of the system can be as effective as the more expensive global strategy. The timing and intensity of interventions, however, is particularly relevant in the case of uncertain parameters on the actual number of infected people. Simulations related to data {from the first wave } of the recent COVID-19 outbreak in Italy are presented and discussed.
\end{abstract}

\section{Introduction}
From the digital tracking systems of the Koreans to UK's initial choice of not wanting to do anything to counter the spread of the virus, passing through the militarized quarantines of the Chinese and those less authoritarian and more involved of the Italians, the reaction of different countries to the COVID-19 outbreak has shown a series of very different approaches that can also be explained considering the different cultural and political attitudes of the countries concerned.

In all cases, after an initial phase, even those governments that were less restrictive in the face of the pandemic's inexorable progress had to take strong containment measures.
There's a graph that has become the symbol of the COVID-19 pandemic most of all. It shows in a simple and intuitive way the importance of slowing down the spread of an epidemic as much as possible ("flattening the curve"), so that the healthcare system can take care of all the sick without collapsing. Its success has helped to save many lives, raising awareness of good practices to slow down an epidemic: stay as much as possible at home, reduce social interactions and wash your hands often and well.

These ''non-pharmaceutical'' intervention measures, however, entail significant social and economic costs and thus policy makers may not be able to maintain them for more than a short period of time. Therefore, a modeling approach based on a limited time horizon that takes into account the social structure of the population is necessary in order to optimize containment strategies. Most current research, however, has focused on control procedures aimed at optimizing the use of vaccinations and medical treatments \cite{BGO,BBSG,SCCC10,CG} and only recently the problem has been tackled from the perspective of non-pharmaceutical interventions \cite{LML10,MRPL}. In addition, data collected by governments are often incomplete and heterogeneous, so a high degree of uncertainty needs to be incorporated into predictive models \cite{Capaldi_etal,CCVH,Chowell,JRGL,MKZC,Roberts}. This is the case of the spreading of COVID-19 worldwide, which have been often mistakenly underestimated due to a combination of factors, including deficiencies in surveillance and diagnostic capacity, and the large number of infectious but asymptomatic individuals \cite{JRGL,MKZC,Zhang_etal}.

For almost a hundred years, mathematical models have been used to describe the spread of epidemics \cite{KMK}. The models currently used largely originate from the model proposed by Kermack and McKendrick at the beginning of last century. Even if the model contains  strong simplification assumptions, the concepts introduced through this model are essential to provide a first intuition on the dynamics of epidemics, an intuition that remains confirmed in more complex models, albeit with numerous modifications (see for example \cite{H00,CS78}). The model provides for the division of the population into compartments, the susceptible,  healthy individuals who may be infected, the infectious, who have already contracted the disease and can transmit it, and the removed, compartment that includes those who are healed and immune.

The hypothesis made by Kermack and McKendrick is that of the homogeneous "mixing"; that is, it is assumed that each individual has the same probability of contacting any other individual in the population. One understands how this hypothesis is unrealistic: we are often in contact with people from our family, our workplace, school class, group of friends and very rarely with those who live in a different place, have different ages and professions. In recent years, therefore, computational models have been developed that try to take into account additional social characteristics of individuals in order to arrive at more accurate predictions by keeping, however, the simplicity of compartmental models \cite{CHALL89,GFMDC12,LGC12,H96,IMP,FP}.

In this paper starting from a general compartmental model with social structure, {typically the age dependence}, we consider the external action of a policy maker that aims at reducing the spread of the epidemics by applying non pharmaceutical intervention measures, such as social distancing and quarantine. The mathematical problem is formulated as an optimal control problem characterized by a functional cost whose objective is to minimize the number of infectious people in a given time horizon. Through an instantaneous control strategy we compute an explicit feedback control that allows us to derive new SIR-type compartmental models capable of describing the epidemic peak reduction. Previously, this type of approach has been used successfully in the case of social models of consensus \cite{Albi1,Albi2,Albi3,Albi4,BFK,CFPT,DPT}.  

The feedback controlled models are subsequently extended to take into account the presence of uncertain infection parameters and data. In fact, to have reliable forecasts it is of paramount importance to consider the presence of uncertain quantities as a structural feature of the epidemic dynamics. {This aspect is of paramount importance in the case of pandemic COVID-19, in which undetected infectious individuals play a key role in the spread of the disease. In this regard, it is worth noting that our methodology can be easily extended to the case of more complex compartmental models. The decision to limit ourselves to a simple SIR-type compartmentalization was related on the one hand to the increased complexity given by the dependence on the social structure, which proved to be crucial in the case of the COVID-19 pandemic, and on the other hand to the introduction of a systematic uncertainty in the number of infected to avoid a complex sub-compartmentalization of the infectious population and the consequent difficulties due to parameter identification and the inability to follow a data-driven approach \cite{Giordano,IC}.} 
From a mathematical point of view, we can rely on the methods of uncertainty quantification (UQ) to obtain efficient and accurate solutions based on stochastic orthogonal polynomials for the differential model with random inputs \cite{X}. 

Few results are actually available regarding methods of UQ in epidemic systems, we mention in this direction \cite{CCVH,Chowell,Capaldi_etal,Roberts}. The main idea is to increase the dimensionality of the problem adding the possible sources of uncertainty from the very beginning of the modeling. Hence, we extrapolate statistics by looking at the so-called quantities of interest, i.e. statistical quantities that can be obtained from the solution and that give some global information with respect to the input parameter like expected solution of the problem or higher order moments. Several techniques can be adopted for the approximation of the quantities of interest, in this paper we adopt stochastic Galerkin methods that allow to reduce the problem to a set of deterministic equations for the numerical evaluation of the solution in presence of uncertainties. 
{Compared to conventional Monte Carlo methods, based on stochastic sampling, these methods guarantee an exponential convergence in the case of smooth uncertainty distributions and allow a much more accurate and efficient estimation of random parameters.} We refer the interested reader to recent surveys and monographs on the topic \cite{DPZ,JP,PareschiUQ,X}. 

In particular, we consider the case in which the policy maker applies his control based on several possible estimators on the actual number of infected people. The need for long-term interventions shows that alternative actions based on the social structure of the system can be as effective as the more expensive optimal strategy. The importance of the timing and intensity of interventions is particularly relevant in the case of uncertain parameters on the actual number of infected people.

The rest of the manuscript is organized as follows. In Section 2 we introduce the structured social SIR model and formulate the mathematical approach for containment measures to reduce the spread of the disease. Next, a feedback controlled model used in the subsequent analysis is derived within a short time horizon approximation. In Section 3 we generalize the feedback controlled model to take into account the presence of uncertainties. Section 4 is dedicated to the presentation of some numerical examples including {applications to the first wave of the COVID-19 epidemic in Italy.} In separate Appendices we provide details {on the generalizations of the present approach to more realistic epidemic models for COVID-19 including additional compartmentalizations,} on the stochastic Galerkin method employed to efficiently address the uncertainties, and on the social interaction matrices characterizing the contact rates.

\section{Control of epidemic dynamics}
The starting model in our discussion is a SIR-type compartmental model with a social structure. The presence of a social structure is in fact essential in deriving appropriate sustainable control techniques from the population for a protracted period, as in the case of the recent COVID-19 epidemic. {We will discuss in Section 3 how to modify the model through the introduction of a stochastic parameter that takes into account the dependence on uncertain data, and thus implicitly introduce the role of undetected infectious in the dynamics.}
 
\subsection{Compartmental models with social structure}
The heterogeneity of the social structure, which impacts the diffusion of the infective disease, is characterized by the vector $\aa\in \LL \subseteq \mathbb{R}^{d_{\as}}$ characterizing its social state and whose components summarize, for example, the age of the individual, its number of social connections or its economic status \cite{H96,H00}.
We denote by $\fS(\aa,t)$, $\fI(\aa,t)$ and $\fR(\aa,t)$, the  distributions at time $t > 0$ of susceptible, infectious and recovered individuals, respectively in relation to specific social characteristics. We assume that the rapid spread of the disease and the low mortality rate allows to ignore changes in the social structure, such as the aging process, births and deaths.

Consequently, for a given population of total number $N$, we have that
\begin{align*}
&\fS(\aa,t)+\fI(\aa,t)+\fR(\aa,t) = f (\aa), \qquad \int_{\LL} f(\aa)d\aa  = N,
\end{align*}
where $f (\aa)$ is the total distribution of the social features defined by the vector $\aa$. Hence, we recover the total fraction of the population which belong to the susceptible, infected and recovered as follows 
\begin{equation}
S(t)=\int_{\LL}\fS (\aa,t)\,d\aa,\quad I(t)=\int_{\LL}\fI (\aa,t)\,d\aa,\quad R(t)=\int_{\LL}\fR (\aa,t)\,d\aa. 
\end{equation}

In a situation where changes in the social features act on a slower scale with respect to the spread of the disease, the socially structured compartmental model follows the dynamics \begin{equation}\label{eq:SIR}
\begin{split}
\frac{d}{dt} \fS(\aa,t)&= - \fS(\aa,t)\frac1{N}\int_{\LL} \beta(\aa,\aa_*){\fI(\aa_*,t)}\ d\aa_* \\
\frac{d}{dt} \fI(\aa,t) &=  \fS(\aa,t)\frac1{N}\int_{\LL} \beta(\aa,\aa_*){\fI(\aa_*,t)}\ d\aa_* - \gamma (\aa) \fI(\aa,t) \\
\frac{d}{dt} \fR(\aa,t) &= \gamma(\aa) \fI(\aa,t),
\end{split}
\end{equation}
where the function $\beta(\aa,\aa_*) \geq 0$ represents the uncertain interaction rate among individuals with different social features and $\gamma(\aa) \geq 0$ the recovery rate which may depend on the social feature. 

Often, in socially structured models the interaction rate between people is assumed to be separable, and proportionate to the activity level of the social feature \cite{H96,H00}, as follows
\begin{equation}
\beta (\aa,\aa_*) = \dfrac{b (\aa)b (\aa_*)}{\int_0^{+\infty} b (\aa)f (\aa)\ d \aa} 
\end{equation}
with $b (\aa)$ the average number of people contacted by a person with social feature $\aa$ per unit time. Alternative approaches are based on preferential mixing \cite{GFMDC12,CHALL89}. Specific examples of age-dependent social interaction matrices are reported in Appendix \ref{app:data}.

We introduce the usual normalization scaling 
\[
\frac{\fS(\aa,t)}{N}\to \fS(\aa,t),\quad\frac{\fI(\aa,t)}{N}\to \fI(\aa,t),\quad\frac{\fR(\aa,t)}{N}\to \fR(\aa,t),\quad 
\int_\Lambda f(\aa) d\aa= 1, 
\]
and observe that the quantities $S(t)$, $I(t)$ and $R(t)$ satisfy the conventional SIR dynamics
\begin{equation}\label{eq:SIRm}
\begin{split}
\frac{d}{dt} S(t)&= - \int_\Lambda\fS(\aa,t)\int_{\LL} \beta(\aa,\aa_*){\fI(\aa_*,t)}\ d\aa_*\,d\aa \\
\frac{d}{dt} I(t) &=  \int_\Lambda\fS(\aa,t)\int_{\LL} \beta(\aa,\aa_*){\fI(\aa_*,t)}\ d\aa_*\,d\aa - \int_\Lambda\gamma (\aa) \fI(\aa,t)\,d\aa,
\end{split}
\end{equation}
where the fraction of recovered is obtained from $R(t)=1-S(t)-I(t)$. We refer to \cite{H96,H00} for analytical results concerning model \eqref{eq:SIR} and \eqref{eq:SIRm}. In the following we will adopt the simple compartmental model \eqref{eq:SIR} to derive our feedback controlled formulation in presence of uncertainty. {The extension to more realistic compartmental models in epidemiology, designed specifically for the COVID-19 pandemic, can be carried out in a similar way \cite{GattoPNAS}}.

In order to simplify the description, we will consider the {one-dimensional} case $d_a=1$ and set the social dependence as the age $a$ of the individual because of  its importance in epidemic dynamics. It is clear, however, that similar containment procedures {can impact also on other social features, like the wealth of the individuals \cite{DPTZ}}.  We will first formulate the feedback controlled SIR model in the deterministic case and subsequently extend our approach to the presence of uncertain parameters.

\subsection{Optimal control of structured compartmental model}
In order to define the action of a policy maker introducing a control over the system based on social distancing and other containment measures linked to the social structure we consider an optimal control framework.  The choice of an appropriate functional is problem dependent \cite{SCCC10}. 

In our setting, we account the minimization of the total number of the infected population $I(t)$ through the an age dependent control action depending both on time and pairwise interactions among individuals with different ages. 

Thus, we introduce the optimal control problem
\begin{equation}\label{eq:func}
\min_{u\in \mathcal U}J(u):= \int_0^T\psi(I(t))  dt+\dfrac{1}{2}\int_0^T\int_{\LL\times\LL} {\nu(\as,\as_*,t)}|u(\as,\as_*,t)|^2\ d\as d\as_*  dt,
\end{equation}
subject to
\begin{equation}\label{eq:SIRc}
\begin{split}
\frac{d}{dt} \fS(\as,t)&= - \fS(\as,t)\int_{\LL} (\beta(\as,\as_*) -u(\as,\as_*,t)){\fI(\as_*,t)}\ d\as_* \\
\frac{d}{dt} \fI(\as,t) &=  \fS(\as,t)\int_{\LL} (\beta(\as,\as_*) -u(\as,\as_*,t)){\fI(\as_*,t)}\ d\as_* - \gamma (\as) \fI(\as,t) \\
\frac{d}{dt} \fR(\as,t) &= \gamma (\as) \fI(\as,t)
\end{split}
\end{equation}
with initial condition $\fI(\as,0) = \fI_0(\as)$, $\fS(\as,0) = \fS_0(\as)$ and $\fR(\as,0) = r_0(\as)$.

The number of infected individual is measured by a monotone increasing  function $\psi(\cdot)$ such that $\psi:[0,1]\to\mathbb{R}_{+}$. This function models the policy maker's perception of the impact of the epidemic by the number of people currently infected {and in the sequel will be referred to as \emph{perception function}}. For example $\psi(I)=I^q/q$, for $q>1$ implies an underestimation of the actual number of infected corresponding to $q=1$. The control aims to minimize this measure of the total infected population by reducing the rate of interaction between individuals. We consider a quadratic cost for its actuation.

Such control is restricted to the space of admissible controls 
\[
\mathcal U \equiv \left\{ u\,|\, 0 \leq u(\as,\as_*,t) \leq \min\{M,\beta(\as,\as_*)\},\,\, \forall\, (\as,\as_*,t)\in \Lambda^2\times[0,T],\, M>0\right\},
\] 
which ensure the admissibility of the solution for \eqref{eq:SIRc}. {The above restriction on admissible controls can be relaxed if we consider controls that violate the previous condition locally but preserve the inequality in integral form after integration against $i(a_*,t)$.}   

The solution to problem \eqref{eq:func}-\eqref{eq:SIRc} is computed through the optimality conditions obtained from the Euler-Lagrangian as follows
\begin{align*}
\mathcal{L}&(\fS,\fI,\fR,\pS,\pI,\pR,u) = J(u) + \cr&\int_0^T\int_{\LL} \pS\cdot\left(\frac{d}{dt}\fS  +{\fS(\as,t)}{}\int_{\LL}\left(\beta (\as,\as_*)-u (\as,\as_*,t)\right) \fI(\as_*,t)d\as_*\right)\ d\as\,dt\cr
&+\int_0^T\int_{\LL}\pI\cdot\left(\frac{d}{dt}\fI - {\fS(\as,t)}{}\int_{\LL}\left(\beta (\as,\as_*)-u (\as,\as_*,t)\right) \fI(\as_*,t)d\as_*+\gamma (\as)\fI(\as,t)\right) d\as\,dt\cr
& + \int_0^T\int_{\LL} \pR\cdot\left( \dfrac{d}{dt}\fR - \gamma(\as)i(\as,t)\right) d\as\,dt
\end{align*}
where $\pS (\as,t),\pI (\as,t),\pR (\as,t)$ are the associated Lagrangian multipliers. 
%
By computing the variations with respect to $(s,i,r)$ we retrieve the adjoint system
\begin{equation}
\begin{split}
\label{eq:opt1}
\frac{d}{dt} \pS &= (\pS-\pI)\int_{\LL}(\beta (\as,\as_*)-u (\as,\as_*,t)){\fI(\as_*,t)}{}d\as_*\\
\frac{d}{dt} \pI &= \psi'(I(t))+\int_{\LL}(\pS(\as_*,t)-\pI(\as_*,t))(\beta (\as_*,\as)-u (\as_*,\as,t)){\fS(\as_*,t)}{}d\as_*+\gamma (\as)\pI
\end{split}
\end{equation}
with terminal conditions  $ \pS (\as,T) =0, \pI (\as,T) =0$ and $\pR (\as,T)= 0$.  
Note that the contribution of $\pR (\as,t)$ vanishes since the control does not act directly on population $R$, and the removed population is not considered in the minimization of the functional. The optimality condition reads
\begin{align}\label{eq:opt2}
{\nu(\as,\as_*,t)} u (\as,\as_*,t)&= \left(\pS-\pI\right)\fS( \as,t){i(\as_*,t)}{}.
\end{align}

The optimality conditions \eqref{eq:opt1}-\eqref{eq:opt2}
are first order necessary conditions for the optimal control {$u(\as,\as_*,t)$}. In order to be admissible then the control reads
\[
u (\as,\as_*,t)=\max\left\{0, \min\left\{\dfrac{\pS-\pI}{{\nu(\as,\as_*,t)} }\fS( \as,t){i(\as_*,t)},\phi_{M,\beta}(\as,\as_*)\right\}\right\},
\]
where  $\phi_{M,\beta}(\as,\as_*)=\min\{\beta(\as,\as_*),M\}$.

The approach just described, however, is generally quite complicated when there are uncertainties as it involves solving simultaneously the forward problem \eqref{eq:func}- \eqref{eq:SIRc} and the backward problem \eqref{eq:opt1}- \eqref{eq:opt2}. Moreover, the assumption that the policy maker follows an optimal strategy over a long time horizon seems rather unrealistic in the case of a rapidly spreading disease such as the COVID-19 epidemic. {Let us emphasize that extending the above optimal control formulation to more complex compartmental models designed specifically for COVID-19, like SEPIAR or SIDHARTE \cite{GattoPNAS,Giordano}, can be done by generalizing the control functional \eqref{eq:func} to include, for example, the hospitalized compartment, or other specific indicators that can be measured from the data. For all of these models, the feedback control strategy described in the next section does not change substantially.} {We refer the reader to Appendix \ref{sec:gen} for more details.}

\subsection{Feedback controlled compartmental models}\label{sect:ins}
In this section we consider short time horizon strategies which permits to derive suitable feedback controlled models. These strategies are suboptimal with respect the original problem \eqref{eq:func}-\eqref{eq:SIRc} but they have proved to be very successful in several social modeling problems \cite{Albi1,Albi2,Albi3,Albi4,DPT}. To this aim, we consider a short time horizon of length $h>0$ and formulate a time discretize optimal control problem through the functional $J_h(u)$ restricted to the interval $[t,t+\h]$, as follows
\begin{equation}\label{eq:func_I}
\min_{u\in\mathcal U} J_h(u) := \psi(I(t+\h))+\frac{1}{2}\int_{\Lambda\times\Lambda}{\nu(\as,\as_*,t)}|u(\as,\as_*,t)|^2d\as d\as_*
\end{equation}
subject to
\begin{align}\label{eq:SIR_I}
\fS(\as,t+\h) &=  \fS(\as,t) - \h {\fS(\as,t)}{}\int_{\LL} \left(\beta(\as,\as_*) - u(\as,\as_*,t)\right)\fI(\as_*,t) d\as_* \\
\label{eq:SIR_I2}
\fI(\as,t+\h) &= \fI(\as,t)  + \h {\fS(\as,t)}{}\int_{\LL} \left(\beta(\as,\as_*) - u(\as,\as_*,t)\right)\fI(\as_*,t) d\as_*- \h\gamma(\as) \fI(\as,t).
\end{align}
By recalling that the macroscopic information on the infected is
\[
I(t+\h)= I(t)+ \h \int_{\LL} \left[ {\fS(\as,t)}{}\int_{\LL} \left(\beta(\as,\as_*) - u(\as,\as_*,t)\right)\fI(\as_*,t) d\as_*- \gamma(\as) \fI(\as,t)\right]\ d\as,
\] 
we can derive the minimizer of $J_h$ computing $D_u J_h(u)\equiv 0$. We retrieve the following nonlinear equation
\begin{align}\label{eq:ic}
{\nu(\as,\as_*,t)}u (\as,t)={\h}\fS(\as,t) i(\as_*,t)\psi'(I(t+h)).
\end{align}
{In order to pass to the limit $h\to 0$ we must rescale the penalization term as {$\nu(\as,\as_*,t) = \h \kappa(\as,\as_*,t)$} so that we can introduce the above instantaneous strategy directly in the discrete system \eqref{eq:SIR_I}-\eqref{eq:SIR_I2}}. 

The resulting controlled dynamic, {corresponding to the feedback controlled continuous system \eqref{eq:SIRc}}, reads as follows
\begin{subequations}\label{eq:SIR_ic}
\begin{align}
\frac{d}{dt} \fS(\as,t)&= - \fS(\as,t)\int_{\LL} \Big(\beta(\as,\as_*) -{\frac{\fS(\as,t)\fI(\as_*,t) \psi'(I(t))}{\kappa(\as,\as_*,t)}}\Big)\fI(\as_*,t)d\as_*\\
\frac{d}{dt} \fI(\as,t) &=  \fS(\as,t)\int_{\LL} \Big(\beta(\as,\as_*) -{\frac{\fS(\as,t)\fI(\as_*,t) \psi'(I(t))}{\kappa(\as,\as_*,t)}}\Big)\fI(\as_*,t)d\as_*- \gamma (\as) \fI(\as,t) \\
\frac{d}{dt} \fR(\as,t) &= \gamma (\as) \fI(\as,t).
\end{align}
\end{subequations}
In what follows we provide a sufficient conditions for the admissibility of the instantaneous control in terms of the penalization term {$\kappa(\as,\as_*,t)$}. Indeed we want to assure that the dynamics preserve the monotonicity of the number of susceptible population $s(a,t)$.

\begin{proposition} Let $\beta(\as,\as_*)\geq\delta>0$, and $\psi'(\cdot)$ {a monotonically non decreasing function}, then  for all $(\as,\as_*)\in\LL\times\LL$ {and $t>0$}, solutions to \eqref{eq:SIR_ic} are admissible if the penalization $\kappa$ satisfies the following inequality
\begin{equation}\label{eq:cond}
	\delta {\kappa(\as,\as_*,t)}\geq {s_0(\as)\bar i(\as_*)\psi'(\bar I)}{},\qquad \forall (\as,\as_*)\in\LL\times\LL
\end{equation}
where $\bar i(\as)$ and $\bar I$ are respectively the peak reached by the infected of {age} $\as$ and by the total population {of infected}.
\end{proposition}
\begin{proof} By imposing the non-negativity of the controlled reproduction rate inside the integral
we have
\[
\beta(\as,\as_*) {\kappa(\as,\as_*,t)}\geq {s(\as,t) i(\as_*,t)\psi'(I(t))}{}.
\]
This inequality has to be satisfied for every time $t\geq 0$. Next we observe that the number of susceptible $s(\as,t)$ is decreasing in time, therefore $s_0(a)\geq s(a,t)$ for all $t$. Moreover $i(\as,t)$ reaches a peak before decreasing to 0 (note that this peak can also be in $t=0$), say $\bar I$ for the macroscopic variable and $\bar i(a)$. Thus, thanks to the { monotonicity of $\psi'(\cdot)$}, we can restrict the previous inequality as follows
\[
\delta {\kappa(\as,\as_*,t)}\geq {s_0(\as) \bar i(\as_*)\psi'(\bar I)}{}.
\]
\end{proof}

\begin{figure}
	\centering
	\includegraphics[scale = 0.5]{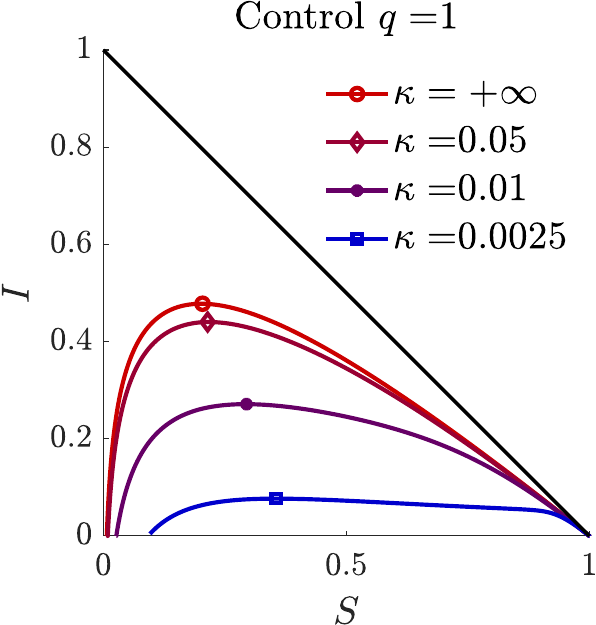}
	\includegraphics[scale = 0.5]{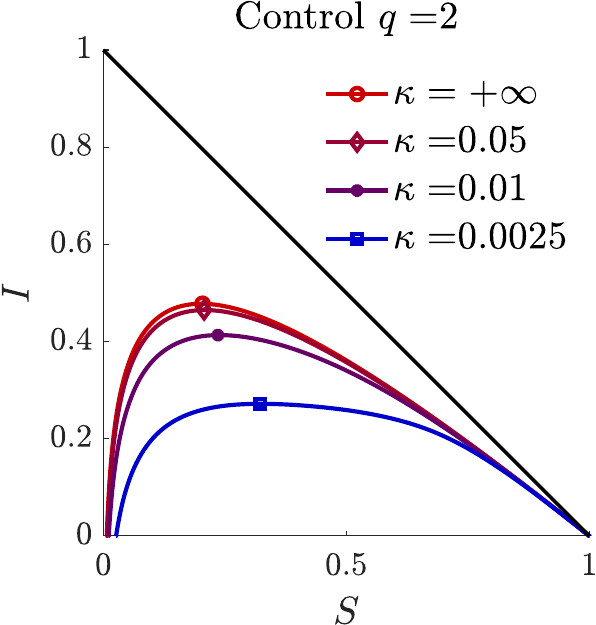}
	\caption{Phase diagram of susceptible-infected trajectories for the controlled SIR-type model with homogenous mixing and $\psi(I) = I^q/q$. Different choices of the penalization term $\kappa$ are reported. Left plot the case $q=1$, right plot $q=2$. {The line markers point out the peaks of the infected population for each choice of $\kappa$.}}
\label{fg:phase}
\end{figure}

In Figure \ref{fg:phase} we report the phase diagram of susceptible-infected trajectories for the controlled model with homogeneous mixing
\begin{equation}
\begin{split}
\frac{d}{dt} S(t)&= - \left(\beta-\frac1{\kappa}S(t)I(t)\psi'(I(t))\right)S(t)I(t) \\
\frac{d}{dt} I(t) &= \left(\beta-\frac1{\kappa}S(t)I(t)\psi'(I(t))\right)S(t)I(t)-\gamma I(t),
\end{split}
\label{eq:homo}
\end{equation}
with $\psi(I)=I^q/q$.
The dynamics are similar to the classical SIR model but with a nonlinear contact rate. In particular, the trajectories are flattened when the value of the control is such that $\kappa\beta \approx S(t)I(t)\psi'(I(t))$. Note, however that this status is not an equilibrium point of the system.

To understand this, let us observe that an equilibrium state $(S^*,I^*)$ for \eqref{eq:homo} satisfies the equations
\[
\begin{split}
\left(\kappa\beta-S^*I^*\psi'(I^*)\right)S^*I^*=0\\
\left((\kappa\beta-S^*I^*\psi'(I^*))S^*-\kappa\gamma\right)I^*=0.
\end{split}
\]
An equilibrium point corresponds to the classical state in which we have the extinction of the disease $I^*=0$ and $S^*$ arbitrary and defined by the asymptotic state of the dynamics \cite{H00}. Now, let's suppose that $I^* \neq 0$, $S^*\neq 0$, we can look for solutions where control is able to perfectly balance the spread of the disease
\[
\begin{split}
\kappa\beta=S^*I^*\psi'(I^*)\\
\kappa\beta S^*=(S^*)^2 I^* \psi'(I^*)+\kappa\gamma,
\end{split}
\]
consequently
\[
\kappa\beta S^* = (S^*)^2 \frac{\kappa\beta}{S^*} + \kappa \gamma = 0,
\]
which is satisfied only for $\gamma=0$ when $\kappa \neq 0$. {The stability analysis of this unique equilibrium point can be performed by standard arguments and we omit the details.} 


\section{Control of epidemic dynamics with uncertainties}\label{sect:control_UQ}

Since the beginning of the outbreak of new infectious diseases, the actual number of infected and  recovered people is typically underestimated, causing fatal delays in the implementation of public health policies facing the propagation of epidemic fronts. This is the case of the spreading of COVID-19 worldwide, often mistakenly underestimated due to deficiencies in surveillance and diagnostic capacity \cite{JRGL,RR}. Health systems are struggling to adopt systematic testing to monitor actual cases. Moreover, another important epidemiological factor affecting data reliability is the large proportion of asymptomatic \cite{JRGL,MKZC,Zhang_etal}. 

Among the common sources of uncertainties for dynamical systems modeling epidemic outbreaks we may consider the following
\begin{itemize}
\item noisy and incomplete available data;
\item structural uncertainty due to the possible inadequacy of the mathematical model used to describe the  phenomena under consideration.
\end{itemize}
In the following we consider the effects on the dynamics of uncertain data, such as the initial conditions on the number of infected people or the interaction and recovery rates. On the numerical level we consider techniques based on stochastic Galerkin methods, for which spectral convergence on random variables is obtained under appropriate regularity assumptions \cite{X}. For simplicity, the details of the numerical method that allows to reduce the uncertain dynamic system to a set of deterministic equations are reported in Appendix \ref{app:sG}.

\subsection{Socially structured models with uncertain inputs}

We introduce the random vector $\z = (z_1,\dots,z_{d_z})$ whose components are assumed to be independent real valued random variables
\[
z_k:(\Omega,F) \rightarrow (\mathbb R,\mathcal B_{\mathbb R}),  \qquad k = 1,\dots,d_z
\]
with $\mathcal B_{\mathbb R}$ the Borel set. We assume to know the probability density $p(\z): \mathbb R^{d_z} \rightarrow \mathbb R^{d_z}_+$ characterizing the distribution of $\z$. 
Here, $z\in\mathbb{R}^{d_z}$ is a random vector taking into account various possible sources of uncertainty in the model.

In presence of uncertainties we generalize the initial modeling by introducing the quantities $s(\z,\as,t)$, $i(\z,\as,t)$ and $r(\z,\as,t)$ representing the distributions at time $t\ge 0$ of susceptible, infectious and recovered individuals. The total size of the population is a deterministic conserved quantity in time, i.e.
\[
s(\z,\as,t)  + i(\z,\as,t) + r(\z,\as,t) = f(\as), \qquad \int_{\LL} f(\as)d\as = N. 
\]
Hence, the quantities 
\begin{equation}
S(\z,t)=\int_{\LL}\fS (\z,\as,t)\,d\as,\quad I(\z,t)=\int_{\LL}\fI (\z,\as,t)\,d\as,\quad R(\z,t)=\int_{\LL}\fR (\z,\as,t)\,d\as,
\end{equation}
denote the uncertain fractions of the population that are susceptible, infectious and recovered respectively. 

The social structured model with uncertainties reads
\begin{equation}\label{eq:SIR_u}
\begin{split}
\frac{d}{dt} \fS(\z,\as,t)&= - \fS(\z,\as,t)\int_{\LL} \beta(\z,\as,\as_*)\dfrac{\fI(\z,\as_*,t)}{N}\ d\as_* \\
\frac{d}{dt} \fI(\z,\as,t) &=  \fS(\z,\as,t)\int_{\LL} \beta(\z,\as,\as_*) \dfrac{\fI(\z,\as_*,t)}{N}\ d\as_* - \gamma (\z,\as) \fI(\z,\as,t) \\
\frac{d}{dt} \fR(\z,\as,t) &= \gamma (\z,\as) \fI(\z,\as,t)
\end{split}
\end{equation}
To illustrate the impact of uncertainties let us consider the simple following example. In the case of homogeneous mixing with uncertain contact rate $\beta(z)=\beta+\alpha z$, $\alpha > 0$, $z\in \mathbb{R}$ distributed as $p(z)$ the model reads
\begin{equation}
\begin{split}
\frac{d}{dt} S(z,t)&= - \left(\beta+\alpha z\right)S(z,t)I(z,t) \\
\frac{d}{dt} I(z,t) &= \left(\beta+\alpha z\right)S(z,t)I(z,t)-\gamma I(z,t),
\end{split}
\label{eq:homo2}
\end{equation}
with deterministic initial values $I(z,0)=I_0$ and $S(z,0)=S_0$.
The solution for the proportion of infectious during the initial exponential phase is \cite{Roberts}
{\[
I(z,t) = I_0 e^{(\beta+\alpha z)S_0 t - \gamma t},
\]}
and its expectation
{
\begin{equation}
{\mathbb E}[I(\cdot,t)]=I_0 \int_{\mathbb{R}} e^{(\beta+\alpha z)S_0 t - \gamma t}p(z)\,dz = I_0 e^{\beta S_0 t - \gamma t} W(t),
\end{equation}}
where the function 
{
\[
W(t)=\int_{\mathbb{R}} e^{\alpha z S_0 t} p(z)\,dz
\]}
represents the statistical correction factor to the standard deterministic exponential phase of the disease {$I_0 e^{\beta S_0 t - \gamma t}$}. If $z$ is uniformly distributed in $[-1,1]$ we can explicitly compute 
{
\[
W(t)=\frac{\sinh \left({\alpha S_0 t}\right)}{\alpha S_0 t} > 1,\quad t>0. 
\]}
More in general, if $z$ has zero mean then by Jensen's inequality we have $W(t)>1$ for $t > 0$, so that the expected exponential phase is amplified by the uncertainty (see \cite{Roberts}). 

In a similar way, keeping $\beta$ constant, but introducing a source of uncertainty in the initial data $I(z,0)=I_0+\mu z$, $\mu > 0$ and $z\in \mathbb{R}$ distributed as $p(z)$ the solution in the exponential phase reads
{
\[
I(z,t) = (I_0+\mu z) e^{\beta S_0 t - \gamma t},
\]}
and then its expectation 
{
\begin{equation}
\begin{split}
{\mathbb E}[I(\cdot,t)]&=\int_{\mathbb{R}}(I_0+\mu z) e^{\beta S_0 t - \gamma t}p(z)\,dz= (I_0 + \mu\bar{z}) e^{\beta S_0 t - \gamma t},
\end{split}
\end{equation}}
where $\bar{z}$ is the mean of the variable $z$. Therefore, the expected initial exponential growth behaves as the one with deterministic initial data $I_0+\mu \bar{z}$. Of course, if both sources of uncertainty are present the two effects just described sum up in the dynamics.

{
\begin{remark}
The presence of a large number of undetected infected is at the basis of the construction of numerous epidemiological models with an increasingly complex compartmental structure in which the original compartment of the infected is subdivided into further compartments with different roles in the propagation of the disease \cite{Giordano,GattoPNAS,IC}. 
To clarify the relationships to other deterministic compartmental models, let us consider, for simplicity, model \eqref{eq:SIR_u} in absence of a social structure 
\begin{equation}
\begin{split}
\frac{d}{dt} S(z,t)&= - \beta(z) S(z,t)I(z,t) \\
\frac{d}{dt} I(z,t) &= \beta(z)S(z,t)I(z,t)-\gamma(z) I(z,t),\\
\frac{d}{dt} R(z,t) &= \gamma(z) I(z,t),
\end{split}
\label{eq:homo3}
\end{equation}
and with a one-dimensional random input $z \in \mathbb{R}$ distributed as $p(z)$. Furthermore, for a function $F(z,t)$ we will denote its expected value as $\bar F(t) = \mathbb{E}[F(\cdot,t)]$. Now, starting from a discrete probability density function
\[
p_k=P\left\{Z=z_k\right\},\qquad \sum_{k=1}^n p_k = 1,
\]
we have $\bar F(t) = \sum_{k=1}^n p_k F_k$, with $F_k=F(z_k)$. 
Taking the expectation in \eqref{eq:homo3}, we can write
\begin{equation}
\begin{split}
\frac{d}{dt} \bar S(t)&= - \bar S(t) \sum_{k=1}^n \tilde \beta_k  p_k I_k(t) \\
\frac{d}{dt} \bar I(t) &= \bar S(t) \sum_{k=1}^n \tilde \beta_k  p_k I_k(t)-\sum_{k=1}^n \gamma_k p_k I_k(t),\\
\frac{d}{dt} \bar R(t) &= \sum_{k=1}^n \gamma_k p_k I_k(t),
\end{split}
\label{eq:homopar}
\end{equation}
with $\tilde \beta_k = S_k \beta_k/\bar S $, $k=1,\ldots,n$. For example, in the case n=2, by identifying $I_d=p_1 I_1$ and $I_u=p_2 I_2$ with the compartments of detected and undetected infectious individuals, respectively, we can formulate the equivalent partitioning
\begin{equation}
\begin{split}
\frac{d}{dt} \bar S(t)&= - \bar S(t) \left(\tilde \beta_1  I_d(t)+\tilde \beta_2  I_u(t)\right) \\
\frac{d}{dt}  I_d(t) &= \bar S(t) \tilde \beta_1  I_d(t)-\gamma_1 I_d(t),\\
\frac{d}{dt}  I_u(t) &= \bar S(t) \tilde \beta_2  I_u(t)-\gamma_2 I_u(t),\\
\frac{d}{dt} \bar R(t) &= \gamma_1 I_d(t)+\gamma_2 I_u(t),
\end{split}
\label{eq:SPIAR}
\end{equation}
which has the same structure of a SIAR compartmental model including the undetected (or the asymptomatic) class.
In the following, we will not rely on discrete probability distributions, but on continuous representations that can be associated with the overall probability of having a certain number of infectious individuals (detected or undetected). The additional dependence of the epidemiological parameters from the random variable allows us to take into account changes in the corresponding dynamics of disease transmission and recovery. 
\end{remark}  
}
  
\subsection{The feedback controlled model with random inputs}
In presence of uncertainties the optimal control problem \eqref{eq:func}-\eqref{eq:SIRc}  is modified as follows 
\begin{equation}\label{eq:func_z}
\min_{u\in \mathcal U}J(u):= \int_0^T\mathcal R [\psi(I(\cdot,t))] dt+\dfrac{1}{2}\int_0^T\int_{\LL\times\LL} {\nu(\as,\as_*,t)}|u(\as,\as_*,t)|^2\ d\as d\as_*  dt,\qquad 
\end{equation}
being $\mathcal R[\psi(I(\cdot,t))]$ a suitable operator taking into account the presence of the uncertainties $\z$. Examples of such operator that are of interest in epidemic modeling are the expectation with respect to uncertainties
\begin{equation}
{\mathcal R[\psi(I(\cdot,t))]=}\mathbb E[\psi(I(\cdot,t))]= \int_{\mathbb R^{d_z}} \psi(I(\z,t)) \; p(\z)d\z
\label{eq:R1}
\end{equation} 
or relying on data which underestimate the number of infected
\begin{equation} 
{\mathcal R[\psi(I(\cdot,t))] = \psi(I(\z_0,t)),}
\label{eq:R2}
\end{equation}
{where $\z_0$ is a given value such that $\psi(I(\z_0,t)) \leq \psi(I(\z,t))$, $\forall\, \z\in \mathbb{R}^{d_z}$ and $t>0$.}

In \eqref{eq:func_z} $\mathcal U$ the space of admissible controls is  defined as 
\[
\mathcal U=\left\{ u\,|\, 0 \leq u(\as,\as_*,t) \leq\min\{M,\min_{\z}\beta(\z,\as,\as_*)\},\,\, \forall\, (\as,\as_*,t)\in \LL^2\times[0,T],\, M>0\right\},
\] 
{or in a more relaxed form if we consider the above inequality after integration against $i(\as_*,t)$.}

The minimization problem \eqref{eq:func_z} is subject to the following dynamics
\begin{equation}\label{eq:SIR_z}
\begin{split}
\frac{d}{dt} \fS(\z,\as,t)&= - \fS(\z,\as,t)\int_{\LL} (\beta(\z,\as,\as_*) -u(\as,\as_*,t)){\fI(\z,\as_*,t)}{}\ d\as_* \\
\frac{d}{dt} \fI(\z,\as,t) &=  \fS(\z,\as,t)\int_{\LL} (\beta(\z,\as,\as_*) -u(\as,\as_*,t)){\fI(\z,\as_*,t)}{}\ d\as_* - \gamma (\z,\as) \fI(\z,\as,t) \\
\frac{d}{dt} \fR(\z,\as,t) &= \gamma (\z,\as) \fI(\z,\as,t)
\end{split}
\end{equation}
with initial condition $\fI(\z,\as,0) = \fI_0(\z,\as)$, $\fS(\z,\as,0) = \fS_0(\z,\as)$ and $\fR(\z,\as,0) = r_0(\z,\as)$.

The implementation of instantaneous control for dynamics in presence of uncertainties follows from the derivation presented in Section \ref{sect:ins}. {We can derive the minimizer of $J_h$ computing $D_u J_h(u)\equiv 0$ from the restriction of the minimization problem \eqref{eq:func_z} in the interval $[t,t+h]$ or equivalently   
\[
{\mathcal R}\left[\frac{\partial \psi(I(\cdot,t+\h))}{\partial u}\right] = {\nu(a,t)}u (\as,\as_*,t),
\]
where we assumed ${\partial {\mathcal R}\left[\psi(I(\cdot,t+\h))\right]}/{\partial u}={\mathcal R}\left[{\partial \psi(I(\cdot,t+\h))}/{\partial u}\right]$, to obtain the following nonlinear identity
\begin{align}\label{eq:icz}
{\nu(a,t)}u(\as,\as_*,t)={\h}{\mathcal R}[\fS(\cdot,\as,t) \fI(\cdot,\as_*,t)\psi'(I(\cdot,t))].
\end{align}
The above assumption on $\mathcal R[\cdot]$ is clearly satisfied by \eqref{eq:R1} and \eqref{eq:R2}, where in the case of $\eqref{eq:R2}$ we used the notation $${\mathcal R}[\fS(\cdot,\as,t) i(\cdot,\as_*,t)\psi'(I(\cdot,t))]=\fS(\z_0,\as,t) i(\z_0,\as_*,t)\psi'(I(\cdot,t)).$$ By introducing the scaling $\nu(a,t) = \h \kappa(\as,\as_*,t)$ we can pass to the limit for $h\to 0$ to get}  
\begin{equation}\label{eq:u_z}
u (\as,\as_*,t)=\frac{ 1}{\kappa(\as,\as_*)}\mathcal R \left[\fS(\cdot,\as,t)\fI(\cdot,\as_*,t)\psi'(I(\cdot,t))\right],
\end{equation}
which defines the feedback controlled model in presence of uncertainties.

\section{Examples from the COVID-19 outbreak in Italy}

In this section we present several numerical tests on the constrained compartmental model with uncertain data. Details of the numerical method used are given in Appendix \ref{app:sG}. In an attempt to analyze sufficiently realistic scenarios, in the following examples we will refer to values taken from Italian data on the COVID-19 epidemic \cite{Protezione}. More precisely, in the first test case we illustrate the behavior of the model in a simplified setting in absence of uncertainty and social structure and without trying to reproduce scenarios closely related to current data. In the second test case, following a progressively more realistic approach, we consider the impact of the presence of uncertain data in the controlled model with homogeneous social mixing and calibrated on Italian data. {The same setting is then considered in Test 3 taking into account the additional effects given by the social structure of the system, characterized by suitable social interaction functions and an age-dependent recovery rate. The final scenario, explored in Test 4, examines the influence on the spread of infectious disease induced by relaxed confinement measures related to the social structure of the system.}

\subsection{Test 1. Containment in homogeneous social mixing dynamics}
To illustrate the effects of controls introduced that mimic containment procedures, let us first consider the case where the social structure is not present. Furthermore, to simplify further the modeling, in this first example we neglect any dependence on uncertain data. 

We consider as initial small number of infected and recovered $i(0) = 3.68 \times 10^{-6}$, $r(0) = 8.33\times 10^{-8}$. These normalized fractions refer specifically to the first reported values in the case of the Italian outbreak of COVID-19, even if in this simple test case we will not try to match the data in a quantitative setting but simply to illustrate the behavior of the feedback controlled model.
Based on recent studies \cite{JRGL,Zhang_etal,Liu}, the initial infection rate of COVID-19 $R_0 = \beta/\gamma$ has been estimated between $2$ and $6.5$. Here, to exemplify the possible evolution of the pandemic we consider a value close to the lower bound, taking $\beta =0.25$ and $\gamma = 0.10$, {namely a recovery rate of $10$ days}, so that $R_0 = 2.5$. 

\begin{figure}
\centering
\includegraphics[scale = 0.27]{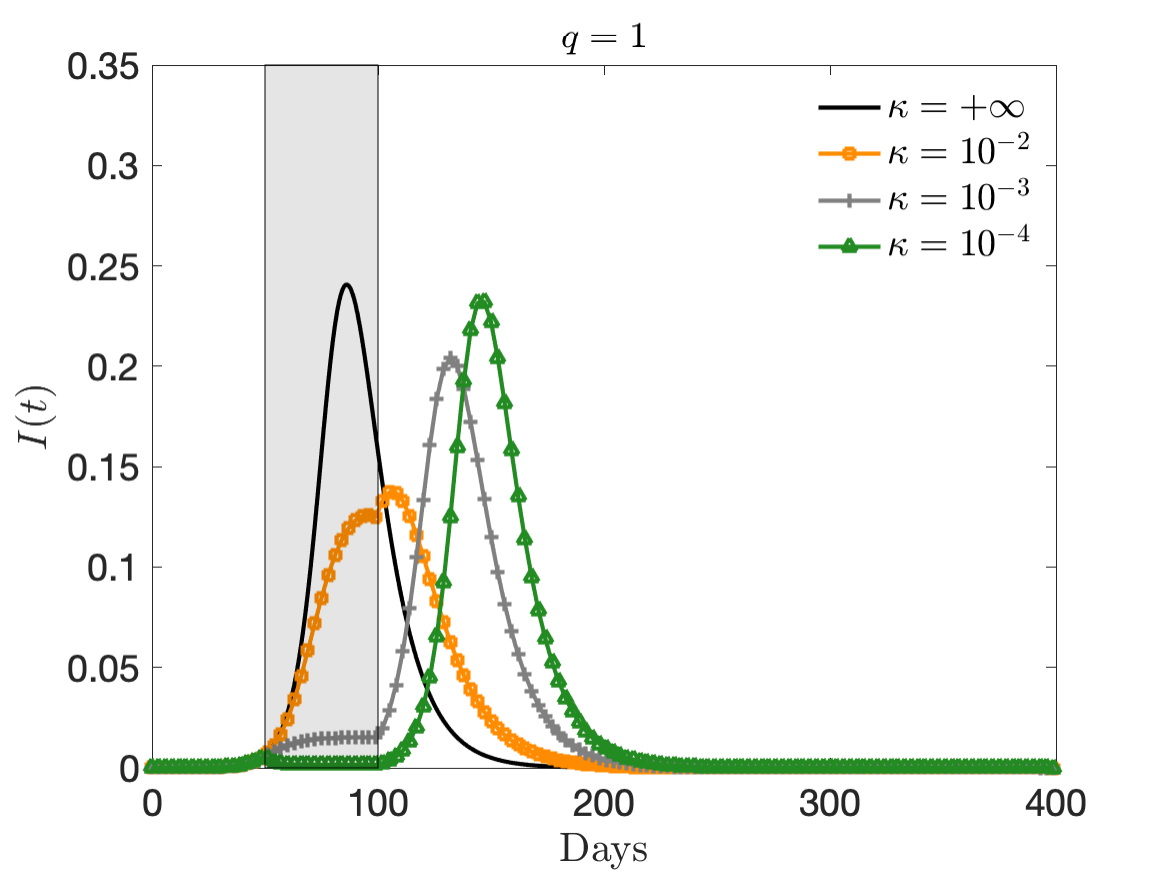}
\includegraphics[scale = 0.27]{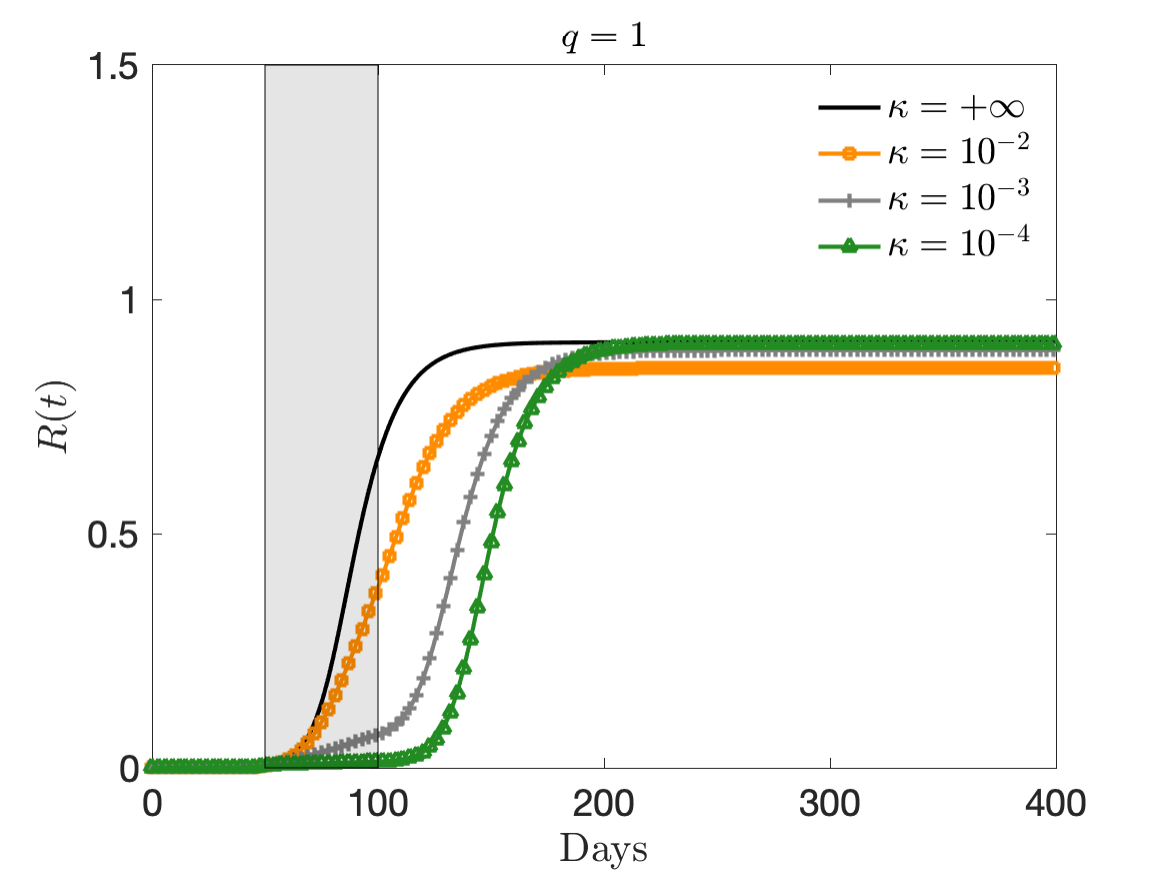} \\
\includegraphics[scale = 0.27]{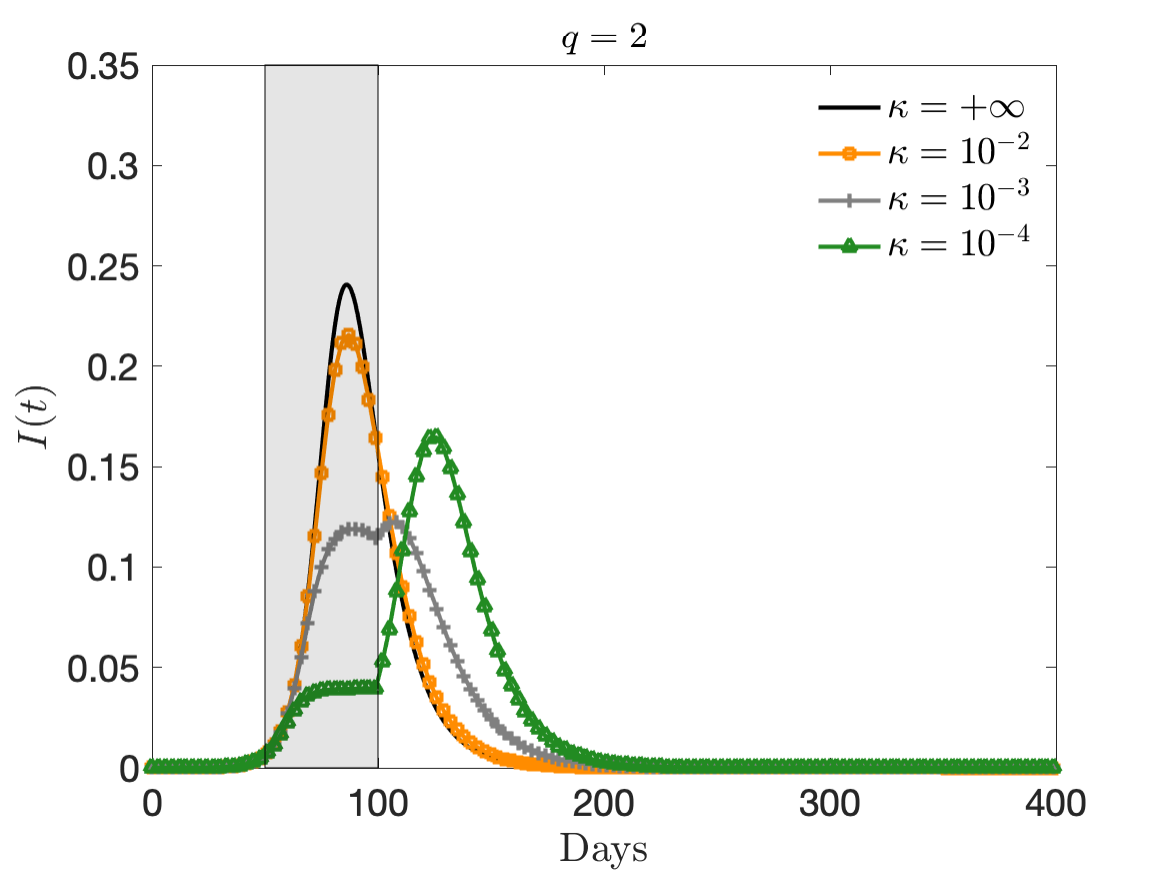}
\includegraphics[scale = 0.27]{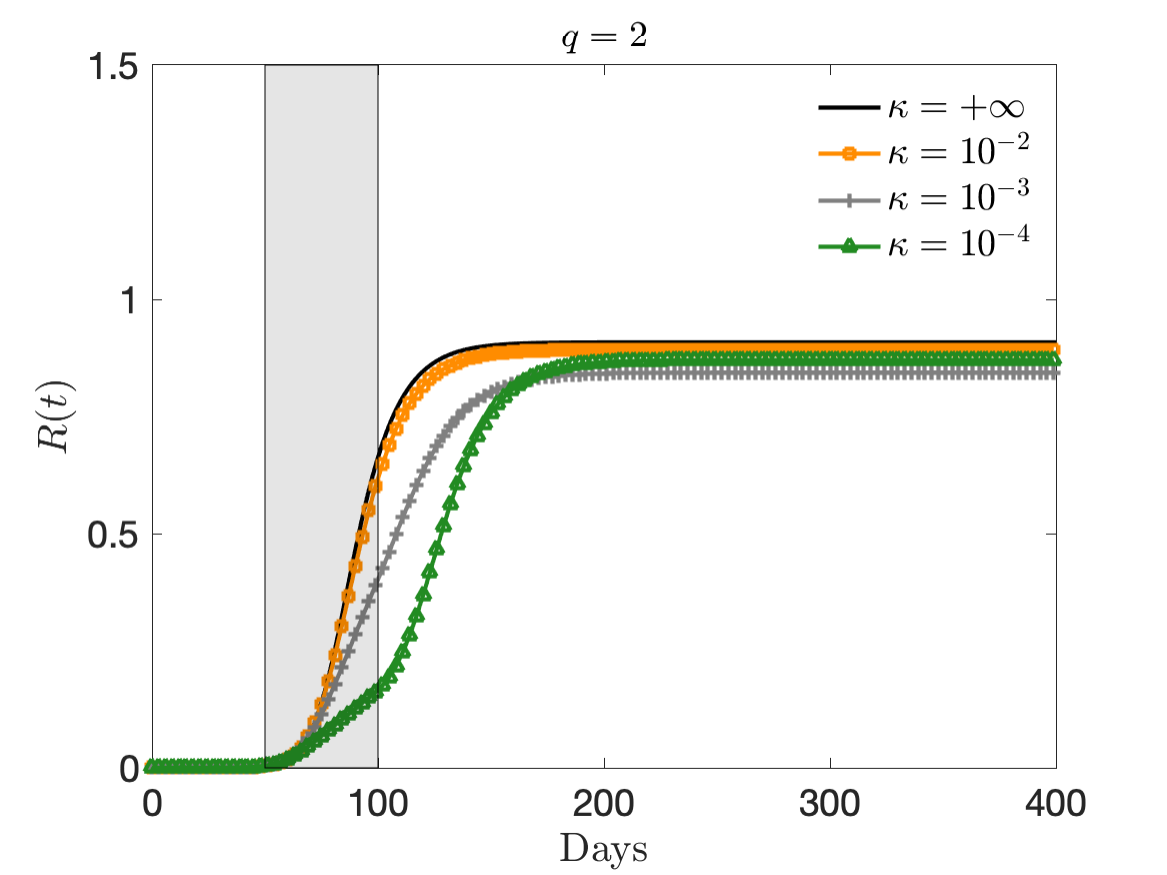}\\
\includegraphics[scale = 0.27]{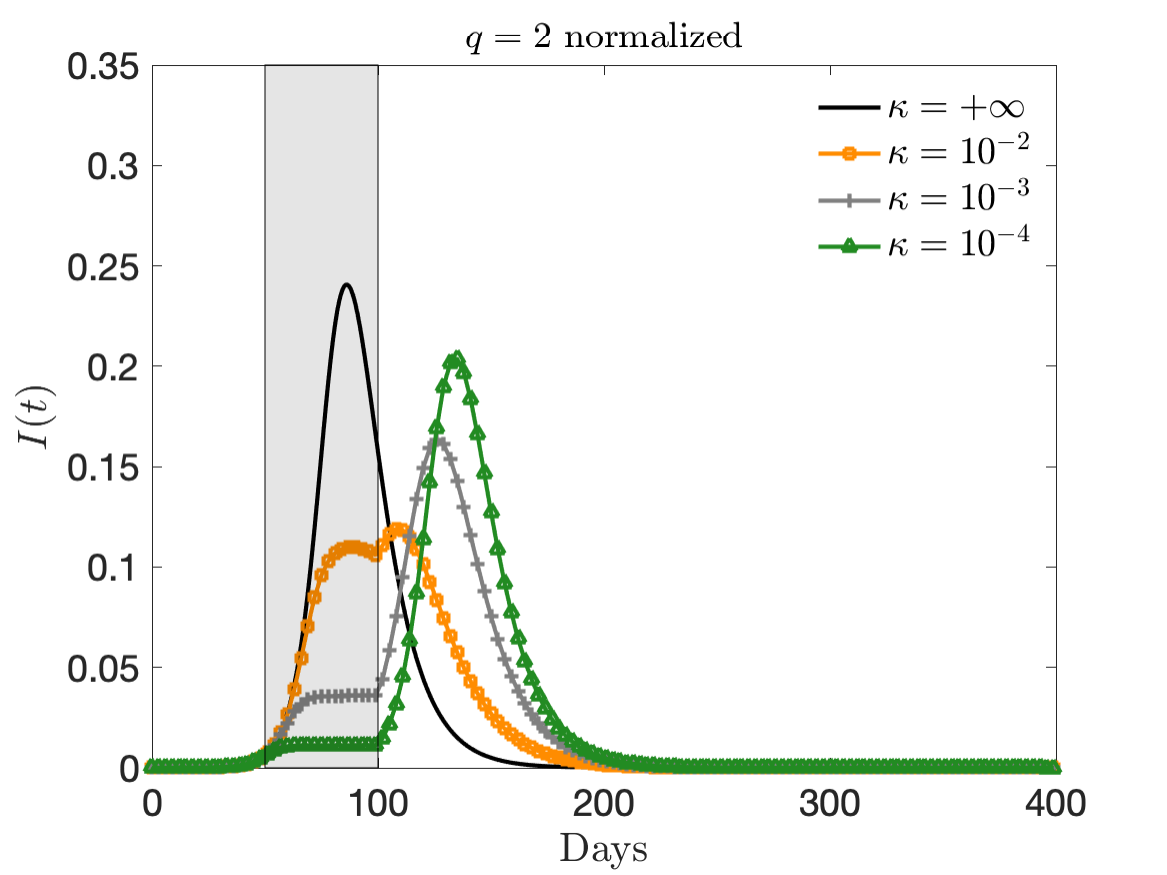}
\includegraphics[scale = 0.27]{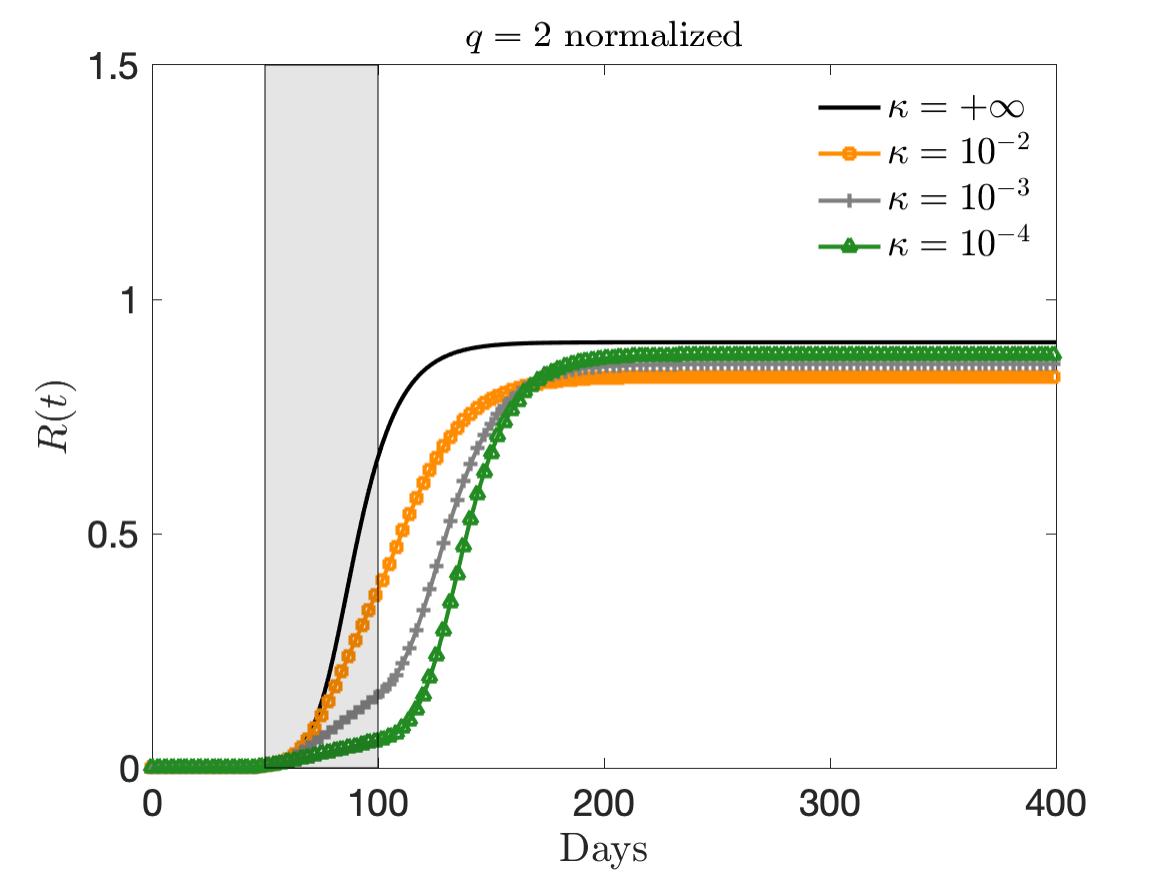}
\caption{\textbf{Test 1}.  Evolution of the fraction of infected (left) and recovered (right) based on {the feedback constrained model \eqref{eq:homo} for $t\in [50,100]$, a perception function} $\psi(I)=I^q/q$, $q = 1,2$ and several penalizations $\kappa = 10^{-2}, 10^{-3}, 10^{-4}$. The choice $\kappa = +\infty$ corresponds to the unconstrained case. {In last row the normalized case $\psi(I) = C_2 I^2/2$, with $C_2 = 12$. }}
\label{fig:test1a}
\end{figure}

\begin{figure}
\centering
\includegraphics[scale = 0.27]{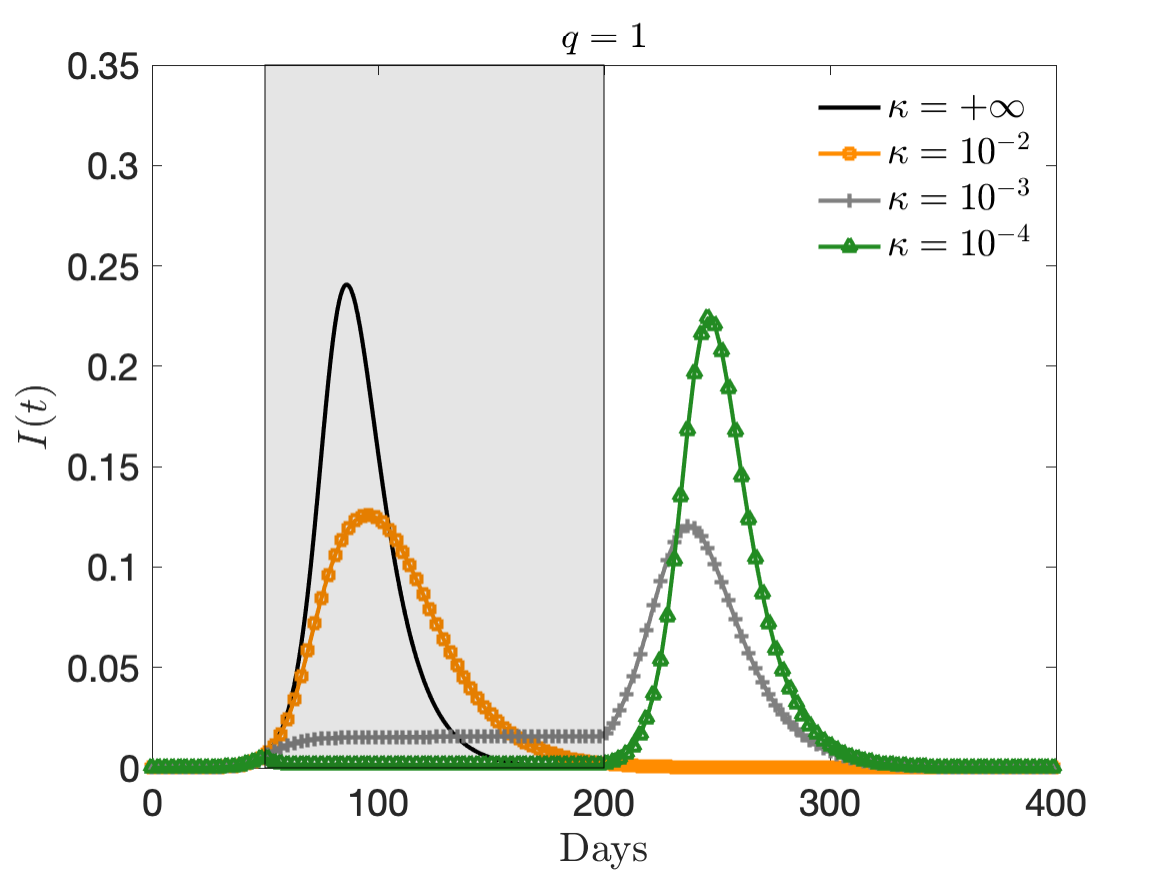}
\includegraphics[scale = 0.27]{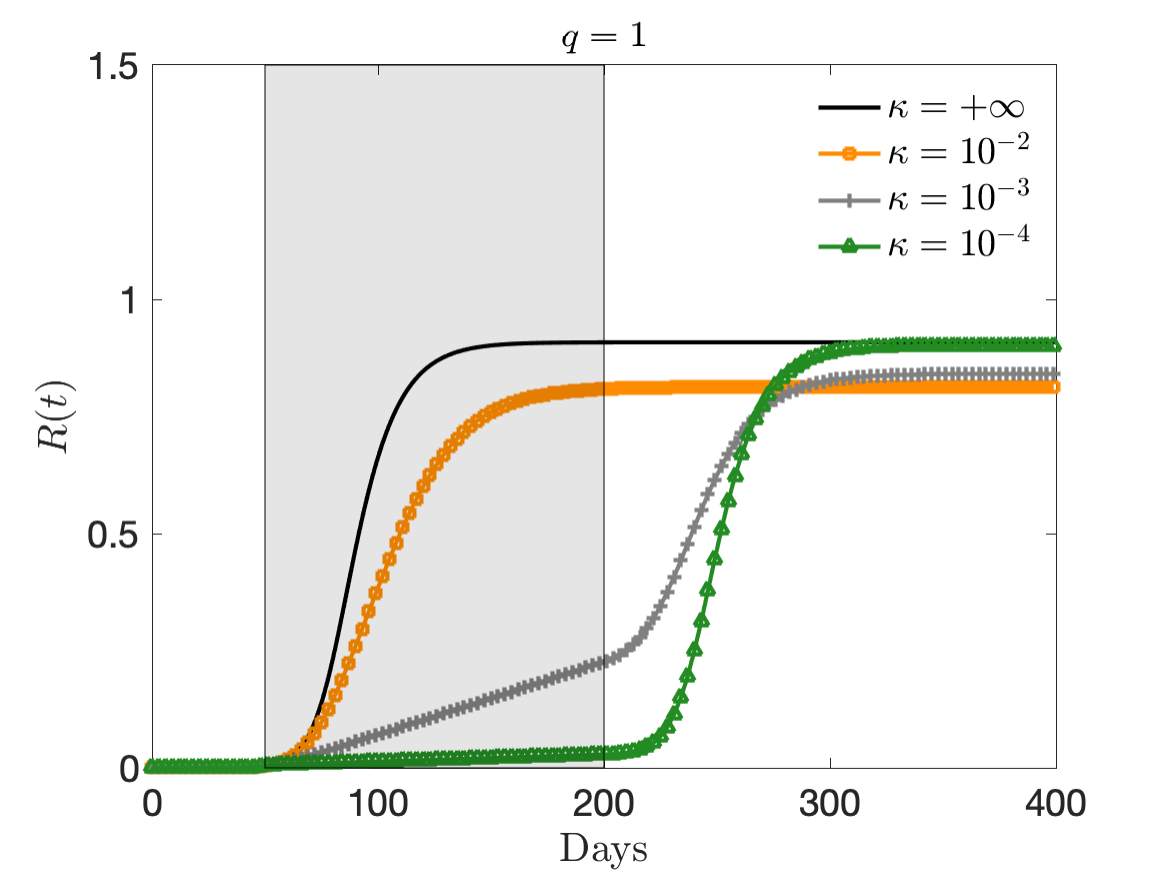} \\
\includegraphics[scale = 0.27]{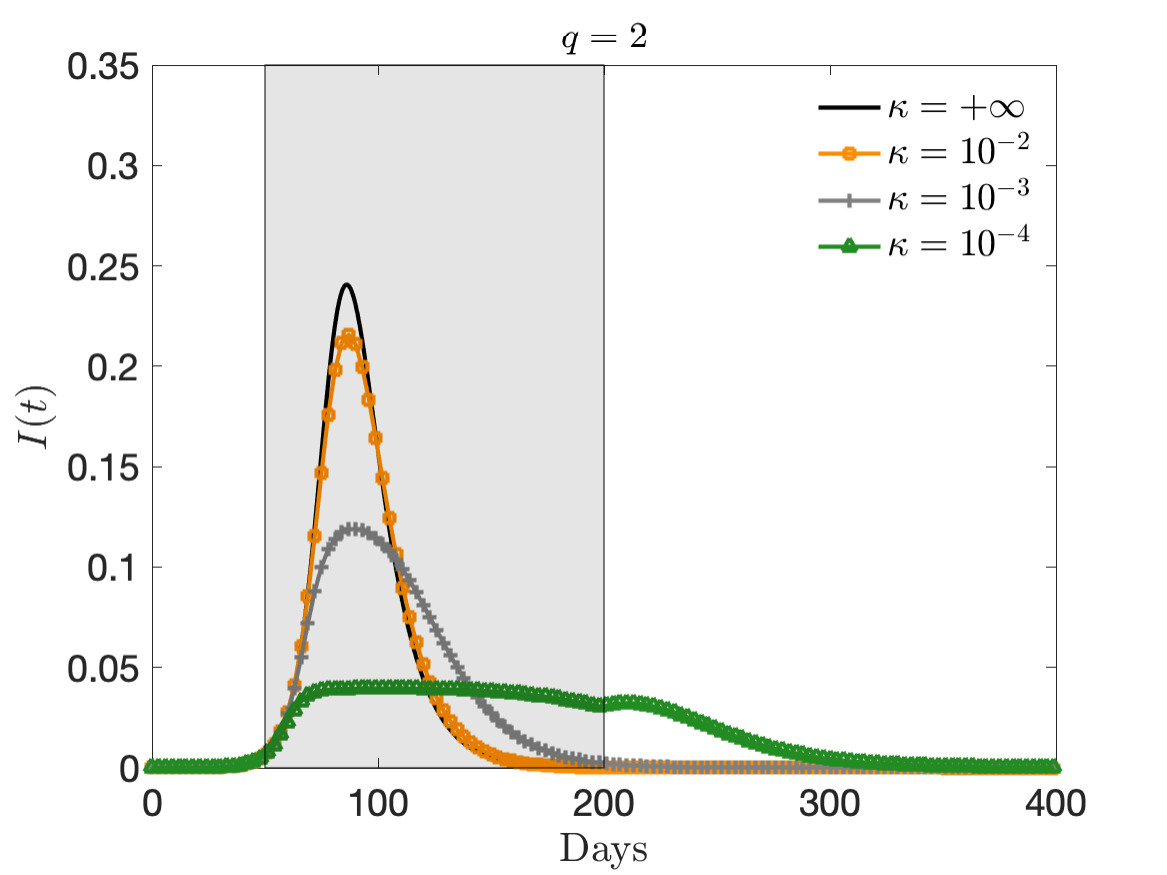}
\includegraphics[scale = 0.27]{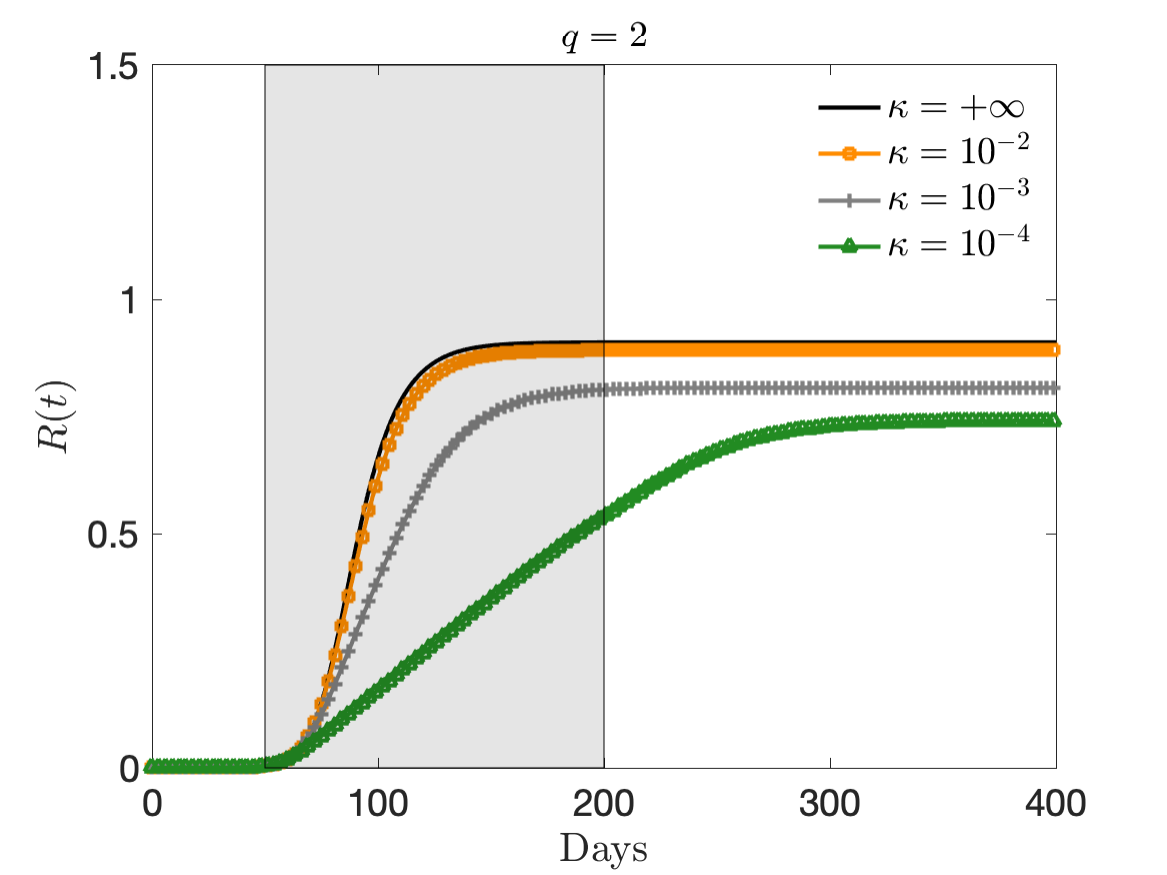}\\
\includegraphics[scale = 0.27]{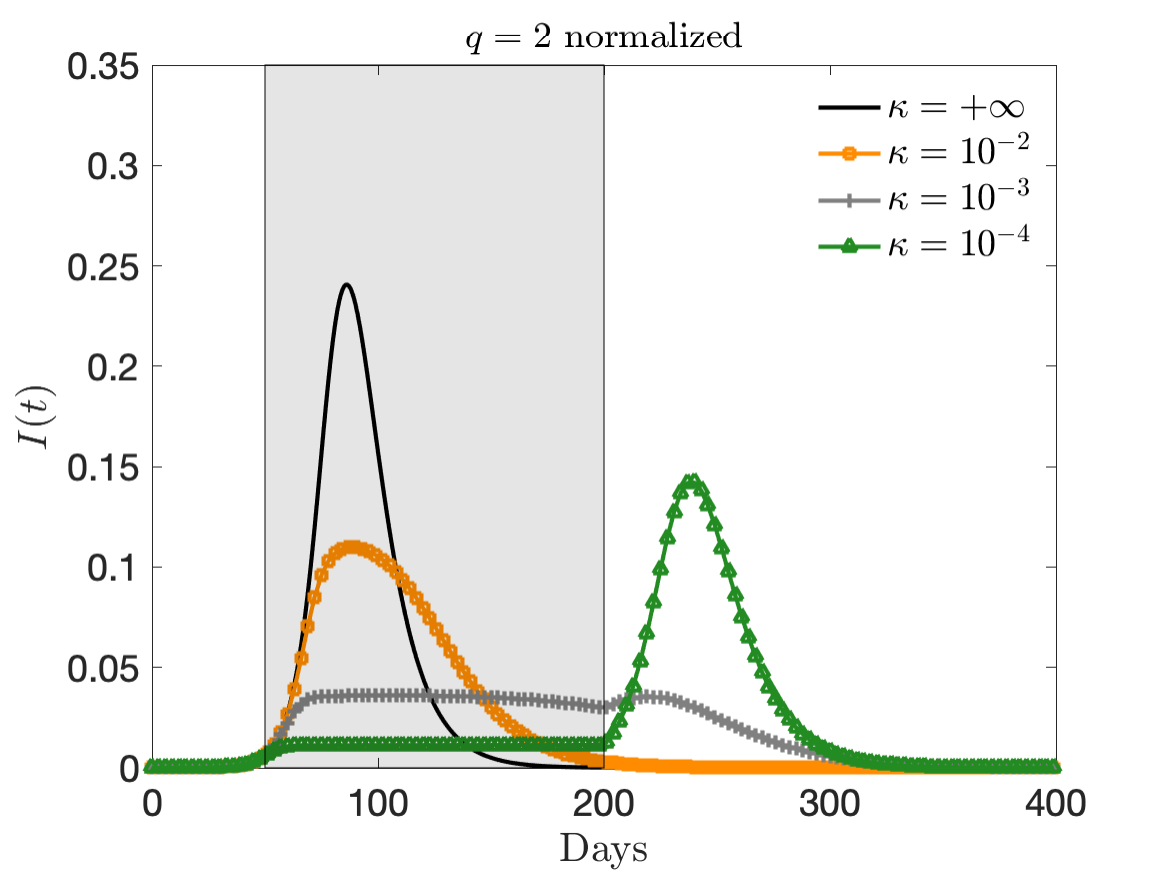}
\includegraphics[scale = 0.27]{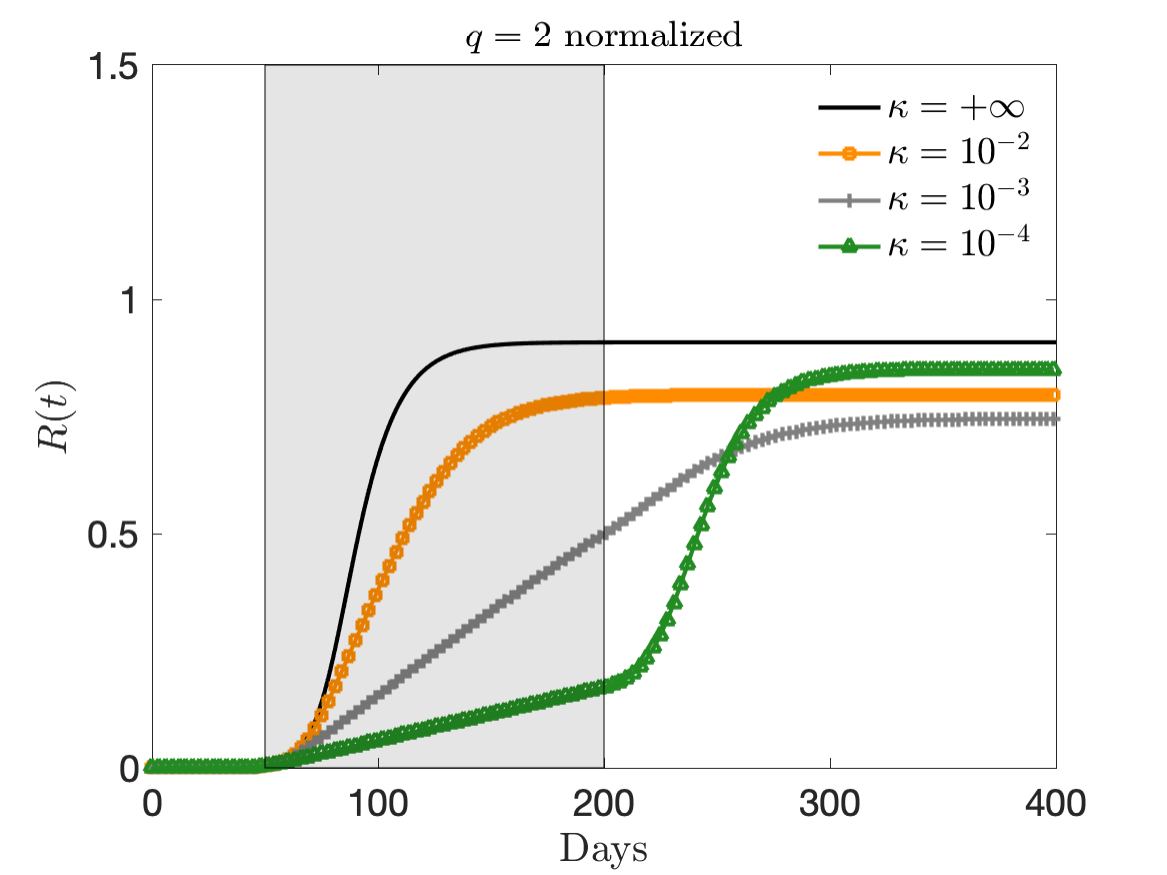}
\caption{\textbf{Test 1}. Evolution of the fraction of infected (left) and recovered (right) based on {the feedback constrained model \eqref{eq:homo} for $t\in [50,200]$, a perception function} $\psi(I)=I^q/q$, $q = 1,2$ and several penalizations $\kappa = 10^{-2}, 10^{-3}, 10^{-4}$. The choice $\kappa = +\infty$ corresponds to the unconstrained case.  {In last row the normalized case $\psi(I) = C_2 I^2/2$, with $C_2 = 12$. }}
\label{fig:test1b}
\end{figure}

In Figures \ref{fig:test1a} and \ref{fig:test1b} we report the infected and recovered dynamics based on the activation of the control in two different time frames. {In Figure \ref{fig:test1a} the activation for $t\in [50,100]$, which means that after $100$ days we suppose that all containment restrictions are cancelled. In Figure \ref{fig:test1b} we consider a larger activation time frame $t\in[50,200]$.}  

With the choice of the perception function $\psi(I)=I^q/q$, $q \geq 1$, we can observe how the control term is able to flatten the curve even if, as expected, the case $q=2$ gives rise to a weaker control action. {To make the two controls, $q = 1$ and $q = 2$, comparable for the same penalization factor $\kappa$ we also consider the normalized case, where $\psi(I) = C_q I^q/q$, and the constant $C_q>0$ is a normalization constant such that }
{
\[
\int_0^{I_{\textrm{max}}} C_1 I\; dI = \int_0^{I_{\textrm{max}}} C_2 \dfrac{I^2}{2}\; dI, 
\] 
with $I_{\textrm{max}}$ an estimate of the maximum number of infected in absence of control. In particular, in our test case, from $I_{\textrm{max}}\approx \dfrac{1}{4}$, assuming $C_1=1$ we obtain $C_2 = 12$. }

Note that, if the activation time is too short the control is not able to significantly change the total number of infected (and therefore recovered). On the other hand, by enlarging the activation time in combination with a sufficiently small penalty constant, the peak infection is not only reduced, but the total number of infected people is decreased. To achieve this, the control should be kept activated for a sufficiently long time and with the right intensity in a kind of plateau regime where there is a perfect balance between the containment effect and the spread of the disease. On the contrary, if the control is too strong, the majority of the population remains susceptible and consequently the disease will start spreading again forming a second wave after the containment policy is removed. {Similar conclusions (that may appear counterintuitive) have been shown also by other authors (see for example \cite{Britton,Lunelli}).}

\begin{figure}
\centering
\includegraphics[scale = 0.21]{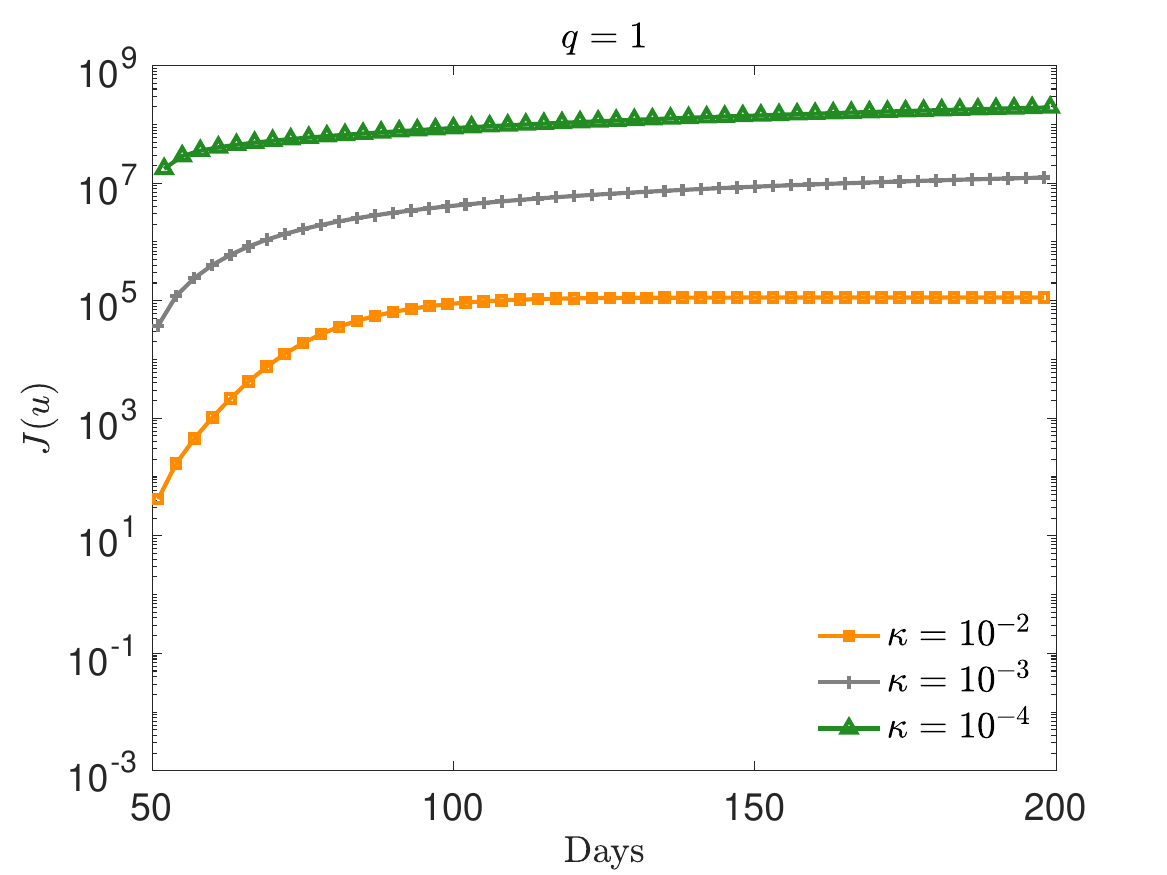}\hskip -.3cm
\includegraphics[scale = 0.21]{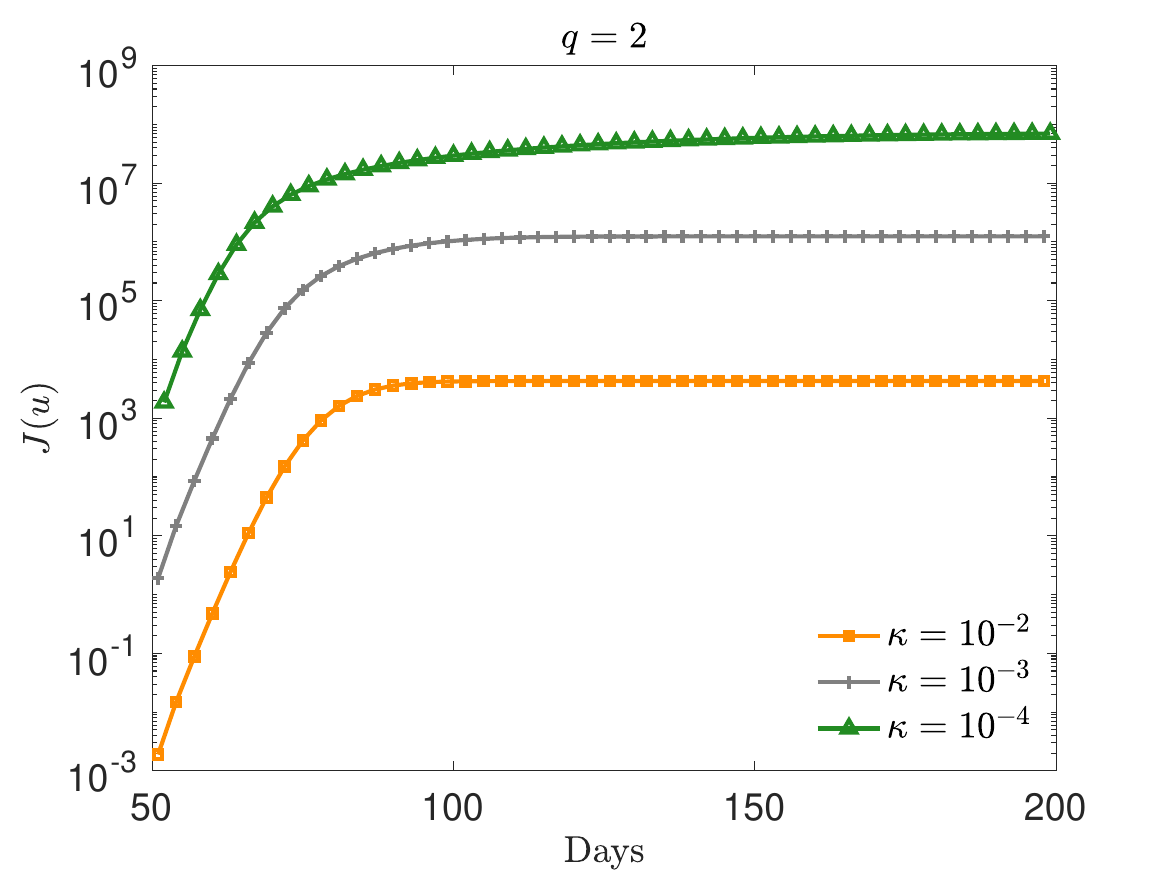}\hskip -.3cm 
\includegraphics[scale = 0.21]{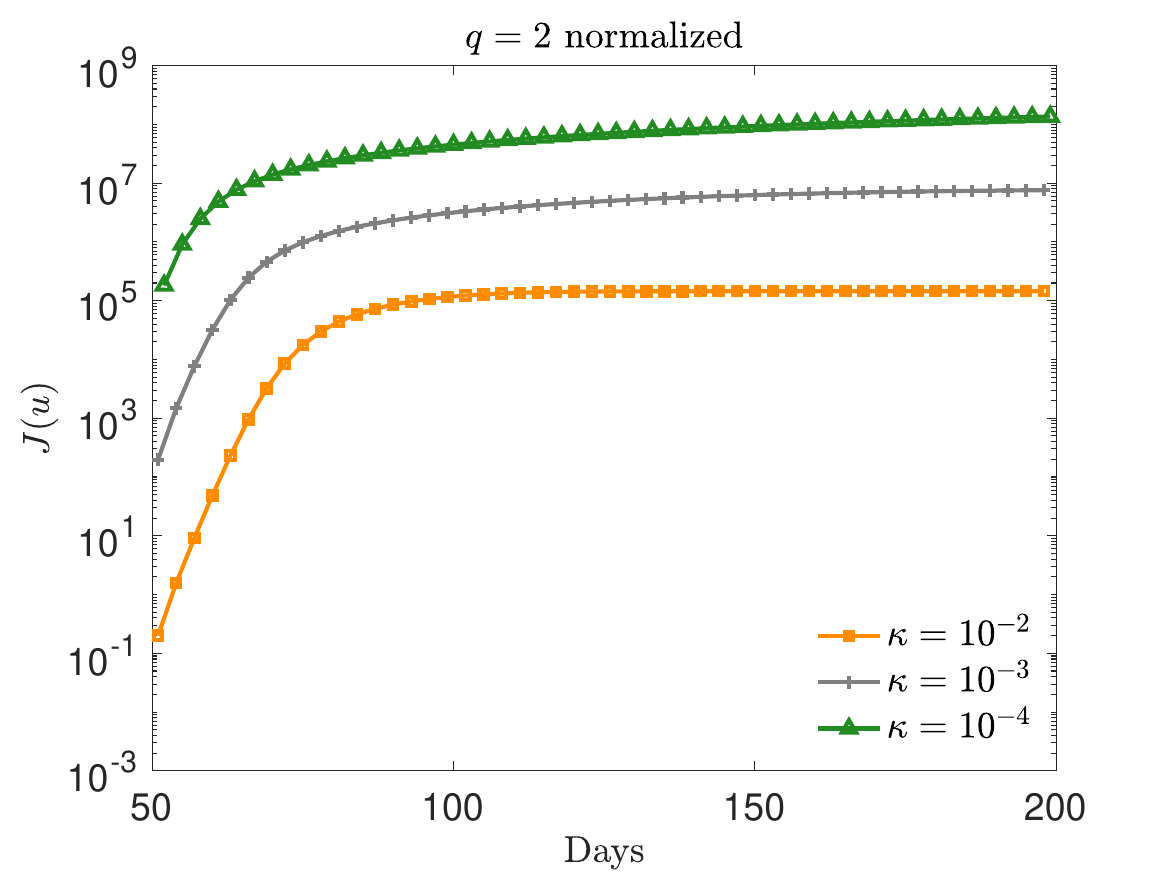} 
\caption{\textbf{Test 1}. Evaluation of the cost functional $J(u)$ for the dynamics in Figure \ref{fig:test1b}.}
\label{fig:cost}
\end{figure}

\begin{figure}
\centering
\subfigure[Infected]{
\includegraphics[scale = 0.27]{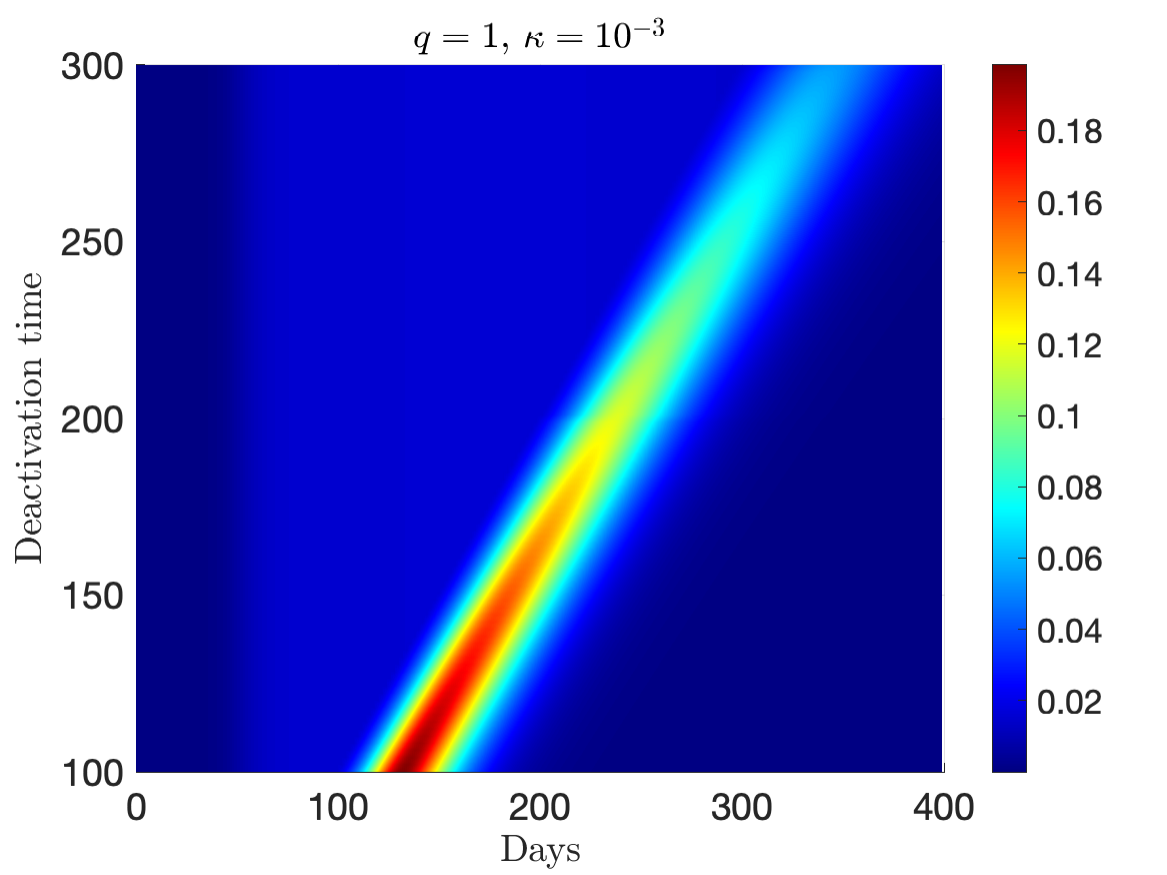}}
\subfigure[Recovered]{
\includegraphics[scale = 0.27]{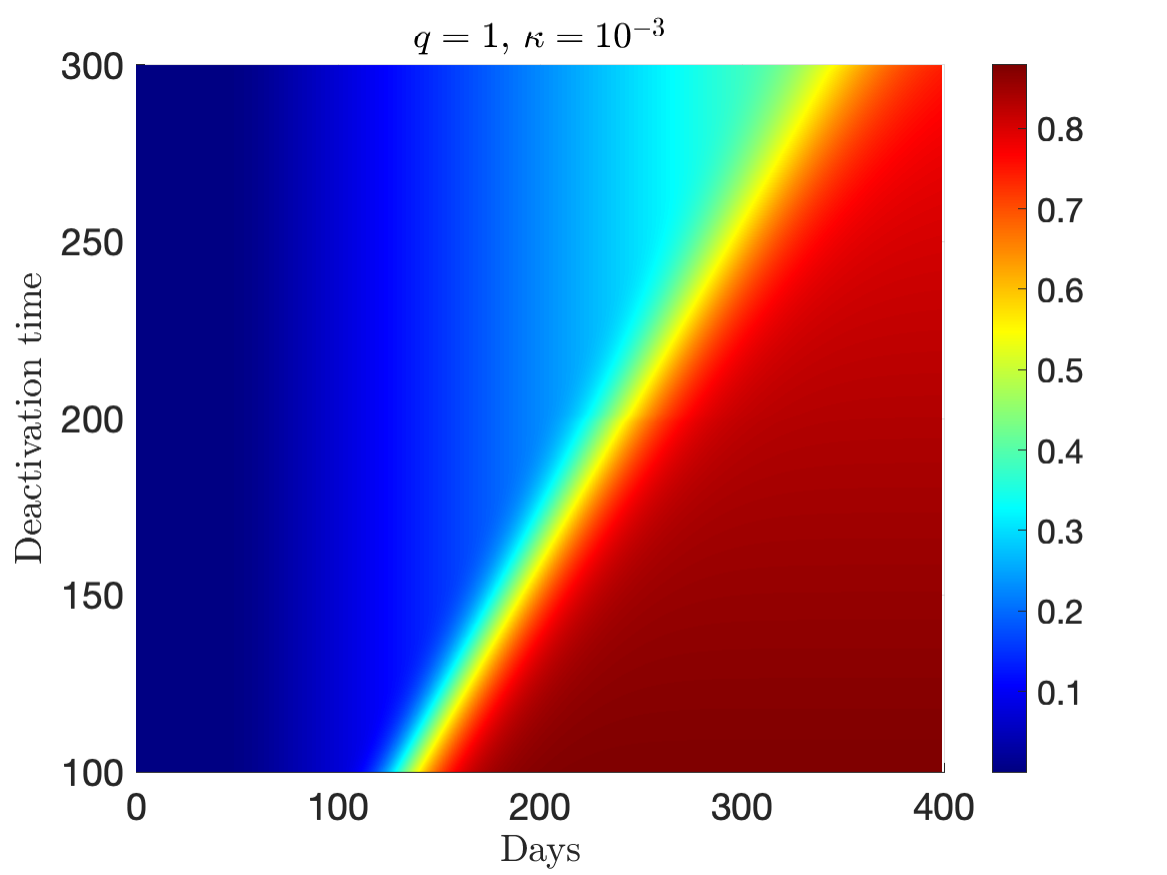}}\\
\subfigure[Infected]{
\includegraphics[scale = 0.27]{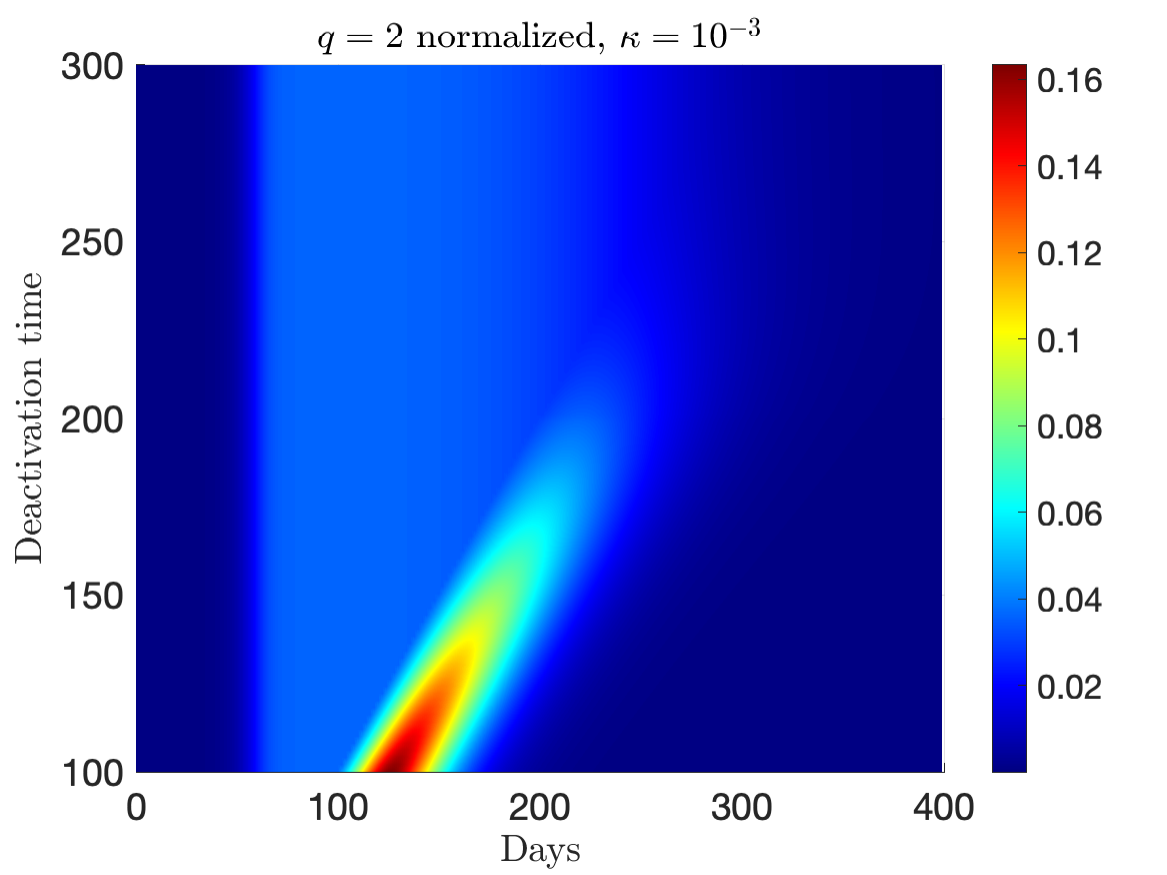}}
\subfigure[Recovered]{
\includegraphics[scale = 0.27]{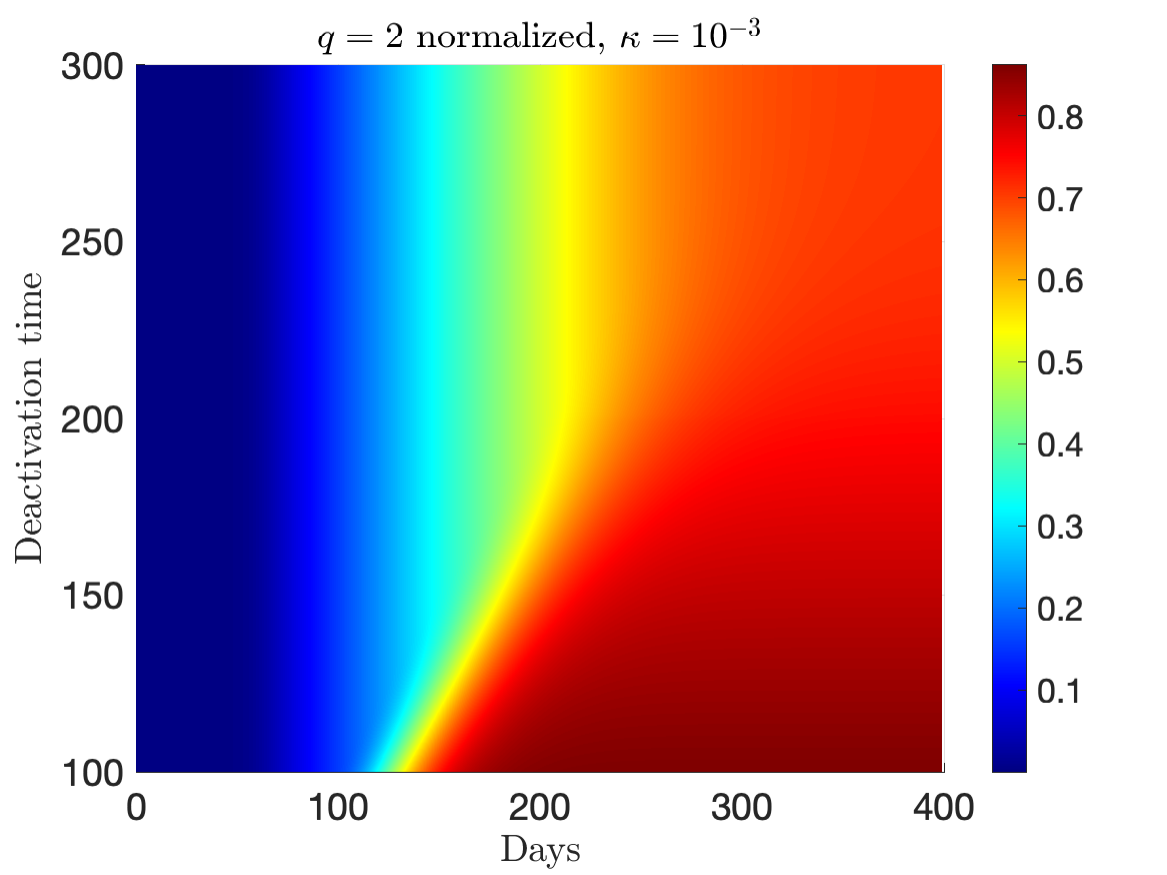}}
\caption{\textbf{Test 1}. Evolution of {the of the fraction of infected (left) and recovered (right) based on the feedback constrained model \eqref{eq:homo} for a perception function $\psi(I)=I^q/q$, $q=1$ and $q=2$ normalized, for different control actions in $[50,t_d]$ with a deactivation time $t_d \in [100,300]$} and a fixed penalization $\kappa = 10^{-3}$. }
\label{fig:deactiv}
\end{figure}

{The cost functional depends on the value of $q$ and can be evaluated summing up contributions in \eqref{eq:func_I} with  explicit form of the control given by \eqref{eq:ic}. In Figure \ref{fig:cost} the cost of the two interventions is compared. We can see how a higher cost is associated with  $q=1$ that can be obtained with the control $q=2$ only for smaller penalizations. {As expected, the normalized case $q=2$ essentially realigns the cost of interventions. Then, in Figure \ref{fig:deactiv} we compare the performance of the two controls in $[50,t_d]$ with a deactivation time $t_d \in [100,300]$.} We consider $\kappa = 10^{-3}$ for both $q=1$  and $q=2$ normalized. It can be observed that there is a minimum control horizon for both strategies, in order to avoid the onset of a second infection peak. A sufficiently long control horizon is therefore necessary to reduce the impact of the infection.}

\subsection{Test 2: Impact of uncertain data on the epidemic outbreak}

Next we focus on the influence of uncertain quantities on the controlled system with homogeneous mixing focusing on available data for COVID-19 outbreak in Italy, see \cite{Protezione}. 
The estimation of epidemiological parameters is a very difficult problem that can be addressed with different approaches \cite{Capaldi_etal,Chowell,Roberts}. 
It is worth to mention that, in the case of COVID-19, the number of infected and recovered has been largely underestimated, especially in the early phases of the epidemic, see \cite{JRGL,MKZC}. Here, we restrict ourselves to identifying the deterministic parameters of the model through a suitable data fitting procedure, considering possible deviations due to such underestimates as part of the subsequent uncertainty quantification process.
 
\subsubsection*{The data fitting process}
{In details, we have adopted a two-level approach in estimating the parameters in absence of uncertainties. In the phase preceding the lockdown we estimated the epidemic parameters in an unconstrained regime, where we assumed no social containment procedure was activated. This estimate was then kept int the subsequent lockdown phase where we estimated as a function of time the value of the control penalty parameter. Both these two calibration steps were analyzed under the assumption of homogeneous mixing, therefore model \eqref{eq:homo} has been used in lockdown phase.} 

{
First, we estimated the parameters $\beta,\gamma>0$ in the time interval $t \in [t_0,t_u]$ solving a lest square problem based on the minimization of the relative $L^2$ norm of the difference between the reported number of infected $\hat{I}(t)$ and recovered $\hat{R}(t)$, and the theoretical evolution of the unconstrained model whose solution at time $t\ge 0$ is indicated with $I(t)$ and $R(t)$. More precisely, we considered the following minimization problem}

{
\be\label{eq:min1}
\min_{\beta,\gamma \in \mathbb R_+} 
\left[(1-\theta) \| I(t)-\hat I(t) \|_{L^2([t_0,t_u])}+\theta \| R(t)-\hat R(t) \|_{L^2([t_0,t_u])} 	\right], 
\ee
where $\theta \in [0,1]$ and $\| \cdot \|_{L^2([t,s])}$ denotes the relative $L^2$ norm over a time horizon $[t,s]$.}

{Problem \eqref{eq:min1} has been solved with the constraints $\beta \in [0,1]$ and $\gamma \in \left[\dfrac{1}{24},\dfrac{1}{10}\right]$. Indeed, according to several studies the time to viral clearance during the early phase of the epidemic, corresponding to the time from the first positive test to the first negative test, can approximately span in average from 10 to 24 days, see \cite{Chen_etal,GattoPNAS,LavezzoCrisanti}. }

\begin{figure}
\centering
\includegraphics[scale = 0.35]{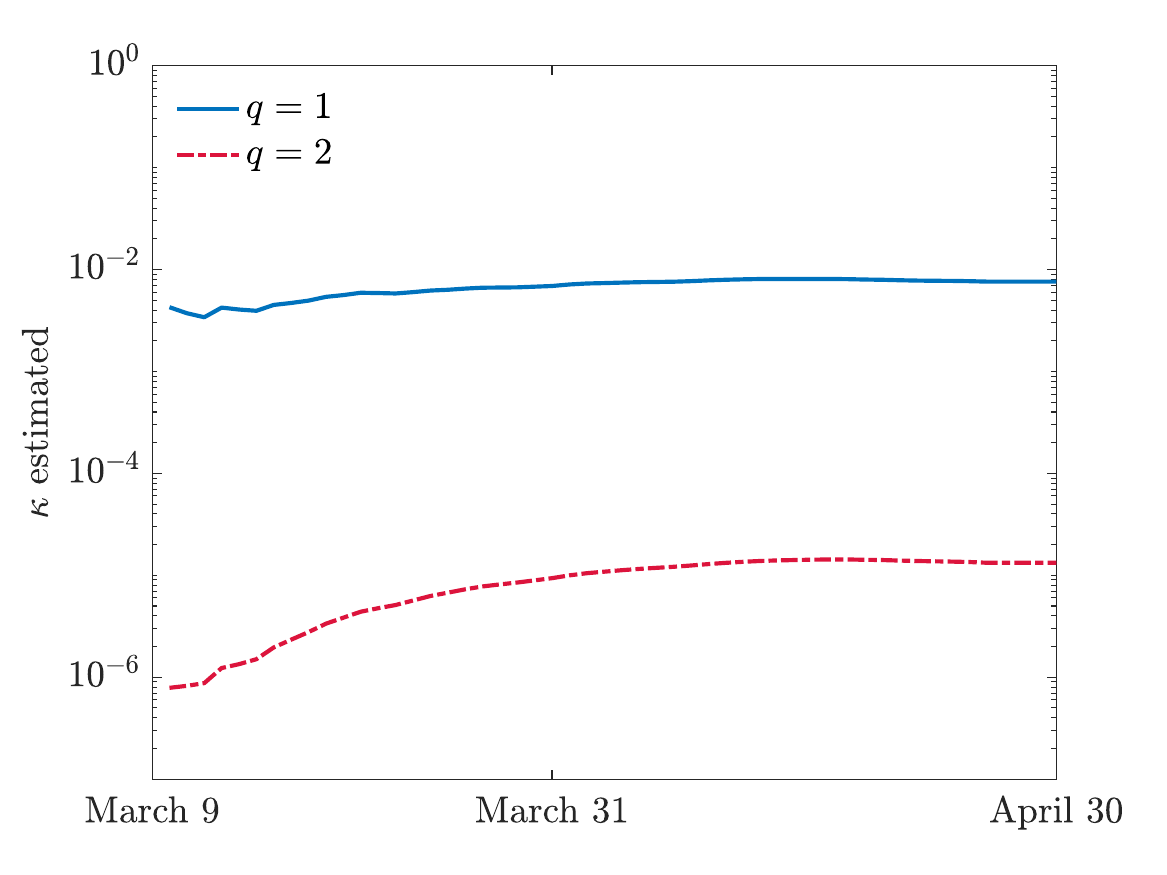}
\caption{\textbf{Test 2}. Estimated control penalization terms over time from reported data on number of infected and recovered in the case of COVID-19 outbreak in Italy. }
\label{fig:datak}
\end{figure}

{
At the end of the above optimization process we obtained the values $\beta_e \approx 0.31$, $\gamma_e \approx 0.049$ computed by averaging the optimization results with $\theta = 10^{-2}$ and $\theta = 10^{-6}$. The choice of a small value for $\theta$ is due to the low reliability on the recovered data at this early stage.}

{
Next, we estimate the penalization $\kappa = \kappa(t)$ in time by solving \eqref{eq:homo}, in the lockdown time interval $t \in (t_u,t_c]$ and for a sequence of time steps $t_i$, the corresponding least square problems in $[t_i - k_l, t_i + k_r]$ where $k_\ell,k_r\ge 1$ are integers, and where we fix the values $\beta_e,\gamma_e$ estimated in the first optimization step. In details, we solve the following minimization problem
\be\label{eq:min2}
\min_{\kappa(t_i) \in \mathbb R_+} \left[(1-\theta) \| I(t)-\hat I(t) \|_{L^2([(t_i-k_\ell,t_i+k_r])}+\theta \| R(t)-\hat R(t) \|_{L^2([t_i-k_\ell,t_i + k_r])}  \right], 
\ee
over a window of seven days corresponding to $k_\ell = 3$ and $k_r = 4$  for regularization along one week of available data.} Both minimization problems \eqref{eq:min1}-\eqref{eq:min2} have been solved testing various numerical methods in combination with adaptive solves for the systems of ODEs. The results have been obtained using Matlab functions \texttt{fmincon} in combination with \texttt{ode45}. The available data start on February 24 2020, when moderate social restrictions were enforced by the Italian government, and since the lockdown started on March 9 2020, thus we considered $t_u - t_0 = 14$ (days).

The corresponding time dependent values for the expected penalization for a perception function $\psi(I)=I^q/q$ are reported in Figure \ref{fig:datak}. After an initial adjustment phase the penalty terms converge towards a constant value that we can assume as fixed in predictive terms for future times in a lockdown scenario. This is consistent with a situation in which society needs a certain period of time to adapt to the lockdown policy. 

\begin{remark}
{Finally, we remark that the data fitting procedure is easily generalizable to epidemic models with additional compartments \cite{GattoPNAS,Giordano}. Let us denote with $C_i(t)$ the compartments that are related to the reported data $\hat C_j(t)$, $j=1,\ldots,h$, as the number of actual cases, hospitalized, deaths, etc. In the first step one minimizes a weighted $L^2$ norm in the form
\be
\label{eq:ming}
\min_{{\bf p}\in {\cal P}}
\sum_{j=1}^h w_j \| C_j(t)-\hat C_j(t)\|_{L^2[t_0,t_u]}, 
\ee
where $w_j \in [0,1]$ are suitable weights such that $\sum_{j=1}^h w_j=1$ and ${\bf p}$ is the vector of the model parameters that need to be estimated. Since the problem may admit multiple minima leading to unrealistic solutions one usually perform the above optimization process under constraints on the range of values ${\cal P}$ of some parameters such as recovery rate, incubation period, etc.}
   
{An important difference, compared to a simple SIR compartmentalization, is the presence of compartments that are not data-driven as exposed, pre-symptomatic, asymptomatic, etc. that prevent the realization of the data fitting since their values are unknown. A way to overcome this difficulty is to solve the differential model starting from an unknown time $t^* < t_0$ using as initial data the presence of a single individual in the first compartment promoting the infection (for example the exposed) and zero individuals in all other compartments. The idea is to simulate the early phase of the epidemic, by optimizing also the initial time $t^*$ in problem \eqref{eq:ming}.}

{After this, in the second step one considers the corresponding feedback controlled model (see Appendix \ref{sec:gen}) and solves the minimization problem
\be\label{eq:min2g}
\min_{{\boldsymbol \kappa}(t_i)} \sum_{j=1}^h w_j \| C_j(t)-\hat C_j(t)\|_{L^2([t_i-k_\ell,t_i + k_r])}, 
\ee  
where ${\boldsymbol \kappa}(t_i)$ is the vector of penalization terms that need to be estimated. Even in this case, assumptions on the range of values of the penalization terms may be necessary to avoid unrealistic solutions. Finally, we underline that most epidemic models are not data-driven, but one can always assume that the total number of reported cases underestimates the actual number of cases and perform the above data fitting to obtain a lower bound on the evolution of the epidemic. Then including a suitable data uncertainty, as in the present work, allows the recovery of realistic scenarios for the pandemic progression. }
\label{rk:df}
\end{remark}

\subsubsection*{Introducing data uncertainty}

{To account for the impact of uncertainties in the data and parameters we then consider a two-dimensional random variable $\z = (z_1,z_2)$ with independent components such that 
\be\label{eq:ir_z1}
i(\z,0) = i_0(1+\mu z_1),\qquad r(\z,0) = r_0(1+\mu z_1),\qquad \mu >0, 
\ee
and
\be\label{eq:beta_z2}
\gamma(\z) = \gamma_e + \alpha_\gamma z_2,\qquad \beta(\z) = \beta_e - \alpha_\beta z_2, \qquad \alpha_\gamma,\alpha_\beta >0,
\ee
where $z_1,z_2$ are chosen distributed as symmetric Beta functions in $[0,1]$, $i_0$ and $r_0$ are the initial number of reported cases and recovered taken from \cite{Protezione} on February 24, 2020.  
The choice of a Beta distribution for $p(\z)=p_1(z_1)p_2(z_2)$ is coherent with other authors \cite{Roberts,X}.} However, different probability distribution functions may be considered if additional information on the nature of the uncertainties are available.

\begin{figure}
\centering
\includegraphics[scale = 0.28]{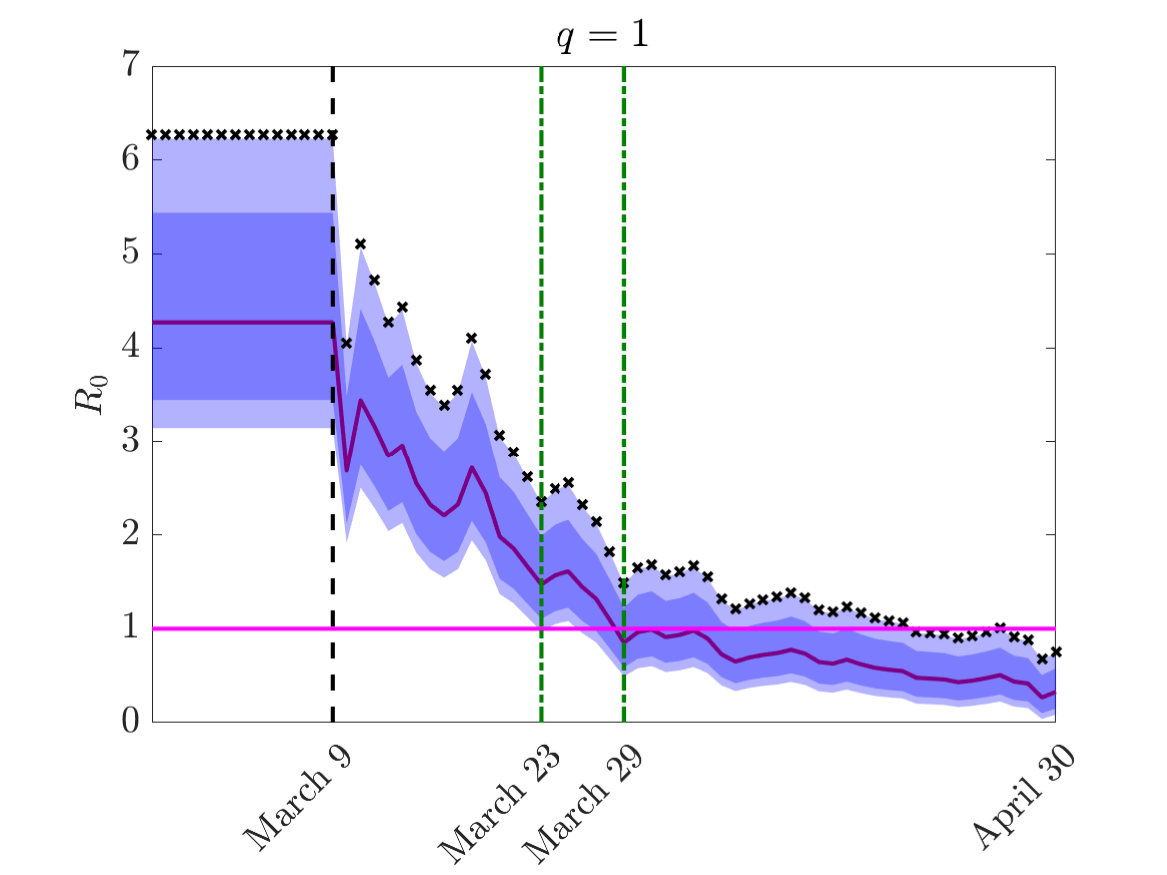}
\includegraphics[scale = 0.28]{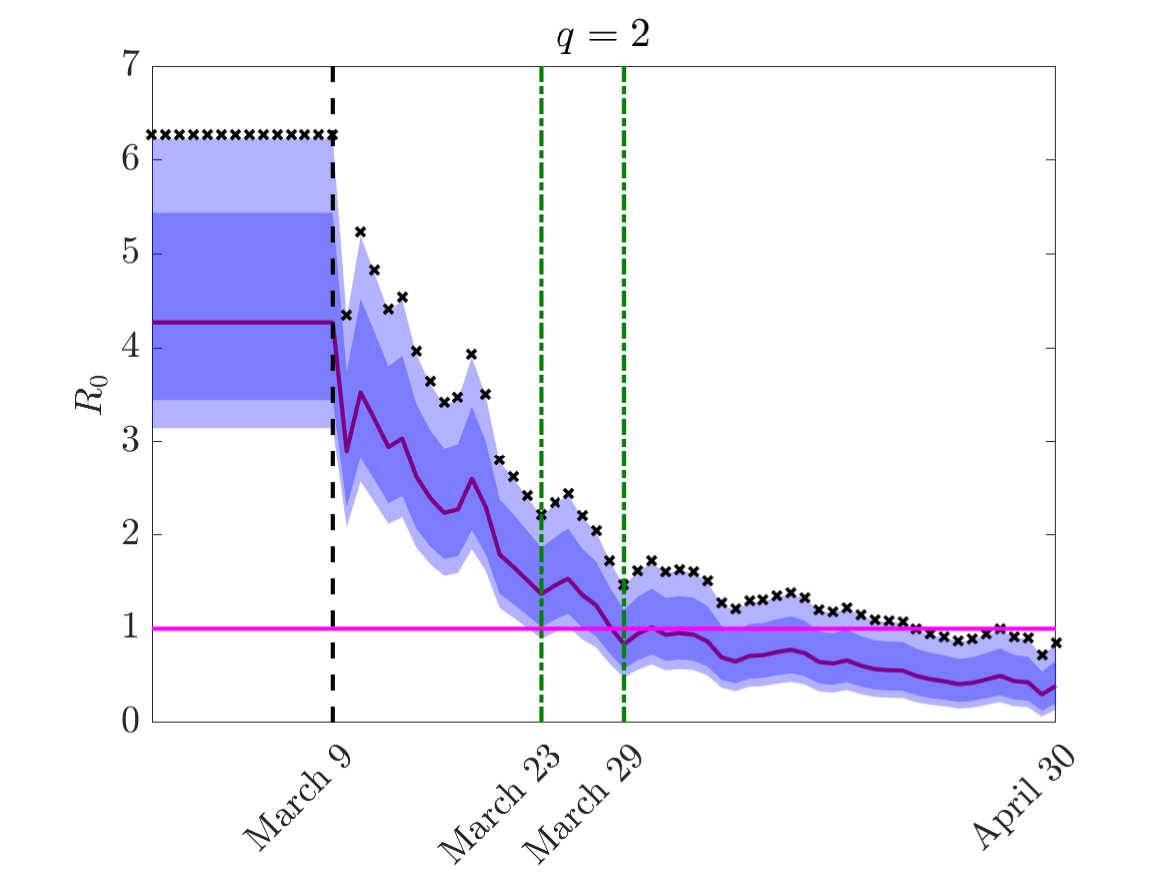}
\caption{\textbf{Test 2}. Estimated reproduction number $R_0$ {from the feedback controlled model with uncertain  data \eqref{eq:beta_z2} for a perception function $\psi(I)=I^q$, $q=1$ (left) and $q=2$ (right) together with the confidence bands. We mark with dash-dotted green lines the days in which the lower $95\%$ band and the expected $R_0$ fell below one, and with $x$-markers the estimated reproduction number relative to data fitting. }}
\label{fig:dataR}
\end{figure}

{
It should be noted that the estimated reproduction number, computed in the first optimization step, corresponds to $R^e_0=\beta_e/\gamma_e \approx 6.3$, which is at the upper limit of the values reported in the literature for COVID-19 \cite{JRGL,Zhang_etal,Liu}. Being aware of the limitations of the data fitting on reported data, to consider a more realistic range of values we assumed a stochastic dependence in $\beta$ and $\gamma$ {taking into account the faster recovery of asymptomatic individuals \cite{GattoPNAS} and the fact that asymptomatic individuals might be slightly less infectious than symptomatic cases \cite{Lancet}.}
In the simulations we take $\alpha_{\beta} = 0.03$, $\alpha_\gamma = 0.04$ and $z_2 \sim B(2,2)$ in \eqref{eq:beta_z2}. Under this assumptions the reproduction number covers a range of values approximatively in $[3.13,6.27]$ with an expected value around $4.25$.} 

{The reproduction number in the feedback controlled model is estimated from
\begin{equation}\label{eq:R0_z2}
R_0(z_2,t) = \dfrac{\beta(z_2) - u(t) \chi(t >\bar t)}{\gamma(z_2)}.
\end{equation}
In \eqref{eq:R0_z2} the time $\bar t$ is the lockdown time, in the case under study March 9th, and $\chi(\cdot)$ the indicator function. The feedback control $u(t)$ is defined from  \eqref{eq:u_z} in the case of homogeneous mixing and assuming
$\mathcal R[S(\cdot,t)I(\cdot,t)\psi'(I(\cdot,t))]=S(\z_0,t)I(\z_0,t)^{q}$ as in \eqref{eq:R2}, where $I(\z_0,t)$ is the total number of infected reported at time $t$. This leads to
\begin{equation}
u(t)= \frac1{\kappa(t)} S(\z_0,t)I(\z_0,t)^{q},
\label{eq:r0e}
\end{equation}
in agreement with the fact that the confinement restrictions have been implemented accordingly to the reported data.}
Note that, otherwise the action of the uncertainty translates into the control and leads to the unrealistic effect that the largest is the number of unreported infected the largest is the action of the control in the system. 

{In Figure \ref{fig:dataR} we report the expected value of $R_0$ together with the $95\%$ and $50\% $ confidence levels with respect to the variable $z_2$. 
The estimated reproduction number relative to data fitting is reported with $x$-marked symbols and represents an upper bound for $R_0(z_2,t)$. The results show that the $R_0$ reproduction number, thanks to the containment actions, has been drastically reduced and its expected value fell below one between March 23rd and March 29th for both $q = 1$ and $q = 2$. After March 29th the observed $R_0$ is stably below unity. Note that, in both controls the results are very similar, without any need of renormalization for $q=2$ due to the data fitting process.}

\begin{figure}
\centering
\includegraphics[scale = 0.28]{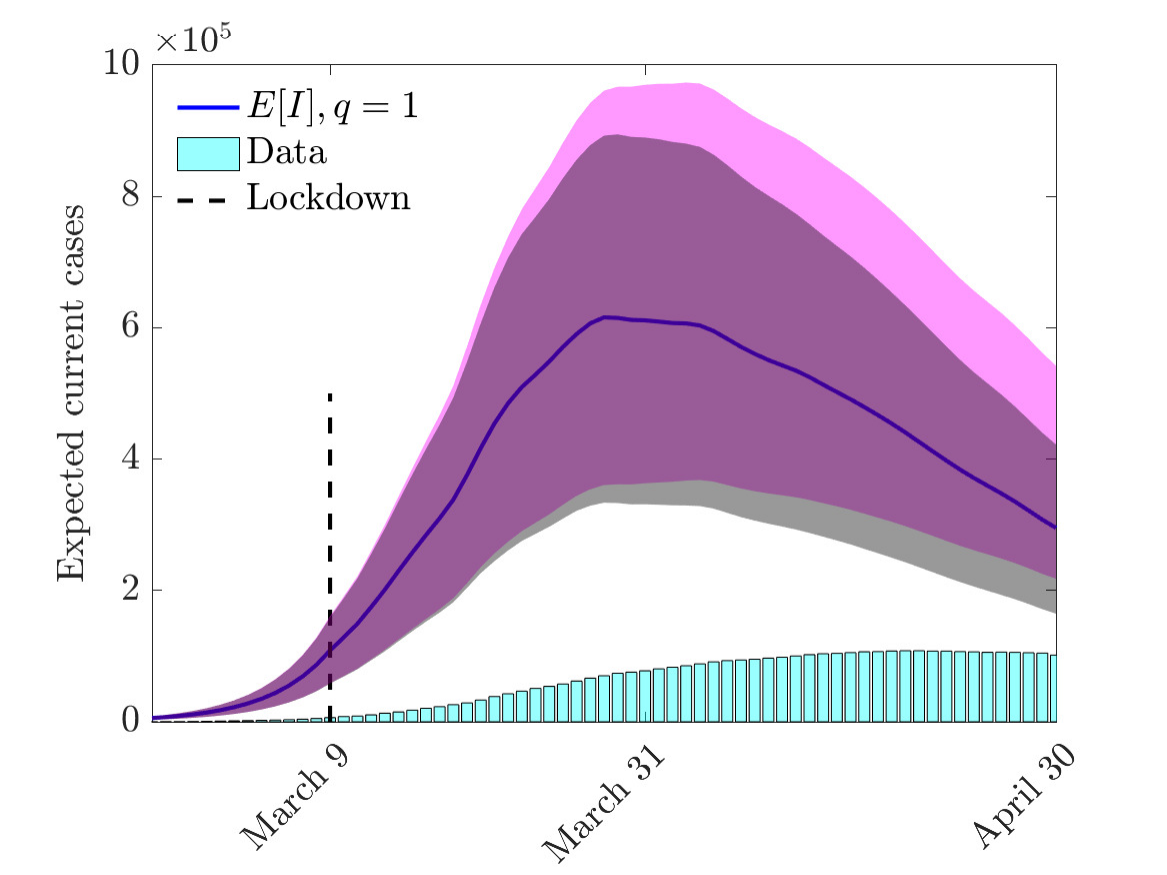}
\includegraphics[scale = 0.28]{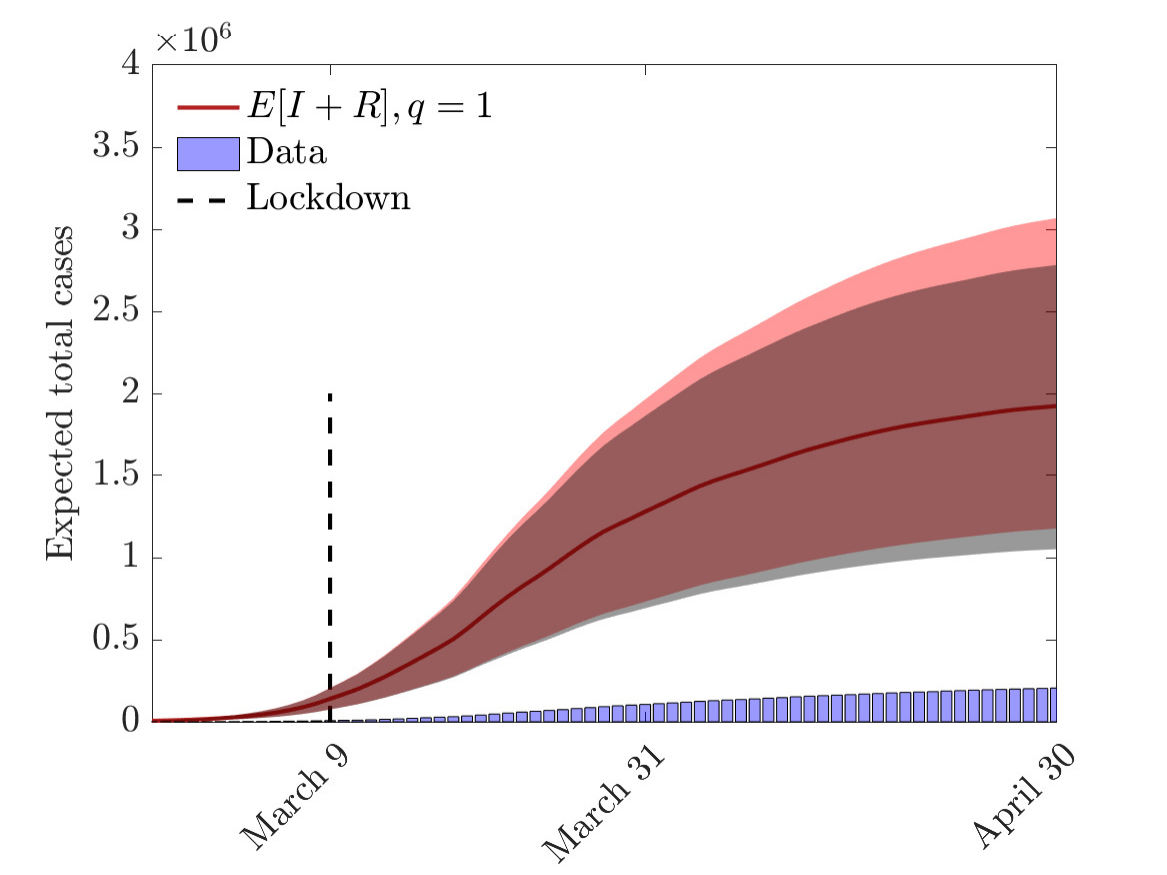}
\caption{\textbf{Test 2}. {Evolution of expected current cases (left) and of the expected total cases (right) and their 95\% confidence bands with respect to $z_1$ (shaded color) and $z_2$ (shaded gray) for the feedback controlled model with perception function $\psi(I)=I$, and uncertain initial data \eqref{eq:ir_z1}-\eqref{eq:beta_z2}}.}
\label{fig:data}
\end{figure} 

{Next we considered the evolution of the uncertain number of infected. In the following we assumed a strongly underestimated initial number of infected (including asymptomatic), taking {$\mu = 50$} so that the reported infected along the time horizon of the simulation represent approximately a $20\%$ portion of the total infected persons computed by the feedback controlled model. 
This is in accordance with the WHO findings that around $80\%$ of infected are asymptomatic\footnote{Q\&A: Similarities and differences COVID-19 and influenza. \\
\texttt{ https://www.who.int/news-room/q-a-detail/q-a-similarities-and-differences-covid-19-and-influenza}} and with the results of preliminary serological campaigns promoted in Italy\footnote{Preliminary results on the seroprevalence of SARS-CoV-2 in Italy: \\ \texttt{https://www.istat.it/it/files//2020/08/ReportPrimiRisultatiIndagineSiero.pdf}}.}

{
In Figure \ref{fig:data} we represent the evolution of the expected value of the number of infected obtained by the controlled model with perception function $\psi(I)=I$  in presence of uncertain contact and recovery rates \eqref{eq:beta_z2} and initial uncertain data \eqref{eq:ir_z1} assuming $z_1 \sim B(40,40)$ and $z_2\sim B(2,2)$. 
We represent the expected values of the current cases (left) and of the total cases (right) along with the $95\%$  confidence level with respect to the variables $z_1$ and $z_2$. The shaded color band is relative to the variability in $z_1$ 
whereas the shaded gray band is relative to the variability in $z_2$.
The bars below the graph are the reported values of the number of infected on which the model has been calibrated. The results with $q=2$ do not highlight significant differences with respect to the case $q=1$ and therefore have been omitted. }

\subsection{Test 3: The effect of social contacts in the population}

\begin{figure}
\centering

\includegraphics[scale = 0.4]{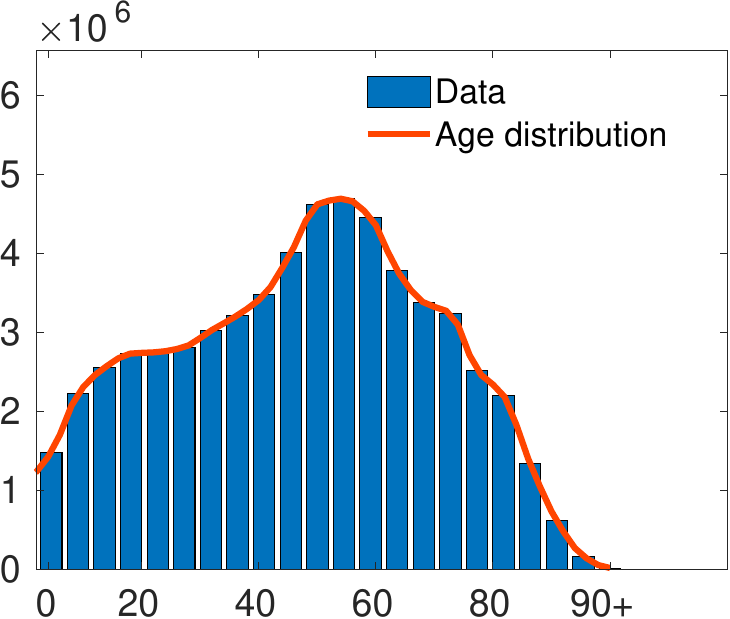}\hskip .8cm
\includegraphics[scale = 0.4]{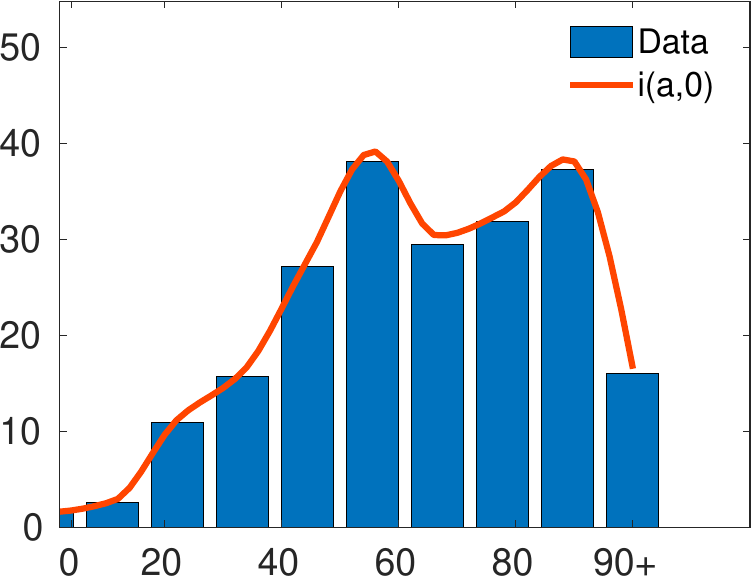}
\caption{\textbf{Test 3}. Distribution of age in Italy (left) and distribution of infected (right) together with the corresponding continuous approximations$^2$.}
\label{fig:age}
\end{figure}

\begin{figure}
\centering
\includegraphics[scale = 0.38]{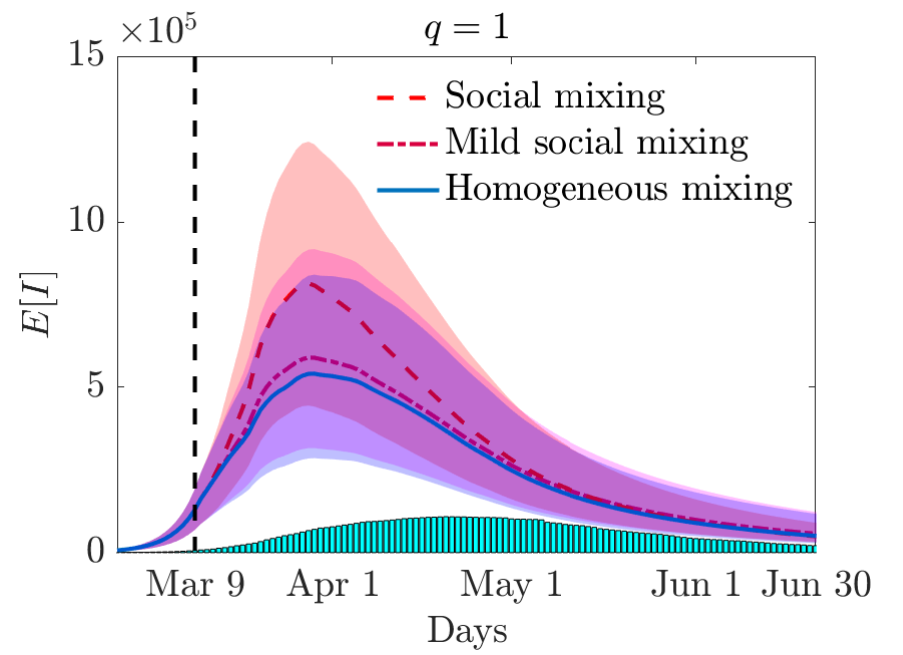}
\includegraphics[scale = 0.38]{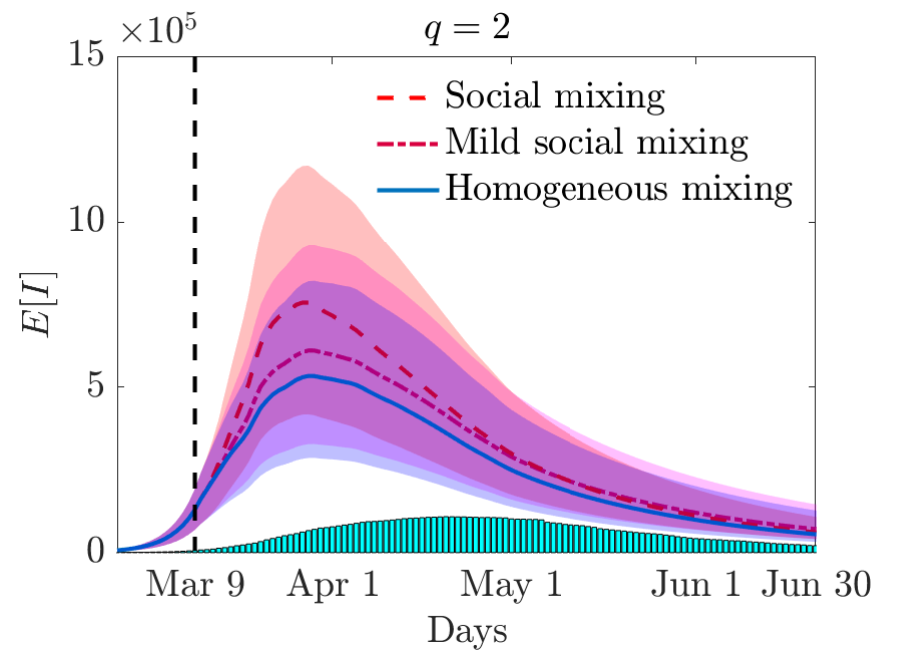}\\
\caption{\textbf{Test 3}. Expected number of infected in time for {the perception function $\psi(I)=I^q$, $q=1$ (left) and $q=2$ (right) and a constant recovery rate together with the confidence bands for homogeneous mixing ($\xi=0$), mild social mixing $(\xi = 0.75)$ and full social mixing ($\xi = 1$)}.}
\label{fig:predict}
\end{figure}

Subsequently, we analyze the effects of the inclusion of age dependence and social interactions in the above scenario. {More precisely we consider the social interaction rate $\beta = \beta(a,a_*)$, recovery rate $\gamma=\gamma(a)$ and uncertain initial number of infected.
 These functions were normalized using the estimated parameters $\beta_e$ and $\gamma_e$ in accordance with
\be\label{eq:bg}
\beta_e =\int_{\LL\times \LL} \beta(a,a_*)f(a)f(a_*)\,da\,da_*,\qquad \gamma_e = \int_{\LL} \gamma(a)f(a)\,da,
\ee
where $f(a)$ is the age distribution with $\Lambda = [0,a_{\rm max}]$, $a_{\max}=100$.}

{
The age dependent social interaction rate $\beta(a,a_*)$ is defined as follows,
\be\label{eq:beta_aa}
\beta(a,a_*) = (1-\xi)\beta_e + \xi \beta_{\rm social}(a,a_*),  
\ee
where $0\leq\xi\leq 1$, thus for $\xi=0$ we recover the {homogeneuos mixing}, whereas for $\xi = 1$ we have a full {social mixing} behavior.}
{
The social interaction function, $\beta_{\rm social}(a,a_*)$, accounts for the interactions due to specific activities $\mathcal A=\{\textrm{Family,\ Education,\, Profession}\}$ and is defined by \eqref{eq:betasocial} in Appendix \ref{app:data}. However, since after the discovery of the first case (February 21), schools, and many places of aggregation were closed in most regions of Northern Italy, we assume that $\beta_E$ is $0$ from February 24 onwards, while $\beta_P$ is reduced by a factor one-half from March 9 onwards.}

{
The choice of age dependent recovery rate $\gamma(a)$ involves a certain degree of arbitrariness, nevertheless it is reasonable to account such heterogeneity as observed in different studies \cite{GammaAge1,GammaAge2,GammaAge3}. In order to account fast recovery rate of young people, and slow recovery of the eldest we chose $\gamma(a)$ to be constant up to a specific age $a_0$ and then a decreasing function of the age. We express mathematically the recovery rate as 
\be\label{eq:gamma_a}
\gamma(a) = C_\gamma\left(\chi(a\leq a_0) + (1-\chi(a\leq a_o))e^{-r(a-a_o)}\right),  
\ee
with $r=4.5, a_o=20$ and $C_\gamma\in\mathbb{R}_+$ such that \eqref{eq:bg} holds.}

\begin{figure}
\centering
\includegraphics[scale = 0.38]{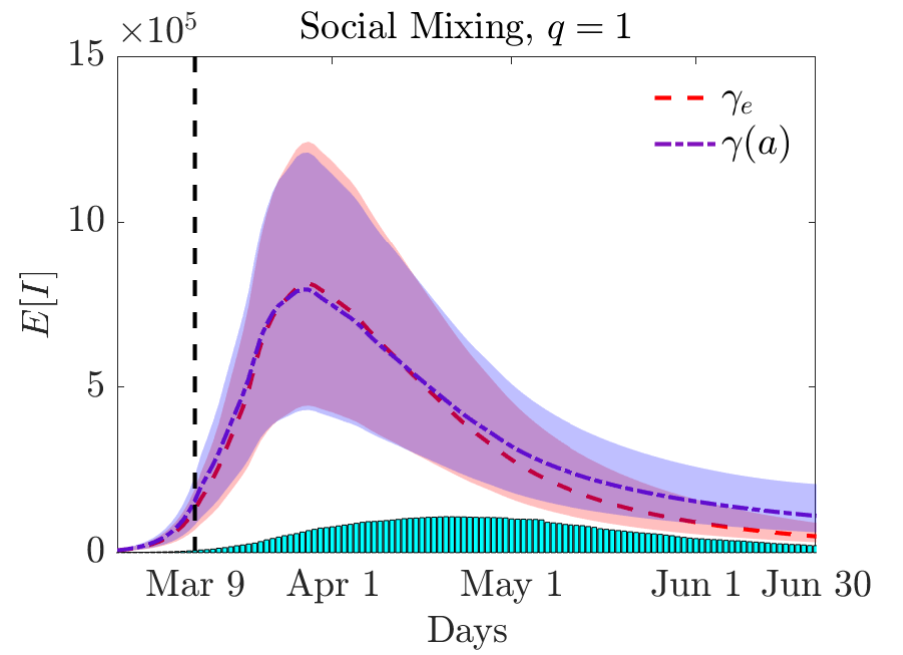}
\includegraphics[scale = 0.38]{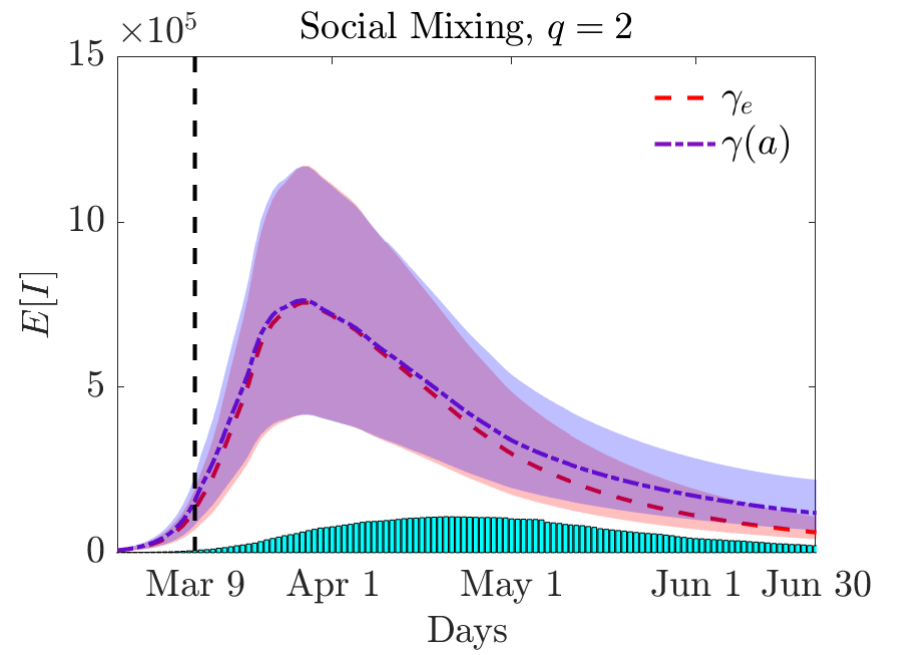}
\\
\includegraphics[scale = 0.38]{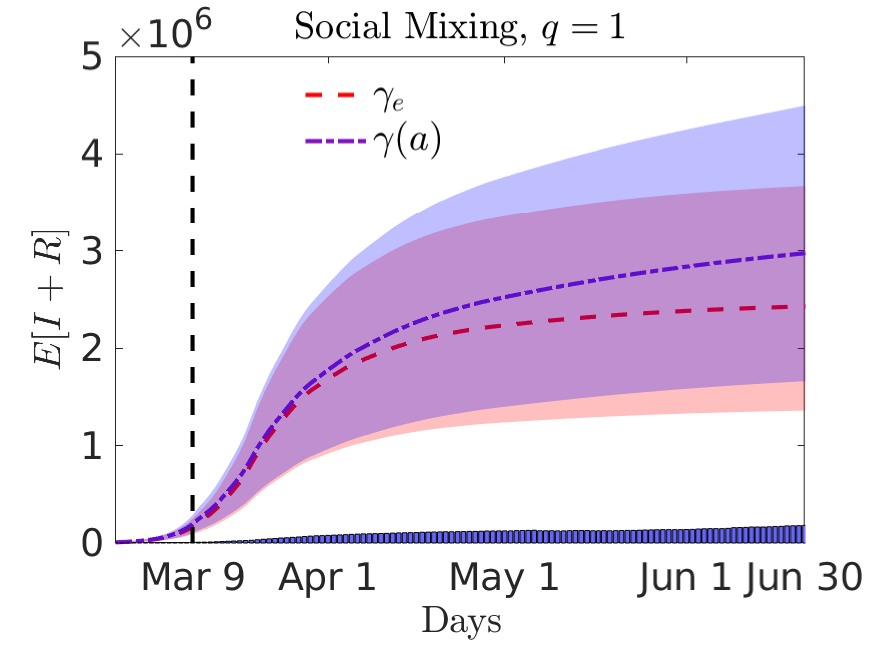}
\includegraphics[scale = 0.38]{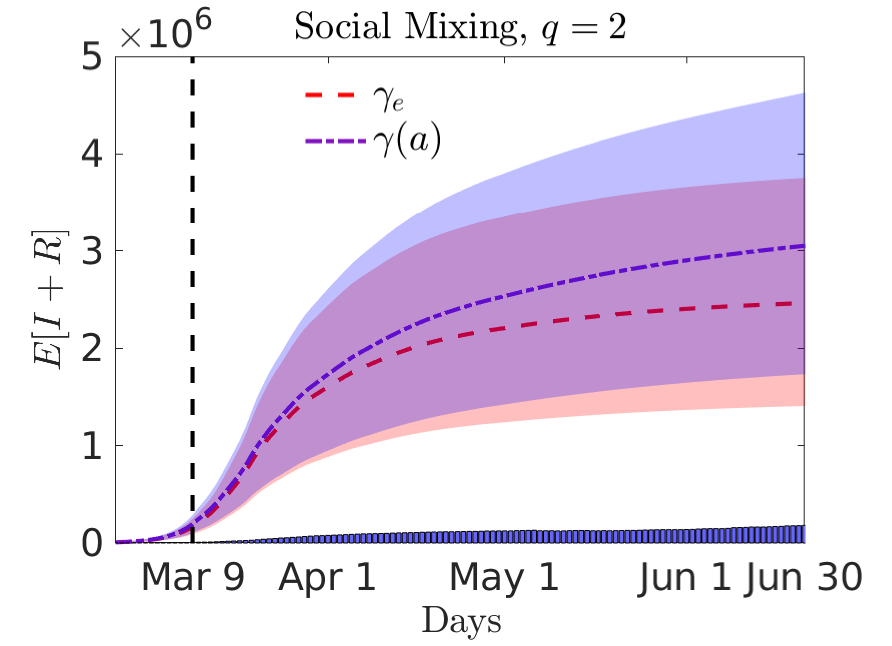}
\caption{\textbf{Test 3}. Expected number of infected and total cases of infected and recovered in time for a perception function $\psi(I)=I^q$, $q=1$ (left) and $q=2$ (right) together with the confidence bands for the social mixing scenario with age-independent or age-dependent recovery rate.}
\label{fig:predict2}
\end{figure}

We divided the computation time frame into two zones and used different models in each zone, in accordance with the policy adopted by the Italian Government. The first time interval defines the period without any form of containment from 24 February to 9 March, the second the lockdown period from 9 March. In the first zone we adopt the uncontrolled model with homogeneous mixing. Hence, in the second zone we compute the evolution of the feedback controlled age dependent model \eqref{eq:SIR_z}-\eqref{eq:u_z} with matching (on average) interaction and recovery rates accordingly to \eqref{eq:bg} and with the estimated value of the control penalization  
$\kappa(t)$ as reported in Figure \ref{fig:datak} until April 30. After April 30 the computation advances in time using as penalization term the constant asymptotic value $\bar\kappa$ reached by $\kappa(t)$. 
The initial values for the age distributions of susceptible and infectious individuals  are shown in Figure \ref{fig:age} in agreement with reported data\footnote{Source ISTAT (\texttt{https://www.istat.it}) and Istituto Superiore di Sanit\`{a} ({\texttt{https://www.epicentro.iss.it}})}. 

{
In Figure \ref{fig:predict} we report the results of the expected number of infected with the related confidence bands in case of homogeneous mixing and different levels of social mixing ($\xi = 0.75,~ \xi =1$) for the constant recovery rate $\gamma_e$. The homogeneous mixing hypothesis leads to a lower estimate of the maximum number of infected and shows a slower decay over time in the case $q=1$, whereas for $q=2$ the decay is comparable.}
{
In Figure \ref{fig:predict2} we compare the case of constant and age-dependent recovery rates for the social mixing scenario. The expected number of infected are shown in the top row, bottom row depicts the total number of recovered and infected people. The heterogeneity of the recovery rate makes the epidemic more persistent and causes an increase in the total number of cases with respect to the homogeneous recovery rate.} 
{
Finally, in Figure \ref{fig:gamma_a} we report the expected age distribution of infectious individuals in time for $q=1$. It is evident the effect of the age dependent recovery rate in the increase of cases among the oldest part of the population. Note that, this effect is partially compensated by the strength of the social mixing which reduces interaction among the elderly.}

\begin{figure}
\centering
\includegraphics[scale = 0.4]{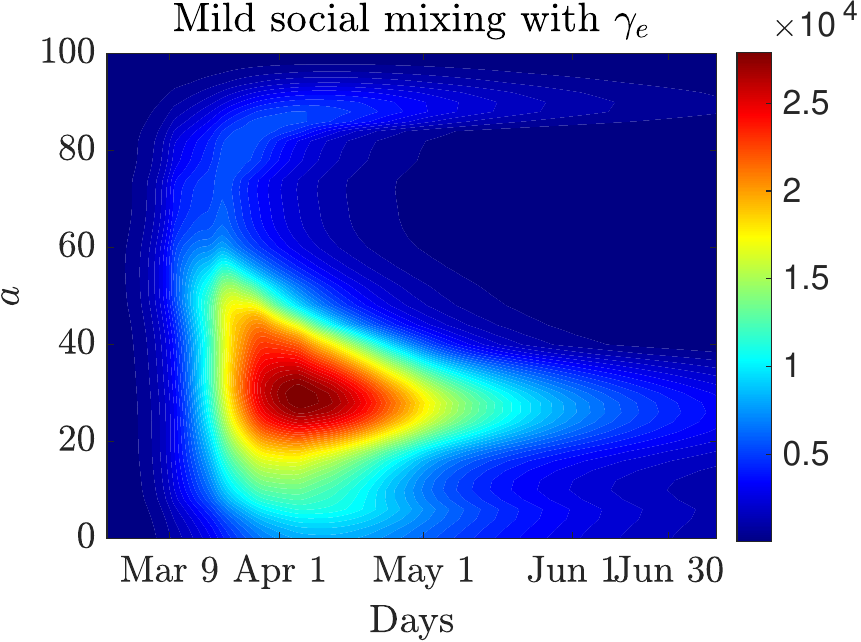}
\includegraphics[scale = 0.4]{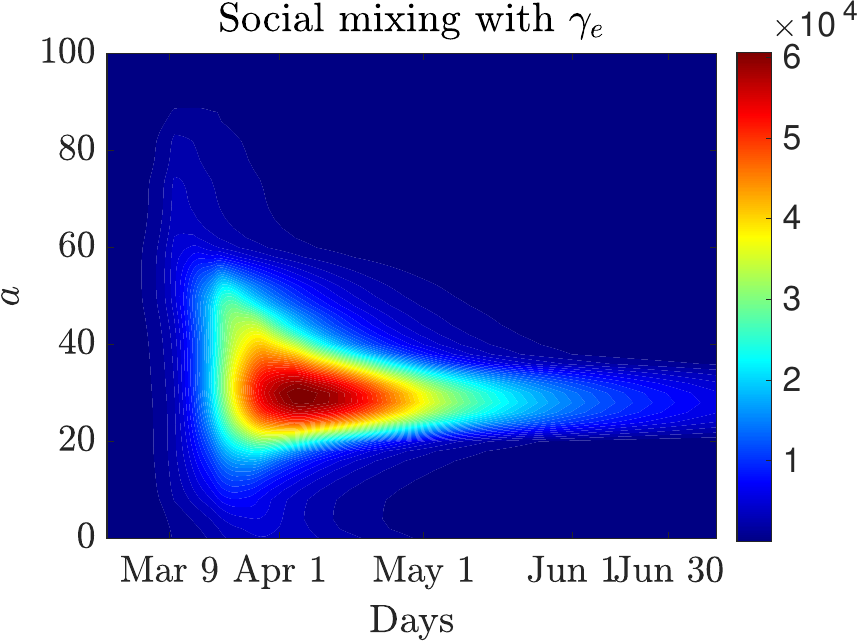}
\\
\includegraphics[scale = 0.4]{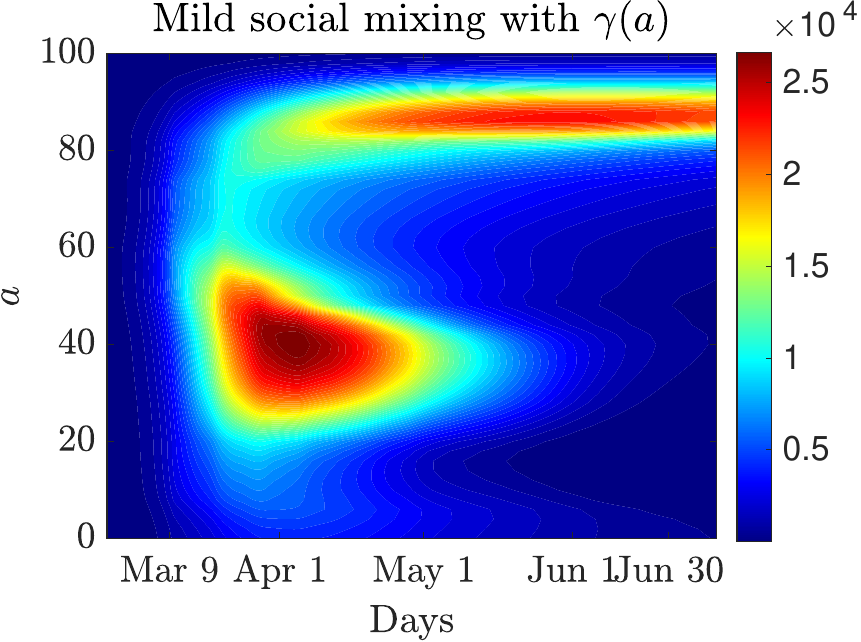}
\includegraphics[scale = 0.4]{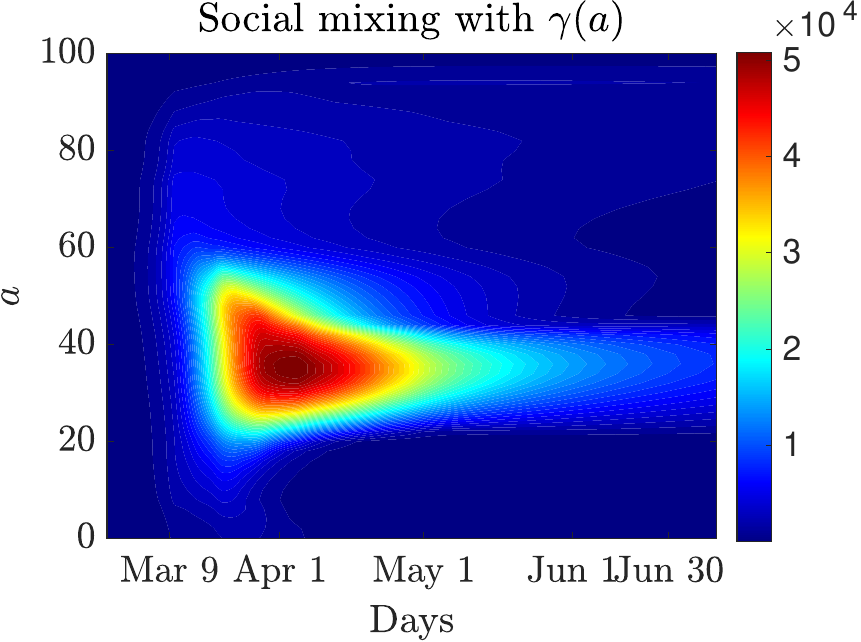}
\caption{\textbf{Test 3}. {Expected age distribution of infectious individuals for a perception function $\psi(I)=I$ with mild (left) and full (right) social mixing. In the top row $\gamma = \gamma_e$ and in the bottom row $\gamma=\gamma(a)$ defined in \eqref{eq:gamma_a}}.} 
\label{fig:gamma_a}
\end{figure}

\subsection{Test 4: Reducing the epidemic through relaxed social containment}

One of the major problems in the application of very strong containment strategies, such as the lockdown applied in Italy, is the difficulty in maintaining them over a long period, both for the economic impact and for the impact on the population from a social point of view. In order  to analyze sustainable control strategies, therefore, it is necessary to resort to models with a social structure and control methods based on specific forms of social distancing that allow the economy to restart and the population to dedicate itself, albeit in a limited way, to its pre-pandemic activities.

In accordance with the interaction function introduced in the Appendix \ref{app:data},
{ we considered the following age dependent relaxation function
\begin{equation}\label{eq:relaxP}
\Psi(a,a_*,t)=\left\lbrace 
\begin{array}{ccc}
 p_s & a,a_*\in \Lambda_P \\
 p_w & {\rm elsewhere}   
\end{array}
\right.
\end{equation}
where the interval $\LL_P$ defines the age group related to a stronger relaxation of the restrictions (in the sequel we assume $\LL_P=[25,65]$), and the parameters $0\leq p_w < p_s\leq 1$ characterize the intensity of the heterogeneity of the relaxation process  over the different age classes. Hence we relax the control parameter defined in \eqref{eq:u_z} according to 
\begin{equation}
u_{\rm relax}(a,a_*,t) = (1-\Psi(a,a_*,t)) u(a,a_*,t).
\label{eq:urelax}
\end{equation}}
{The evolution of the infection is considered within two different relaxation times, at  May 4 and at June 3, as actuated by the Italian Government during the first wave of the pandemic. We report in Table \ref{tab:par} the specific choice of the parameter $p_s$ and $p_w$ associated to different periods. Note that, these values are directly related with an increase of the disease transmission rather than to a relaxation of restrictions. In fact, it is clear that under suitable safety precautions a relaxation of restrictions may not contribute to the progress of the epidemic. 
}
	\begin{table}[t]
	\center
	\begin{tabular}{c|c|c|c|c}
           &until March 9& March 9 - May 3 & May 4 - June 2 & from June 3 \\
		\hline
		$p_s\times 100\%$  & --& 0\%& 15\%&20\%\\
		\hline
		~$p_w\times 100\%$ & --& 0\%& 5\% &10\%\\
		\hline

		\hline
	\end{tabular}
	\caption {Reduction of the feedback control \eqref{eq:urelax} over different time periods due to the relaxation of the lockdown processes  by the choice of the parameter $p_s$ and $p_w$ of the age dependent function $\Psi$ defined in \eqref{eq:relaxP}.}\label{tab:par}
\end{table}

{
In Figure \ref{fg:agedep1} (top row) we report the evolution of the age-controlled model with perception function $\psi(I)=I$, with mild social mixing and with full social mixing and homogeneous recovery rate. Each plot compare the evolution with relaxation of the containment measure and without, respectively indicated with $E[I_{\rm relax}]$ and $E[I]$. 
It is easily observed how the relaxation process increases the number of infected persons. Although the expected value remains under control, the wider confidence bands highlight the risk of a resumption of the epidemic. Of course, a further relaxation of the values in Table \ref{tab:par} will lead to a higher risk of restarting the epidemic wave.
In Figure \ref{fg:agedep1} (bottom row) we report also the evolution of the age-dependent expected number of infected individuals, both for mild social mixing and full social mixing. We remark that, as expected, the relaxation process focuses the increase of infections in the interval characterized by $\LL_P$.}

As can be seen, a gradual strategy can keep the average number of infections under control and have an outcome comparable to the fully controlled model at a potentially lower social cost. The timing and intensity of interventions, however, are crucial to prevent the restart of the epidemic wave.

\begin{figure}
	\includegraphics[scale =0.38]{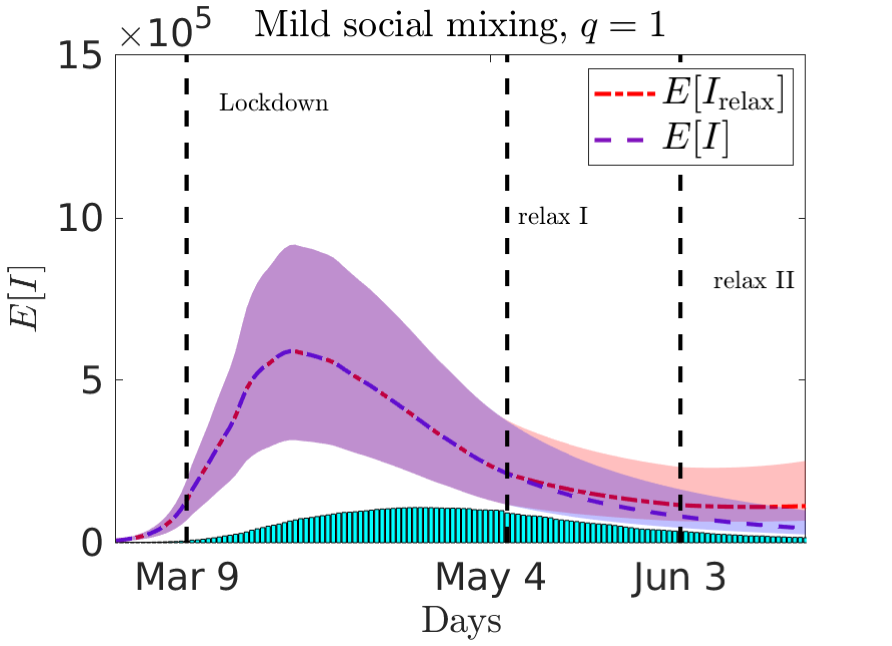}
	\includegraphics[scale =0.38]{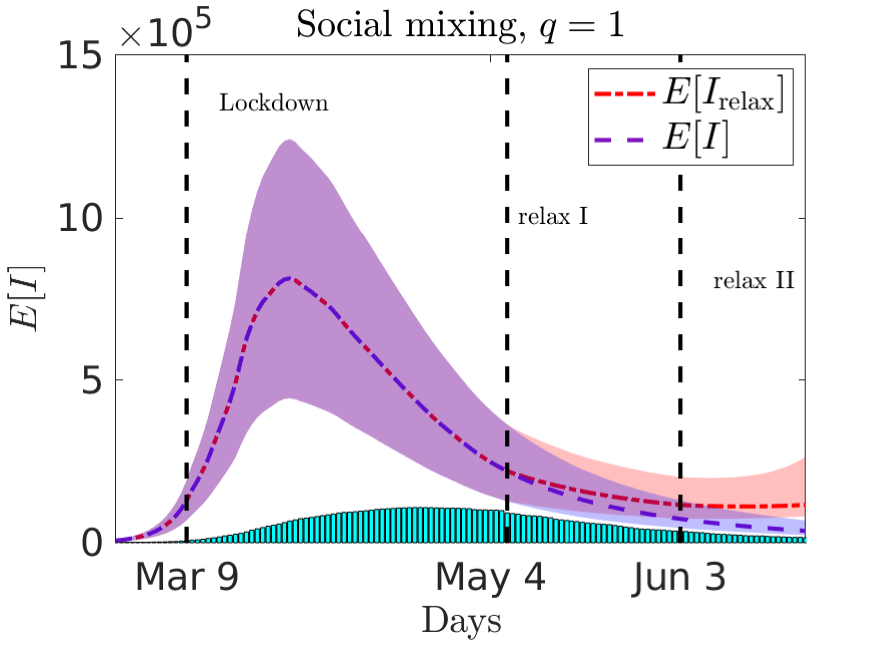} \\
	\includegraphics[scale = 0.4]{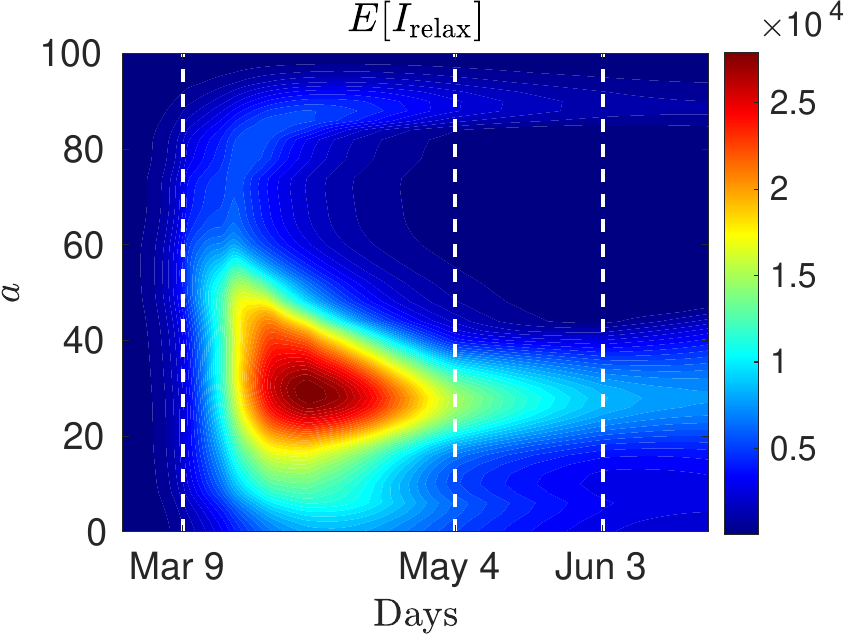}
	\includegraphics[scale = 0.4]{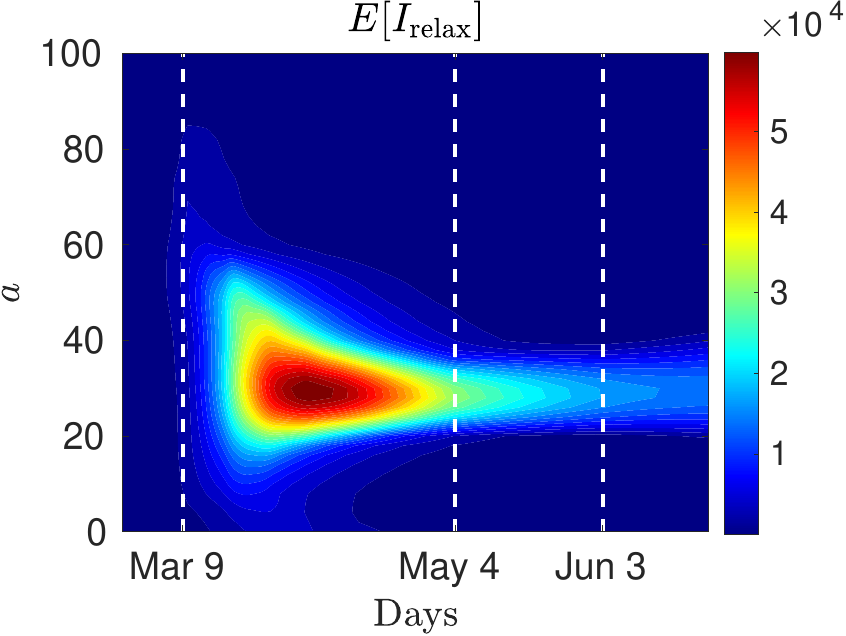}
\caption{\textbf{Test 4}. {Expected number of infected with relaxed control \eqref{eq:urelax} characterized by Table \ref{tab:par}, perception function $\psi(I)=I$, and homogeneous recovery rate, using mild social mixing (left) and with full social mixing (right). In the bottom row the corresponding expected age distribution of infectious individuals is reported.}}
\label{fg:agedep1}
\end{figure}

\section{Conclusions}
Quantifying the impact of uncertain data in the context of an epidemic emergency is essential in order to design appropriate containment measures. Such containment measures, implemented by several countries in the course of the COVID-19 epidemic, have proved effective in reducing the $R_0$ reproduction number to below or very close to one. These large-scale non-pharmaceutical interventions vary from country to country but include social distancing (banning large mass events, closing public places and advising people not to socialize outside their families), closing borders, closing schools, measures to isolate symptomatic individuals and their contacts, and the large-scale lock-down of populations with all but essential prohibited travel.

One of the main problems is the sustainability of these interventions, which until the introduction of a vaccine will have to be maintained in the field for long periods.
However, estimating the reproductive numbers of SARS-CoV-2 is a major challenge due to the high proportion of infections not detected by health care systems and differences in test application, resulting in diverse proportions of infections detected over time and between countries. Most countries initially had only the capacity to test a small proportion of suspected cases and tests were reserved for severely ill patients or high risk groups. The available data therefore offer a systematic partial overview of trends.

In this article, starting from a SIR-type compartmental model with social structure, we developed new stochastic mathematical models describing the actions of a government agency to contain infections among the population in presence of an uncertain number of infectious individuals. {More precisely, in the present model, the state of the infected is defined by a multi-dimensional time-dependent function characterized by the age, which proved essential in the description of the COVID-19 pandemic, and by a systematic uncertainty which permits to avoid additional sub-compartmentalization of the infectious population.}

These assumptions allows to derive a socially structured model that contains the control action in feedback form based on the perception of the policy maker of the spread of the disease. Subsequently, the uncertainty in the model has been approximated by expanding the state variables into orthogonal polynomials in the random space, reducing the problem to a set of deterministic equations for the distribution of the solution through
the course of the epidemic. {The resulting controlled stochastic dynamical system is then solved using this deterministic formulation, which in the case of sufficiently regular uncertainty distributions allows efficient and accurate estimations of the random parameters}. 

{The numerical simulations, performed using data from the recent COVID-19 outbreak in Italy, show, on the one hand, the model's ability to characterize the presence of the asymptomatic population trough the introduction of an uncertainty in the number of infected and in their epidemic role, and, on the other hand, in presence of control to well describe the effects of non pharmaceutical interventions aimed at flattening the infection curve. In particular, identifying some scenarios in agreement with government actions, containment measures by the population based on the resumption of certain occupational activities characterized by specific age groups and social interaction matrices were studied.} 

{Further studies in this direction will aim to consider more comprehensive epidemic models, including the effects of other clinical compartments of interest, along with generalization of the control term to different objective functions or to the case of multiple control terms for each social activity, in order to design optimal strategies to mitigate the overall epidemic impact.}

\appendix
{
\section{Feedback controlled models with additional compartmentalizations}\label{sec:gen}
 In this Appendix, we will detail how to extend the instantaneous control approach introduced in Section \ref{sect:ins} in the case of a a socially-structured SEPIAR-type compartmentalization \cite{GattoPNAS} where the control functional is given by 
\begin{equation}\label{eq:func_m}
\begin{split}
\min_{{\bf u}\in \mathcal U}J({\bf u}):=& \int_0^T\psi(I_T(t)) dt\\
&+\dfrac{1}{2}\sum_{j\in\{P,I,A\}}\int_0^T\int_{\LL\times\LL} {\nu_j(\as,\as_*,t)}|u_j(\as,\as_*,t)|^2\ d\as d\as_*  dt,
\end{split}
\end{equation}
with ${\bf u} = (u_P,u_I,u_A)$ and
subject to
\begin{equation}\label{eq:SEPIARc}
\begin{split}
\frac{d}{dt} \fS(\as,t)&= - \fS(\as,t)\sum_{j\in\{P,I,A\}}\int_{\LL} (\beta_j(\as,\as_*) -u_j(\as,\as_*,t)){\fI_j(\as_*,t)}\ d\as_* \\
\frac{d}{dt} \fE(\as,t) &=  \fS(\as,t)\sum_{j\in\{P,I,A\}}\int_{\LL} (\beta_j(\as,\as_*) -u_j(\as,\as_*,t)){\fI_j(\as_*,t)}\ d\as_* - \delta_E(\as) \fE(\as,t)\\
\frac{d}{dt} \fP(\as,t) &= \delta_E(\as) \fE(\as,t)-\delta_P(\as) \fP(\as,t)\\
\frac{d}{dt} \fId(\as,t) &=  \sigma(\as)\delta_P(\as)\fP(\as,t)-\gamma_I(\as)\fId(\as,t) \\
\frac{d}{dt} \fA(\as,t) &=  (1-\sigma(\as))\delta_P(\as)\fP(\as,t)-\gamma_A(\as)\fA(\as,t) \\
\frac{d}{dt} \fR(\as,t) &= \gamma_I(\as) \fI(\as,t)+\gamma_A(\as) \fA(\as,t).
\end{split}
\end{equation}
In the above compartmentalization $\fE, \fP, \fId$, and $\fA$ represent the number of exposed, pre-asymptomatic, infected with heavy symptoms and asymptomatic/mildly symptomatic individuals, respectively \cite{GattoPNAS}. The parameters $\beta_j$, $j \in \{P, I, A\}$ are the specific transmission rates of the three infectious classes. Exposed individuals, still not contagious, enter the presymptom stage at rate $\delta_E$ and only then become infectious. Pre-symptomatic individuals progress at rate $\delta_P$ to become symptomatic who develop severe symptoms with probability $\sigma$ or asymptomatic individuals with probability $1-\sigma$. Symptomatic individuals recover from infection at rate $\gamma_I$ whereas asymptomatic individuals leave their compartment after having recovered from infection at rate $\gamma_A$.
}

{ 
In \eqref{eq:func_m} we assume that the policy maker aims at minimizing the total number of infected
\[
I_T(t)=\int_{\Lambda} \left(\fP(\as,t)+\fId(\as,t)+ \fA(\as,t)\right)\,d\as.
\]
We highlight that other control strategies may be considered as well leading to different feedback control formulations.}

{For the socially-structured SEPIAR model the feedback controlled formulation is then obtained from the discrete approximation
\begin{equation}\label{eq:func_Im}
\min_{{\bf u}\in \mathcal U}J_h({\bf u}):= \psi(I_T(t+\h))+\dfrac{1}{2}\sum_{j\in\{P,I,A\}}\int_{\LL\times\LL} {\nu_j(\as,\as_*,t)}|u_j(\as,\as_*,t)|^2\ d\as d\as_*,
\end{equation}
subject to
\begin{eqnarray}
\nonumber
\fS(\as,t+\h) &=&  \fS(\as,t)-\h \fS(\as,t)\!\!\!\!\!\!\sum_{j\in\{P,I,A\}}\!\int_{\LL} (\beta_j(\as,\as_*) -u_j(\as,\as_*,t)){\fI_j(\as_*,t)}\ d\as_* \\
\nonumber
\fE(\as,t+\h) &=&  \fE(\as,t)+\h\fS(\as,t)\!\!\!\!\!\!\sum_{j\in\{P,I,A\}}\!\int_{\LL} (\beta_j(\as,\as_*) -u_j(\as,\as_*,t)){\fI_j(\as_*,t)}\ d\as_* - \h\delta_E(\as) \fE(\as,t)\\
\label{eq:SEPIARcc}
\fP(\as,t+\h) &=&  \fP(\as,t)+\h \delta_E(\as) \fE(\as,t+\h)-\h \delta_P(\as) \fP(\as,t)\\
\nonumber
\fId(\as,t+\h) &=& \fId(\as,t)+\h\sigma(\as)\delta_P(\as)\fP(\as,t+\h)-\h\gamma_I(\as)\fId(\as,t)\\ 
\nonumber
\fA(\as,t+\h) &=&  \fA(\as,t)+\h(1-\sigma(\as))\delta_P(\as)\fP(\as,t+h)-\h\gamma_A(\as)\fA(\as,t)
\end{eqnarray}
where the discrete equation for $\fR(\as,t)$ can be omitted since the control does not play any direct role in its evolution. The total number of infected evolved accordingly to 
\begin{equation}
I_T(t+\h) = I_T(t)+\h\int_{\LL}\left(\delta_E(\as)\fE(\as,t+\h)-\gamma_I(\as)\fId(\as,t)-\gamma(\as)\fA(\as,t)\right)\,d\as.
\label{eq:IT} 
\end{equation}
We can derive the minimizer of $J_h$ computing $\nabla_{\bf u} J_h({\bf u}) \equiv 0$ or equivalently
\[
\frac{\partial J_h(u_i)}{\partial u_j} = 0,\qquad i,j \in \{P,I,A\}.
\]
Using \eqref{eq:func_Im} we can compute
\[
\psi'(I_T(t+\h))\frac{\partial I_T(\as,t+\h)}{\partial u_j} =-\nu_j(\as,\as_*,t) u_j(\as,\as_*,t),\quad j \in \{P,I,A\},
\]
which by virtue of \eqref{eq:IT} leads to the discrete feedback controls 
\begin{equation}
u_j(a,a_*,t)=\frac{h^2\delta_E(\as)}{\nu_j(a,a_*,t)}s(a,t) i_j(\as_*,t) \psi'(I_T(t+\h)),\quad  j \in \{P,I,A\}. 
\end{equation}
In order to pass to the limit $\h\to 0$ in \eqref{eq:SEPIARcc} and obtain the feedback controlled model, we must rescale the penalization parameters with respect to the short time-horizon $\h$. With the aim of preserving the compartment of exposed,  the penalization terms should rescale as $\nu_j(a,a_*,t)=h^2 \kappa_j(a,a_*,t)$ to get the feedback controlled model \eqref{eq:SEPIARc} where the control terms are given by
\begin{equation}
u_j(a,a_*,t)=\frac{\delta_E(\as)}{\kappa_j(a,a_*,t)}s(a,t) i_j(\as_*,t) \psi'(I_T(t)),\quad  j \in \{P,I,A\}. 
\label{eq:ca2}
\end{equation}
It is interesting to observe that the control action in \eqref{eq:ca2} is inversely proportional to the incubation time, namely shorter incubation periods requires a stronger control. On the contrary, if one assumes the same rescaling adopted for the SIR model in \eqref{eq:ic}, namely $\nu_j(a,a_*,t)=h \kappa_j(a,a_*,t)$, and select a short time horizon equal to the incubation time $\h=1/\delta_E$, the exposed compartment in \eqref{eq:SEPIARcc} reduces to 
\[
\fE(\as,t+\h) =  \h\fS(\as,t)\!\!\!\!\!\!\sum_{j\in\{P,I,A\}}\!\int_{\LL} (\beta_j(\as,\as_*) -u_j(\as,\as_*,t)){\fI_j(\as_*,t)}\ d\as_*
\]
which substituted into the equation for $\fP(a,t+h)$ in the limit $h\to 0$ leads to a feedback controlled model in the form \eqref{eq:SEPIARc} without the compartment of exposed and where the pre-symptomatic satisfy
\[
\frac{d}{dt} \fP(\as,t) =  \fS(\as,t)\sum_{j\in\{P,I,A\}}\int_{\LL} (\beta_j(\as,\as_*) -u_j(\as,\as_*,t)){\fI_j(\as_*,t)}\ d\as_* - \delta_P(\as) \fP(\as,t).
\]
Finally, we emphasize that the generalization to the case with uncertainty contains no difficulty and follows the same methodology introduced in Section \ref{sect:control_UQ}. In terms of data-fitting, similar to the SIR model, one considers the feedback controlled model in absence of age-dependence and proceeds along the lines indicated in Remark \ref{rk:df}. We will omit the details for brevity.}

\section{Stochastic Galerkin approximation}\label{app:sG}
In this Appendix we give the details of the stochastic Galerkin (sG) method used to solve the feedback controlled system \eqref{eq:SIR_z}-\eqref{eq:u_z} with uncertainties.
To this aim, we consider a random vector $\z = (z_1,\dots,z_d)$ with independent components and whose distribution is $p(\z):\mathbb R^{d_z} \rightarrow \mathbb R^{d_z}_+$. The stochastic Galerkin approximation of the differential model \eqref{eq:SIR_z}-\eqref{eq:u_z} is based on stochastic orthogonal polynomials and provides a spectrally accurate solution under suitable regularity assumptions, see \cite{X}. We consider the linear space $\mathbb P_M$ of polynomials of degree up to $M$ generated by a family of polynomials $\{\Phi_{\mathbf h}(\z)\}_{|\mathbf h|=0}^M$ that are orthonormal in the space $L^2(\Omega)$ such that
\[
\mathbb E[\Phi_h(\cdot)\Phi_k(\cdot)] =\delta_{\mathbf k \mathbf h}, \qquad 0\le |\mathbf k|, |\mathbf h| \le M
\]
being  $\mathbf k =(k_1,\dots,k_d) $ a multi-index, $|\mathbf k| =k_1+\dots+k_{d_z} $ with  $\delta_{kh}$ the Kronecker delta function, and $\mathbb E[\cdot]$ the expectation with respect to $p(\z)$. The construction of the polynomial basis $\{\Phi_{\mathbf h}(\z)\}_{|\mathbf h|=0}^M$ depends on the distribution of the uncertainties and must be chosen in agreement with the Askey scheme \cite{X}. We summarize in Table \ref{tab:pol} several polynomials bases in connection with the law of a random component of $\z$. 

By assuming $s(\z,\as,t)$, $i(\z,\as,t)$ and $r(\z,\as,t)$ in $L^2(\Omega)$ we may approximate these terms through a generalized polynomial chaos expansion in the random space as follows
\be\label{eq:SIR_M}
\begin{split}
s(\z,\as,t) &\approx s^M(\z,\as,t) = \sum_{|\mathbf k|=0}^M \hat s_{\mathbf k}(\as,t) \Phi_{\mathbf k}(\z) \\
i(\z,\as,t) &\approx i^M(\z,\as,t) = \sum_{|\mathbf k|=0}^M \hat i_{\mathbf k}(\as,t) \Phi_{\mathbf k}(\z) \\
r(\z,\as,t) &\approx r^M(\z,\as,t) = \sum_{|\mathbf k|=0}^M \hat r_{\mathbf k}(\as,t) \Phi_{\mathbf k}(\z),
\end{split}
\ee
where the quantities $\hat s_{\mathbf k}$, $\hat i_{\mathbf k}$, $\hat r_{\mathbf k}$ are projections in the polynomial space
\[
\begin{split}
 \hat s_{\mathbf k}  = \mathbb E\left[ s(\cdot,\as,t) \Phi_{\mathbf k}(\cdot)\right], \quad
 \hat i_{\mathbf k}  = \mathbb E\left[ i(\cdot,\as,t) \Phi_{\mathbf k}(\cdot)\right],  \quad
 \hat r_{\mathbf k} = \mathbb E\left[ r(\cdot,\as,t) \Phi_{\mathbf k}(\cdot)\right]. 
 \end{split}
 \]

 \begin{table}[t]
\begin{center}
\begin{tabular}{l | l |  c}\label{tab:Ask}
Probability law & Expansion polynomials & Support \\ 
\hline\hline
& & \\[-.3cm]
Gaussian & Hermite    & $(-\infty,+\infty)$ \\
Uniform    & Legendre & $[a,b]$\\
Beta         & Jacobi     & $[a,b]$\\
Gamma    & Laguerre & $[0,+\infty)$\\
Poisson    & Charlier   & $\mathbb{N}$
\end{tabular}
\end{center}
\caption{The different polynomial expansions connected to the probability distribution of the random component $z_k$, $k = 1,\dots,d_{z}$.}
\label{tab:pol}
\end{table}
The sG formulation of system \eqref{eq:SIR_z} is obtained first by replacing the solutions with their stochastic polynomial expansions 
 \begin{equation}
\begin{split}
\dfrac{d}{dt} s^M(\z,\as,t)&= - s^M(\z,\as,t) \int_\LL \left(\beta(\z,\as,\as_*)-u^M(\as,\as_*)\right) i^M(\z,\as_*,t)d\as_*   \\
\dfrac{d}{dt} i^M(\z,\as,t) &= s^M(\z,\as,t) \int_\LL \left(\beta(\z,\as,\as_*)-u^M(\as,\as_*)\right)i^M(\z,\as_*,t)d\as_*  - \gamma(\z,\as) i^M(\z,\as,t) \\
\dfrac{d}{dt} r^M(\z,\as,t) &= \gamma(\z,\as) i^M(\z,\as,t),
\end{split}
\end{equation}
with
\begin{equation}
u^M(\as,\as_*) = \frac1{\kappa(\as,\as_*)}\mathcal{R}\left[s^M(\z,\as,t)i^M(\z,\as_*,t)\psi'(I^M(\z,t)) \right],
\end{equation}
\[
I^M(\z,t)=\sum_{|\mathbf k|=0}^M \hat I_{\mathbf k}(t)\Phi_{\mathbf k}(\z),\qquad  \hat I_{\mathbf k}(t) = \int_\LL i^M(\z,\as,t)d\as,
\]
and where $s^M$, $i^M$, $r^M$ are defined by \eqref{eq:SIR_M}. Then, thanks to the orthonormality of the polynomial basis of $\mathbb P_M$, multiplying by $\Phi_{\mathbf{m}}$, for all $|\mathbf m| \le M$, and taking the expectation with respect to $p(\z)$ we obtain the following coupled system of $M+1$ deterministic equations for the evolution of each projection
 \begin{equation}\displaystyle
\begin{cases}\vspace{0.5em}\displaystyle
\dfrac{d}{dt}\hat{s}_{\mathbf k}(\as,t)= - \sum_{|\mathbf m|=0}^M \textbf B_{\mathbf k \mathbf m}(\as,t) \hat{s}_{\mathbf m}(\as,t) \\ \vspace{0.5em}
\dfrac{d}{dt} \hat{i}_{\mathbf k}(\as,t) = \displaystyle \sum_{|\mathbf m|=0}^M \textbf B_{\mathbf k \mathbf m}(\as,t) \hat{s}_{\mathbf m}(\as,t) - \sum_{|\mathbf m|=0}^M \mathbf G_{\mathbf k \mathbf m} \hat{i}_{\mathbf m}(\as,t) \\
\dfrac{d}{dt} \hat{r}_{\mathbf k}(\as,t) = \displaystyle\sum_{|\mathbf m|=0}^M \mathbf G_{\mathbf k \mathbf m} \hat{i}_{\mathbf m}(\as,t)
\end{cases}
\end{equation}
where
\begin{equation}
\begin{split}
\textbf B_{\mathbf k \mathbf m} &=  \sum_{|\mathbf l|=0}^M \int_{\mathbb R^{d_z}} \left(\int_{\LL} \left(\beta(\z,\as,\as_*)-u^M(\as,\as_*)\right) \hat{i}_{\mathbf l}(\as_*,t) d\as_* \right)\Phi_{\mathbf{k}} (\z)\Phi_{\mathbf{m}}(\z)\Phi_{\mathbf{l}}(\z) p(\z)d\z\\
\mathbf G_{\mathbf k \mathbf m} &= \int_{\mathbb R^{d_z}} \gamma(\z,\as)\Phi_{\mathbf k}(\z)\Phi_{\mathbf m}(\z)p(\z)d\z.
\end{split}
\end{equation}
The above system is then integrated in time directly in the space of projections. We remark that statistical quantities of interest, such as expectation and variance of infected, can be recovered as
\[
\mathbb{E}[i^M(\cdot,a,t)]=\int_{\mathbb R^{d_z}} i^M(\z,a,t) p(\z)\,d\z = \sum_{|\mathbf k|=0}^M \hat{i}_{\mathbf k}(a,t) \mathbb{E}[\Phi_{\mathbf k}(\cdot)] =\hat{i}_{\mathbf 0}(a,t),
\]
whereas for the {variance} we get
\[
\begin{split}
{\rm Var}(i^M(\cdot,a,t)) &= \mathbb{E}[i^M(\cdot,a,t)^2]-\mathbb{E}[i^M(\cdot,a,t)]^2\\
& = \sum_{|\mathbf k|=0}^M \sum_{|\mathbf h|=0}^M \hat i_{\mathbf k} \hat i_{\mathbf h} \mathbb{E}[\Phi_{\mathbf k}(\cdot)\Phi_{\mathbf h}(\cdot)]-\hat i_{\mathbf 0}^2(a,t)=\sum_{|\mathbf k|=0}^M \hat i_{\mathbf k}^2(a,t)-\hat i_{\mathbf 0}^2(a,t).
\end{split}
\]
In all the simulations reported we used $M=10$ and a fourth order Runge-Kutta method for the time integration. 

\section{Social mixing functions}\label{app:data}

This last appendix is devoted to report the details of the social interaction functions that characterize the dynamics of social mixing. These characteristics are in fact crucial for a correct prediction of outcomes, especially in diseases transmitted by close contacts. Several large-scale studies have been designed in the last decade to determine relevant age-based models in social mixing. Without attempting to review the vast literature on this topic, we mention \cite{Betal,Metal,PCJ} and the references therein. 

The number of contacts per person generally shows considerable variability according to age, occupation, country and even day of the week, in relation to the social habits of the population. Nevertheless, some universal behaviors can be extracted which emerge as a function of specific social activities. Social mixing is highly age-related, which means that people usually tend to interact with other people of a similar age. Young people have a high rate of contact with adults aged around 30-39 and older people over 65, i.e. their parents and grandparents respectively. Contact rates are indeed very high at home and at school. On the other hand, professional mixing is weakly assortative by age and tends to be determined by uniform interactions, approximately between people from 25 and 65 years old.  

For these reasons we consider an interaction function determined by three main sub-functions that characterize the family, the school and the professional mixing.
{
Therefore, a stylized function approximating a realistic contact matrix can be written as follows
\begin{equation}
\beta_{\textrm{\rm social}}(a,a_*) =  C_\beta\sum_{j\in \mathcal A}\omega_j \beta_j(a,a_*),
\label{eq:betasocial}
\end{equation}
where the functions $\beta_j(a,a_*)$ take into account the different contact rates of people with ages $a$ and $a_*$ in relation to specific social activities of the type $\mathcal A = \{F,E,P\}$, where we identify family contacts with $F$, education and school contacts with $E$ and professional contacts with $P$, and the associated weights $\omega_F,\omega_E$ and $\omega_P$.
The particular structure of these social interaction matrices was determined empirically in \cite{Betal,PCJ}.
Here, according to these observations, we propose suitable mathematical functions that can be calibrated to reproduce empirical observations, in what follow we normalize $\beta_e=1$ and $a_{\rm max}=1$.}
\begin{figure}
\centering
\includegraphics[scale = 0.24]{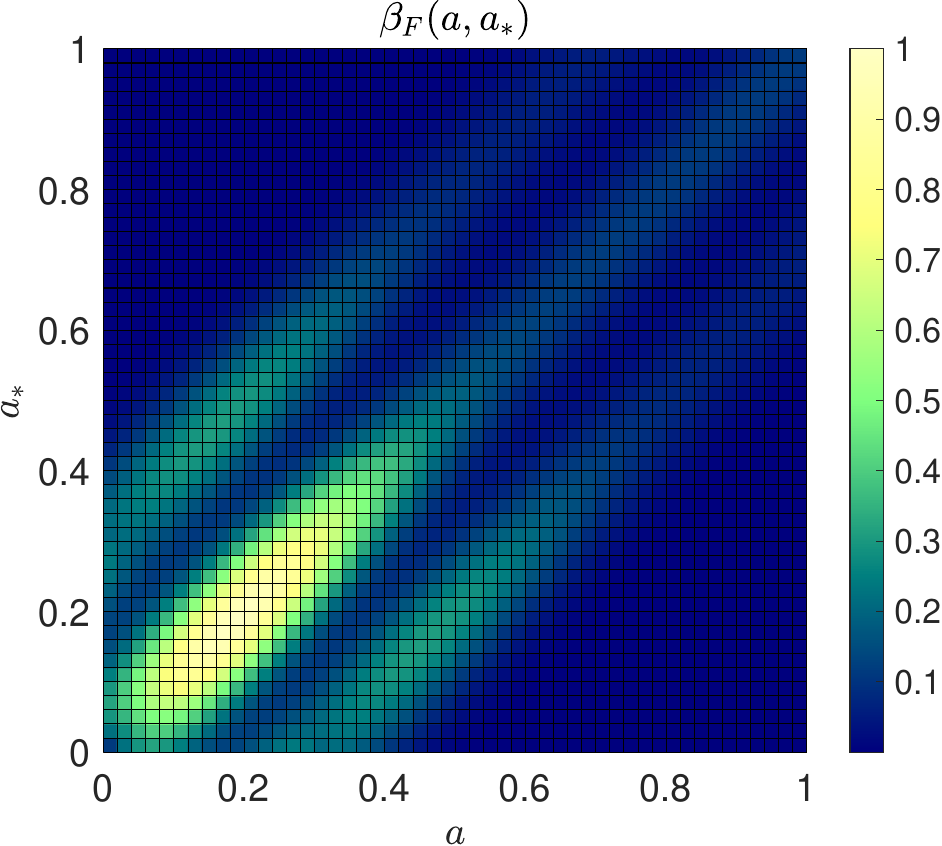}
\includegraphics[scale = 0.24]{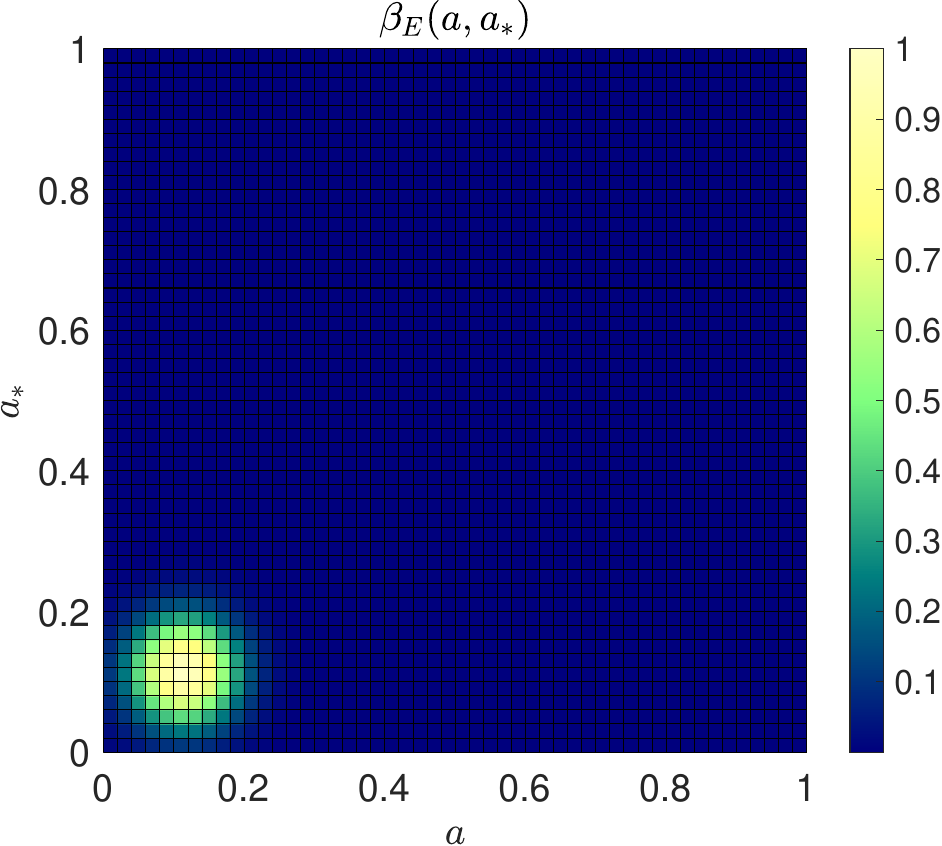}
\includegraphics[scale = 0.24]{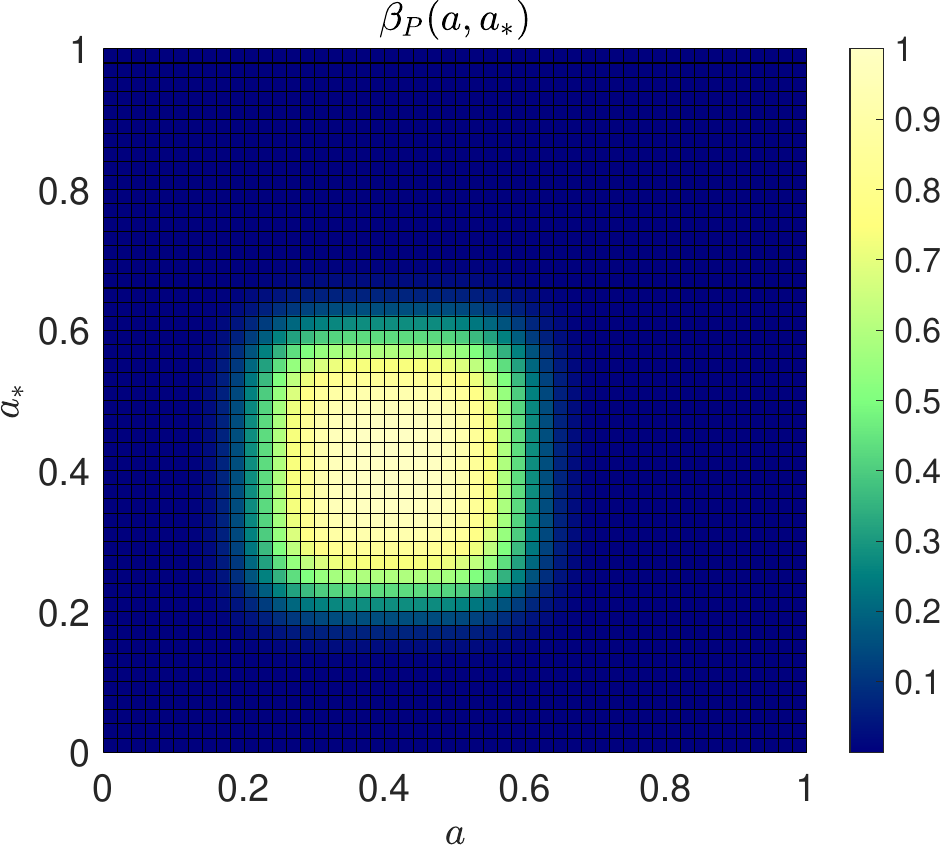}
\caption{From left to right, contour plot of the social interaction functions $\beta_F$, $\beta_E$ and $\beta_P$ taking into account the different contact rates of people with ages $a$ and $a_*$ in relation to specific social activities. The function $\beta_F$ characterizes the family contacts, $\beta_E$ the education and school contacts, and $\beta_P$ the professional contacts.}
\label{fg:matrices}
\end{figure}
{
In details, familiar contacts tend to concentrate on a three-band matrix with a peak around younger ages. This can be reproduced considering the function  
\[
\begin{split}
\tilde \beta_F(a,a_*) &= \dfrac{\lambda_{F,1}}{(1+(a-a_*)^2)^{\lambda_{F,2}}}+ \sqrt{(a^2 + a_*^2)} \exp{\left\{-\frac{1}{\sigma_F^2}\frac{a^2 + a_*^2}{(1+(a-a_*)^2)^{\lambda_{F,2}}}\right\}} . 
\end{split}
\]
Hence, we define the family interactions as
\[
\beta_F(a,a_*) =  \beta_0+\sum_{\ell = \{-\mu,0,\mu\}}  \frac{C_\ell}{2}\left[\tilde \beta_F(a+\mu,a_*)+\tilde\beta_F(a,a_*+\mu) \right],
\]
being $\mu>0$ the age shift at which family contacts occur and $C_{\pm \mu} =1/2$ and $C_0=1$. are reported in Table \ref{tab:data}
}
{
On the other hand, school and professional interactions are more age-specific and the corresponding matrices can be reproduced as follows
\[
\begin{split}
\beta_E(a,a_*) &=   \beta_0 +  \exp\left\{ - \dfrac{1}{\sigma_E^2} \left((a-\lambda_E)^2 + (a_*-\lambda_E)^2 \right) \right\} \\
\beta_P(a,a_*) &= \beta_0 + \exp\left	\{ -\dfrac{1 }{\sigma_P^2}  \left((a-\lambda_P)^4 + (a_*-\lambda_P)^4 \right) \right\},
\end{split}
\]
being $\lambda_E>0$ the average contact age at school, and $\lambda_P>0$ the average professional contact age. The coefficient $C_F,C_E,C_P>0$  measure the impact of the different contact function in the dynamics.
\,
\noindent
In Figures \ref{fg:matrices} and \ref{fg:matrice} we represent the three social interaction functions and the resulting global social interaction function $\beta_{\rm social}(a,a_*)$  assuming $\omega_F=\omega_E=\omega_P=1$. In Table \ref{tab:data} we report the choice of the parameters used in the simulations.
}
\begin{figure}\centering
\includegraphics[scale = 0.5]{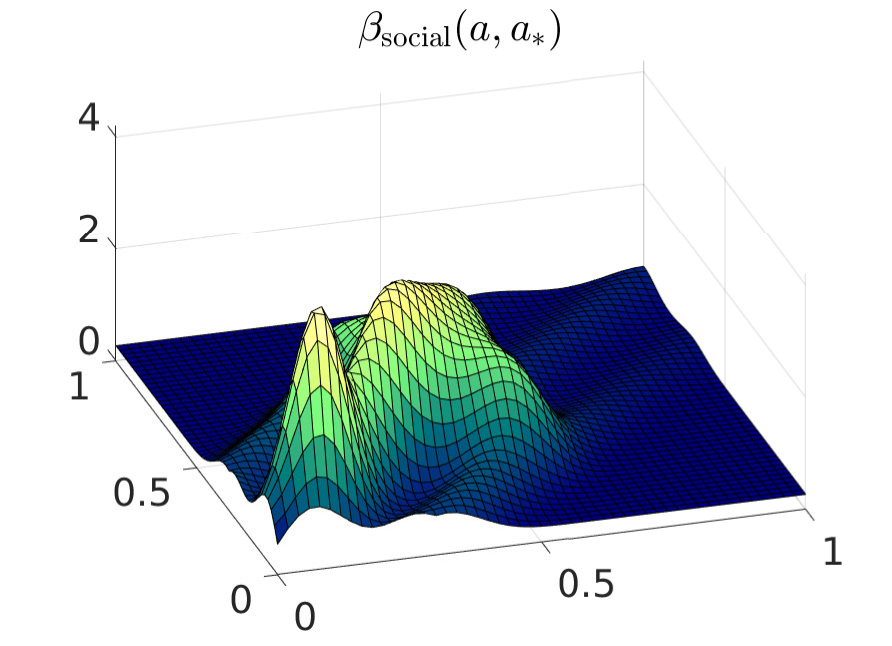}
\caption{Social interaction functions $\beta_{\rm social}=\beta_F+\beta_E+\beta_P$.}
\label{fg:matrice}
\end{figure}

\begin{table}[t]
\begin{center}
\begin{tabular}{c | c c c c c c c |c}
\hline
  Contact function   & \multicolumn{7}{c}{Parameters} \\ \hline\hline
 $\beta_F(a,a_*)$                           & $\beta_0$ & $\sigma^2_F$ & $\lambda_{F,1}$ & $\lambda_{F,2}$ & $\mu$ & $C_{\pm \mu}$ & $C_0$ & $\omega_F$\\ 
			\hline

\multirow{1}{*}{}      &  $0.04$& $0.125$ & $0.02$   & $100$  & $0.3$ & $0.5$ & 1 & 1 \\ 

\hline

	$\beta_E(a,a_*)$	          & $\beta_0$  & $\sigma_E^2$    & $\lambda_E$ & $$ & $$ & $$ & $$&$\omega_E$ \\ 
		          \hline
		        
\multirow{1}{*}{}   & $0.04$       & $0$              & 	$0.105$	 	& & & & & 0 \\

\hline

 	$\beta_P(a,a_*)$		  & $\beta_0$  & $\sigma_P^2$    & $\lambda_P$ & $$ & $$ & $$ & $$ &$\omega_P$\\ 
			  \hline
			  
\multirow{1}{*}{}   & $0.04$        & $0.00125$              & 	$0.4$	 	& & & & &0.5\\ \hline

\end{tabular}
\end{center}
\caption{The parameters defining the details of the interaction functions used in the simulations.}
\label{tab:data}
\end{table}

\begin{acknowledgements}
{This work has been written within the
activities of GNFM and GNCS groups of INdAM (National Institute of High Mathematics). G. Albi and L. Pareschi acknowledge the support of MIUR-PRIN Project 2017, No. 2017KKJP4X \emph{Innovative numerical methods for evolutionary partial differential equations and applications}. M. Zanella was partially supported by the MIUR Program (2018-2022), \emph{Dipartimenti di Eccellenza}, Department of Mathematics, University of Pavia.} 
\end{acknowledgements}

%
\section*{Conflict of interest}
The authors declare that they have no conflict of interest.


\begin{thebibliography}{}
%
%

\bibitem{Albi1} G.~Albi,  L.~Pareschi. Selective model-predictive control for flocking systems.
 \emph{Commun. Appl. Ind. Math.}, \textbf{9}(2): 4--21, 2018.
		
\bibitem{Albi2}  G.~Albi,  L.~Pareschi,  M.~Zanella. Uncertainty quantification in control problems for flocking models. \emph{Math. Probl. Eng.}, \textbf{2015}: Art. ID 850124, 2015.
		
\bibitem{Albi3}  G.~Albi,  M.~Herty, L.~Pareschi. Kinetic description of optimal control problems and applications to opinion consensus. \emph{Commun. Math. Sci.},  \textbf{13}(6): 1407--1429, 2015.
		
\bibitem{Albi4}  G.~Albi, L.~Pareschi, M.~Zanella.  Boltzmann-type control of opinion consensus through leaders.
 \emph{Philos. Trans. R. Soc. Lond. Ser. A: Math. Phys. Eng. Sci.}, \textbf{372}(2028): 20140138, 2014.

\bibitem{BGO} 
M.~Barro, A.~Guiro, D.~Ouedraogo. Optimal control of a SIR epidemic model with general incidence function and a time delays. \emph{Cubo}, \textbf{20}(2): 53--66, 2018.

\bibitem{Betal}
G.~B\'eraud, S.~Kazmercziak, P.~Beutels, D.~Levy-Bruhl, X.~Lenne, N.~Mielcarek, Y.~Yazdanpanah, P.Y.~Bo\"elle, N.~Hens, B.~Dervaux. The French connection: The first large population-based contact survey in France relevant for the spread of infectious diseases. \emph{PLoS ONE}, \textbf{10}(7): e0133203, 2015. 

\bibitem{BBSG} 
L.~Bolzoni, E.~Bonacini, C.~Soresina, M.~Groppi,
Time-optimal control strategies in SIR epidemic models,
\emph{Math. Biosci.}, \textbf{292}: 86--96, 2017.

\bibitem{BFK}
M.~Bongini, M.~Fornasier, D.~Kalise. ({U}n)conditional consensus emergence under perturbed and decentralized feedback controls. \emph{Discrete Contin. Dyn. Syst.}, \textbf{35}(9): 4071-4094, 2015.

\bibitem{Britton} 
T.~Britton, F.~Ball, P.~Trapman. The disease-induced herd immunity level for Covid-19 is substantially lower than the classical herd immunity level. \emph{preprint arXiv:2005.03085}, 2020.

\bibitem{Capaldi_etal}
A.~Capaldi, S.~Behrend, B.~Berman, J.~Smith, J.~Wright, A.~L.~Lloyd. Parameter estimation and uncertainty quantification for an epidemic model. \emph{Math. Biosci. Eng.}, \textbf{9}(3): 553--576, 2012. 

\bibitem{CS78} V.~Capasso, G.~Serio. A generalization of the Kermack-McKendrick deterministic epidemic model. {\em Math. Biosci.}, \textbf{42}(1): 43--61, 1978.

\bibitem{CCVH}
M.~A.~Capistr\'an, J.~A.~Christen, J.~X.~Velasco-Hern\'andez. Towards uncertainty quantification and inference in the stochastic SIR epidemic model, \emph{Math. Biosci.}, \textbf{240}(2): 250--259, 2012. 

\bibitem{CFPT}
M.~Caponigro, M.~Fornasier, B.~Piccoli, E.~Tr\'{e}lat. Sparse stabilization and optimal control of the {C}ucker-{S}male model. \emph{Math. Control Relat. Fields}, \textbf{3}(4): 447-466, 2013. 

\bibitem{CHALL89}
C.~Castillo-Chavez, H.~W.~Hethcote, V.~Andreasen, S.~A.~Levin, W.~M.~Liu. Epidemiological models with age structure, proportionate mixing, and cross-immunity. \emph{J. Math. Biol.}, \textbf{27}(3): 233--258, 1989.

\bibitem{CG}
R.~M.~Colombo, M.~Garavello. Optimizing vaccination strategies in an age structured SIR model. \emph{Math. Bios. Eng.}, \textbf{17}(2): 1074--1089, 2019.

\bibitem{Chen_etal}
J.~Chen et al. COVID-19 infection: the China and Italy perspectives. \emph{Cell Death and Disease}, \textbf{11}: 438, 2020. 

\bibitem{Chowell}
G.~Chowell. Fitting dynamic models to epidemic outbreaks with quantified uncertainty: A primer for parameter uncertainty, identifiability, and forecasts. \emph{Infect. Dis. Model.}, \textbf{2}(3): 379--398, 2017. 

\bibitem{DPZ}
G.~Dimarco, L.~Pareschi, M.~Zanella.  Uncertainty quantification for kinetic models in socio-economic and life sciences. In \emph{Uncertainty Quantification for Hyperbolic and Kinetic Equations}, Editors S. Jin, and L. Pareschi, SEMA SIMAI Springer Series, \textbf{14}, pp 151--191, 2017.

\bibitem{DPTZ}
G.~Dimarco, L.~Pareschi, G.~Toscani, M.~Zanella. Wealth distribution under the spread of infectious diseases.
\emph{Phys. Rev. E} \textbf{102}, 022303, 2020

\bibitem{DPT}
B.~D\"uring, L.~Pareschi,  G.~Toscani. Kinetic models for optimal control of wealth inequalities.  
\emph{Eur. Phys. J. B},  \textbf{91}, 265, 2018.


\bibitem{GammaAge1}
{C.~Faes, S.~Abrams,~D.~Van Beckhoven, G.~Meyfroidt, E.~Vlieghe, N.~Hens, and Belgian Collaborative Group on COVID-19 Hospital Surveillance},
\newblock
Time between Symptom Onset, Hospitalisation and Recovery or Death: Statistical Analysis of Belgian COVID-19 Patients.
\emph{International Journal of Environmental Research and Public Health},
\textbf{17}(20): {7560}, {2020}.


\bibitem{IC} 
S.~Flaxman et al. Estimating the number of infections and the impact of non-pharmaceutical interventions on COVID-19 in 11 European countries, \emph{Report 13. Imperial College COVID-19 Response Team}, 30 March 2020.

\bibitem{FP} 
A.~Franceschetti, A.~Pugliese. Threshold behaviour of a SIR epidemic model with age structure and immigration. \emph{J. Math. Biol.}, \textbf{57}(1): 1--27, 2008.

\bibitem{GattoPNAS}
M.~Gatto, E.~Bertuzzo, L.~Mari, S.~Miccoli, L.~Carraro, R.~Casagrandi, A.~Rinaldo. Spread and dynamics of the COVID-19 epidemic in Italy: Effect of emergency containment measures. \emph{PNAS}, \textbf{117}(19): 10484--10491, 2020. 

\bibitem{Giordano}
G.~Giordano, F.~Blanchini, R.~Bruno, P.~Colaneri, A.~Di Filippo, A.~Di Matteo, M.~Colaneri. Modelling the COVID-19 epidemic and implementation of population-wide interventions in Italy. \emph{Nat Med}, \textbf{26}, 855--860, 2020.

\bibitem{GFMDC12} 
J.~Glasser, Z.~Feng, A.~Moylan, S.~Del Valle, C.~Castillo-Chavez.
Mixing in age-structured population models of infectious diseases. \emph{Math. Bios.}, \textbf{235}(1): 1--7, 2012.
		
\bibitem{H96}{H.~W.~Hethcote},
\newblock {\it Modeling heterogeneous mixing in infectious disease dynamics, in Models for Infectious Human Diseases.}
\newblock V. Isham and G. F. H. Medley, eds., Cambridge University Press, Cambridge, UK, 215--238, 1996.
	
\bibitem{H00}{H.~W.~Hethcote},
The mathematics of infectious diseases.
\emph{SIAM Rev.}, \textbf{42}(4): 599--653, 2000.

\bibitem{Lancet} A.~A.~Sayampanathan,
C.~S.~Heng, P.~H.~Pin, J.~Pang, T.~Y.~Leong, V.~J.~Lee, Infectivity of asymptomatic versus symptomatic COVID-19, \emph{The Lancet}, \textbf{397}, 93--94, 2021.

\bibitem{IMP}
M.~Iannelli, F.~A.~Milner, A.~Pugliese. Analytical and numerical results for the age-structured S-I-S epidemic model with mixed inter-intracohort transmission. \emph{SIAM J. Math. Anal.}, \textbf{23}(3) 662--688, 1992.

\bibitem{JRGL}
K.~Jagodnik, F.~Ray, F.~M.~Giorgi, A.~Lachmann. Correcting under-reported COVID-19 case numbers: estimating the true scale of the pandemic. Preprint \texttt{medRvix:2020.03.14.20036178}.

\bibitem{JP}
S.~Jin, L.~Pareschi. \emph{Uncertainty Quantification for Hyperbolic and Kinetic Equations}, SEMA-SIMAI Springer Series, \textbf{14}, 2017. 

\bibitem{KMK}
W.O.~Kermack, A.G.~McKendrick. 
A Contribution to the Mathematical Theory of Epidemics.
\emph{Proc. Roy. Soc. Lond. A}, \textbf{115}: 700--721, 1927.

\bibitem{LavezzoCrisanti}
E.~Lavezzo et al. Suppression of a SARS-CoV-2 outbreak in the Italian municipality of Vo'. \emph{Nature}, \textbf{584}: 425--429, 2020.

\bibitem{SCCC10}
S.~Lee, G.~Chowell, C.~Castillo-ChÃ¡vez. Optimal control for pandemic influenza: the role of limited antiviral treatment and isolation. \emph{J. Theor. Biol.}, \textbf{265}(2): 136--150, 2010.
			
\bibitem{LGC12}
S.~Lee, M.~Golinski, G.~Chowell.
Modeling optimal age-specific vaccination strategies against pandemic influenza.
\emph{Bull. Math. Biol.}, \textbf{74}(4): 958-980, 2012.
    
\bibitem{LML10}
F.~Lin, K.~Muthuraman, M.~Lawley. An optimal control theory approach to non-pharmaceutical interventions. 
\emph{BMC Infect. Dis.}, \textbf{10}(1), 32, 2010.
		

\bibitem{Liu} Y.~Liu, A.~A.~Gayle, A.~Wilder-Smith, J.~Rockl\"ov. The reproductive number of COVID-19 is higher compared to SARS coronavirus. \emph{J. Travel Med.}, \textbf{27}(2): 1--4, 2020.


\bibitem{Lunelli} A.~Lunelli, A.~Pugliese, C.~Rizzo. Epidemic patch models applied to pandemic influenza: Contact matrix, stochasticity, robustness of predictions. \emph{Mathematical Biosciences}, \textbf{220}: 24--33, (2009). 


\bibitem{MKZC}
K.~Mizumoto, K.~Kagaya, A.~Zarebski, G.~Chowell.
Estimating the asymptomatic proportion of coronavirus disease 2019 (COVID-19) cases on board the Diamond Princess cruise ship, Yokohama, Japan, 2020.
\emph{Euro. Surveill.}, \textbf{25}(10): 2000180, 2020. 

\bibitem{MRPL} 
D.~H.~Morris, F.~W.~Rossine, J.~B.~Plotkin, S.~A.~Levin. 
Optimal, near-optimal, and robust epidemic control.
Preprint arXiv:2004.02209, 2020.

\bibitem{Metal}
J.~Mossong, N.~Hens, M.~Jit, P.~Beutels, K.~Auranen, R.~Mikolajczyk, M.~Massati, S.~Salmaso, G.~Scalia Tomba, J.~Wallinga, J.~Heijne, M.~Sadkowska-Todys, M.~Rosinska, W.~J.~Edmunds. Social contacts and mixing patterns relevant to the spread of infectious diseases. \emph{PLoS Med.}, \textbf{5}(3): e75, 2008. 

\bibitem{Paradisi}
M.~Paradisi, G.~Rinaldi. An Empirical Estimate of the Infection Fatality Rate of COVID-19 from the First Italian Outbreak (4/18/2020). Available at SSRN: https://ssrn.com/abstract=3582811 or http://dx.doi.org/10.2139/ssrn.3582811

\bibitem{PareschiUQ}
L.~Pareschi. An introduction to uncertainty quantification for kinetic equations and related problems. In \emph{Trails in Kinetic Theory: Foundational Aspects and Numerical Methods}, SEMA-SIMAI Springer Series, \textbf{25}, Eds. G. Albi, S. Merino-Aceituno, A. Nota, M. Zanella, 2021. 

\bibitem{PCJ}
K.~Prem, A.~R.~Cook, M.~Jit. Projecting social contact matrices in 152 countries using contact surveys and demographic data. \emph{PLoS ONE}, \textbf{13}(9): e1005697, 2017. 

\bibitem{Protezione}
Presidenza del Consiglio dei Ministri, Dipartimento della Protezione Civile. GitHub: COVID-19 Italia - Monitoraggio situazione, \texttt{https://github.com/pcm-dpc/COVID-19}, 2020. 

\bibitem{RR}
A.~Remuzzi, G.~Remuzzi. COVID-19 and Italy: what next? \emph{Lancet}, \textbf{395}:1225--1228, 2020. 

\bibitem{Roberts}
M.~G.~Roberts. Epidemic models with uncertainty in the reproduction. \emph{J. Math. Biol.}, \textbf{66}: 1463--1474, 2013.

\bibitem{Russell_etal}
T.W Russell, J.~Hellewell, C.~I.~Jarvis, K.~van Zandvoort, S.~Abbott, R.~Ratnayake, S.~Flasche, R.M.~Eggo, W.J.~Edmunds, A.J.~Kucharski. Estimating the infection and case fatality ratio for coronavirus disease (COVID-19) using age-adjusted data from the outbreak on the Diamond Princess cruise ship, \emph{Euro Surveill}, \textbf{25}(12):pii=2000256, 2020.

\bibitem{GammaAge3}
R.~Verity et al., 
Estimates of the severity of coronavirus disease 2019: a model-based analysis,
   \emph{The Lancet Infectious Diseases}, {\bf 20}(6):669--677, 2020.


\bibitem{GammaAge2} I. Voinsky, G.~Baristaite, D.~Gurwitz,
{Effects of age and sex on recovery from COVID-19: Analysis of 5769 Israeli patients},
\emph{Journal of Infection}, \textbf{81}(2):102--103, 2020.



\bibitem{GammaAge}
S.~Wang , F.~Zhong , W.~Bao , Y.~Li , L.~Liu , H.~Wang , Y.~He. 
Age-dependent risks of Incidence and Mortality of COVID- 19 in Hubei Province and Other Parts of China Hongdou. \emph{Frontiers in Medicine}, \textbf{7}, 190, 2020.

\bibitem{X}
D.~Xiu. \emph{Numerical Methods for Stochastic Computations: A Spectral Methods Approach}, Princeton University Press, 2010. 

\bibitem{Zhang_etal}
S.~Zhang, M.~Diao, W.~Yu, L.~Pei, Z.~Lin, D.~Chen. 
Estimation of the reproductive number of novel coronavirus (COVID-19) and the probable outbreak size on the Diamond Princess cruise ship: A data-driven analysis.
\emph{International Journal of Infectious Diseases}, \textbf{93}: 201--204, 2020.


\end{thebibliography}


\end{document}